\documentclass[onecolumn, a4paper, accepted=2025-07-11]{quantumarticle}
\pdfoutput=1

\usepackage[utf8]{inputenc}
\usepackage{amssymb, amsthm}
\usepackage{cite}
\usepackage{hyperref}
\usepackage{authblk}
\usepackage{graphicx}
\usepackage[labelformat=simple]{subcaption}

\usepackage[ruled, vlined, linesnumbered]{algorithm2e}
\usepackage{booktabs}
\usepackage{makecell}
\usepackage[export]{adjustbox}
\usepackage{tikz}
\usetikzlibrary{quantikz}

\usepackage{xcolor}

\newcommand{\bs}[1]{\boldsymbol #1}

\newtheorem{lemma}{Lemma}
\newtheorem{theorem}{Theorem}
\newtheorem{problem}{Problem}

\DeclareMathOperator*{\argmax}{argmax}
\DeclarePairedDelimiter\set\{\}

\title{Lower T-count with faster algorithms}

\author[1,2]{Vivien Vandaele}
\affil[1]{Eviden Quantum Lab, Les Clayes-sous-Bois, France}
\affil[2]{Université de Lorraine, CNRS, Inria, LORIA, F-54000 Nancy, France}

\date{}

\begin{document}
\maketitle

\begin{abstract}
Among the cost metrics characterizing a quantum circuit, the $T$-count stands out as one of the most crucial as its minimization is particularly important in various areas of quantum computation such as fault-tolerant quantum computing and quantum circuit simulation.
In this work, we contribute to the $T$-count reduction problem by proposing efficient $T$-count optimizers with low execution times.
In particular, we greatly improve the complexity of \texttt{TODD}, an algorithm currently providing the best $T$-count reduction on various quantum circuits.
We also propose some modifications to the algorithm which are leading to a significantly lower number of $T$ gates.
In addition, we propose another algorithm which has an even lower complexity and that achieves a better or equal $T$-count than the state of the art on most quantum circuits evaluated.
We also prove that the number of $T$ gates in the circuit obtained after executing our algorithms on a Hadamard-free circuit composed of $n$ qubits is upper bounded by $n(n + 1)/2 + 1$, which improves on the worst-case $T$-count of existing optimization algorithms.
From this we derive an upper bound of $(n + 1)(n + 2h)/2 + 1$ for the number of $T$ gates in a Clifford$+T$ circuit where $h$ is the number of internal Hadamard gates in the circuit, i.e.\ the number of Hadamard gates lying between the first and the last $T$ gate of the circuit.
\end{abstract}

\section{Introduction}

One of the main tasks of a quantum compiler is to minimize the resources needed to execute a given quantum algorithm.
This step is of considerable importance in the compilation stack as it makes quantum computation more practical and efficient.
To complete this task effectively and decide which optimization to perform, we must first identify the most expensive operations hindering our way towards fast and functional quantum computation.
In this regard the $T$ gate is often targeted as it cannot be trivially implemented, for instance via transversal operations, in a fault-tolerant way in most quantum error correcting codes as opposed to Clifford gates.
It implies that implementing a $T$ gate is generally much more costly than performing a Clifford operation~\cite{raussendorf2007topological, fowler2009high, beverland2021cost}.
Also, numerous quantum error correcting codes can be used to perform universal fault-tolerant quantum computation with the Clifford$+T$ gate set.
In such a setup, the depth of a quantum circuit, and so the time required to execute it, is generally directly linked to the $T$-depth of the circuit~\cite{fowler2012time}.
For this reason, much work has been put into the minimization of the $T$-depth in Clifford$+T$ circuits~\cite{amy2013meet, selinger2013quantum, amy2014polynomial, abdessaied2014quantum, niemann2019t, gheorghiu2022quasi}.
But if the number of qubits at disposal is limited, then the depth of the circuit also depends on the number of $T$ gates within it.
In addition, the depth of a circuit dictates the minimum number of physical qubits needed to encode the logical qubits utilized to perform the fault-tolerant computation.
The coherence time of the logical qubits must be greater than the time required to execute the whole circuit, and so the required amount of physical qubits per logical qubits increases as the depth of the circuit increases.
We can discern a feedback loop here: if the depth of the circuit is diminished then less physical qubits are required to encode the logical qubits, which frees up qubits that can be used to further lower the depth of the circuit.
The optimization of the $T$-count can intervene at multiple stages of this process.
Firstly, a lower $T$-count can induce a lower $T$-depth and so a lower circuit depth~\cite{abdessaied2014quantum}.
Also, the number of qubits required to implement the $T$ gates can depend on the $T$-count.
This is for example the case when the $T$ gates are implemented via magic state distillation~\cite{bravyi2005universal}.
A lower number of $T$ gates can thus lower the number of physical qubits required to implement the circuit.

The optimization of the $T$-count also have important applications outside of fault-tolerant quantum computation.
Numerous quantum compilers designed for NISQ devices are incorporating a step consisting in reducing the number of $T$ gates~\cite{amy2020staq, sivarajah2020t, martiel2022architecture}.
It has been demonstrated by these compilers that reducing the $T$-count can lead to shorter circuits and can help with the minimization of other gates such as the $\mathrm{CNOT}$ gate.

Besides compilation, $T$-count minimization also plays an important role in quantum circuits simulation.
As stated by the Gottesman-Knill theorem~\cite{gottesman1998heisenberg}, circuits composed of Clifford gates and Pauli measurements can be efficiently simulated by a classical computer.
Extending the gate set by adding the $T$ gate allows the simulation of universal quantum circuits at the cost of a significant increase in computational time as no algorithm is currently known to efficiently simulate these circuits.
That is why many simulation techniques have a runtime that scales exponentially with respect to the number of $T$ gates~\cite{bravyi2016improved, bravyi2019simulation, qassim2021improved, kissinger2022simulating, kissinger2022classical}.
Minimizing the number of $T$ gates is then essential to exploit the full potential of these simulators and to push back the frontier of non-simulable quantum circuits.

\paragraph{State of the art and contributions.}
Any unitary gate can be implemented by the Clifford$+T$ gate set up to an arbitrary precision.
Therefore, a compilation problem of primordial importance is to find an approximation, over the Clifford$+T$ gate set and that uses a little amount of $T$ gates, of a given unitary operator to an accuracy within $\epsilon > 0$ .
A fundamental solution to this problem, and which can be applied to any finite universal gate set, was given by the Solovay-Kitaev theorem~\cite{kitaev2002classical, dawson2006solovay}.
Other approaches designed for the Clifford$+T$ gate set were then developed to obtain better approximations~\cite{selinger2015efficient, kliuchnikov2013asymptotically, ross2016optimal}.
Further improvements have then been made by introducing measurements~\cite{bocharov2015efficient} and by using a probabilistic mixture of unitaries~\cite{hastings2016turning, campbell2017shorter}.
Recently, it has been shown that these two methods can be combined with a novel approach to achieve better results~\cite{kliuchnikov2022shorter}.
Another important synthesis problem concerns the set of unitary gates which can be exactly implemented over the Clifford+$T$ gate set.
Given one of these unitary gates, the problem then consists in finding an exact Clifford$+T$ implementation of it using a minimal number of $T$ gates.
For this problem, an optimal and efficient algorithm is known for the case of single-qubit unitaries~\cite{kliuchnikov2013fast}.

Once a Clifford$+T$ implementation of a unitary gate has been found, whether through approximate or exact synthesis, some quantum circuit optimization methods can then be applied to reduce the number of $T$ gates in the circuit.
The algorithms developed for this purpose and achieving the best $T$-count reductions are foremostly designed for the restricted class of $\{\mathrm{CNOT}, S, T\}$ quantum circuits.
The problem of $T$-count optimization for this class of circuits has been well defined by showing its equivalence with the problem of decoding Reed-Muller codes~\cite{amy2019t}.
In particular, it was demonstrated that the codewords of the punctured Reed-Muller code of length $2^n-1$ and order $n-4$ are generating the complete set of identities that can be used to optimize the number of $T$ gates in $\{\mathrm{CNOT}, S, T\}$ circuits.
Reducing the number of $T$ gates can then be done by finding relevant identities in this large set.
For example it has been shown that a particular subset of identities, called spider nest identities, can be efficiently exploited to reduce the number of $T$ gates~\cite{de2019techniques, de2020fast, munson2019and, domitrz2022}.
An effective way to find relevant identities that can be applied to reduce the number of $T$ gates was given by the \texttt{TODD} algorithm~\cite{heyfron2018efficient}.
However, an important drawback of the \texttt{TODD} algorithm is its complexity of $\mathcal{O}(n^3 m^5)$ where $n$ is the number of qubits and $m$ is the number of $T$ gates in the initial circuit, which makes it impractical for circuits of large size.
In Section~\ref{sec:beyond_merging}, we show how the complexity of the \texttt{TODD} algorithm can be reduced to $\mathcal{O}(n^4m^3)$.
In addition, we propose some modifications to the \texttt{TODD} algorithm which are resulting in a significantly improved reduction in the number of $T$ gates.
In the same section we propose another algorithm which has an even lower complexity of $\mathcal{O}(n^2m^3)$ and that achieves better results than the original \texttt{TODD} algorithm on most quantum circuits evaluated.
We also prove that the algorithms presented in Section~\ref{sec:beyond_merging} are producing Hadamard-free circuits in which the $T$-count is upper bounded by $(n^2 + n)/2 + 1$ where $n$ is the number of qubits.
Benchmarks are provided in Section~\ref{sec:bench} to evaluate the performances, in terms of $T$-count and execution time, of our algorithms on a library of reversible logic circuits and on large-scale quantum circuits.
In Section~\ref{sec:higher_orders}, we extend our results for minimizing the number of $R_Z(\pi/2^d)$ gates, where $d$ is a non-negative integer.
Finally, in Section~\ref{sec:k_upper_bound}, we demonstrate an upper bound for the number of $R_Z(\pi/2^d)$ gates in a Clifford$+\{R_Z(\pi/2^d), R_Z(2\pi/2^d)\}$ circuit.
For Clifford$+T$ circuits we obtain an upper bound of $(n + 1)(n + 2h)/2 + 1$ for the number of $T$ gates, which can be satisfied in polynomial time and without any ancillary qubit, and where $h$ is the number of internal Hadamard gates in the circuit.

\section{Preliminaries}\label{sec:preliminaries}

\subsection{$T$-count optimization in Hadamard-free circuits}\label{sub:t_count_hadamard_free}

We define the set of Pauli operators $\mathcal{P}_n$ as the set composed of all the tensor products of $n$ Pauli matrices, which are defined as follows:
$$
I = \begin{pmatrix}
1 & 0\\
0 & 1
\end{pmatrix},\quad
X = \begin{pmatrix}
0 & 1\\
1 & 0
\end{pmatrix},\quad
Y = \begin{pmatrix}
0 & -i\\
i & 0
\end{pmatrix},\quad
Z = \begin{pmatrix}
1 & 0\\
0 & -1
\end{pmatrix},
$$
with a multiplicative factor of $\pm 1$.
A Pauli rotation $R_P(\theta)$ is defined as follows:
$$R_P(\theta) = \exp(-i\theta P/2) = \cos(\theta/2)I - i\sin(\theta/2)P$$
for a Pauli operator $P \in \mathcal{P}_n$ and an angle $\theta \in \mathbb{R}$.
For instance, the $T$ gate is defined as a $\pi/4$ Pauli $Z$ rotation: 
$$T = R_Z(\pi/4).$$
The Clifford group, denoted $\mathcal{C}_n$, is generated by the set of $\pi/2$ Pauli rotations acting on $n$ qubits
$$\{R_P(\pi/2) \mid P \in \mathcal{P}_n\}.$$
A Pauli operator $P \in \mathcal{P}_n$ conjugated by a Clifford gate $U \in \mathcal{C}_n$ is always equal to another Pauli operator $P' \in \mathcal{P}_n$, i.e.\ $U^\dag P U = P'$.
This fact also holds when a Pauli rotation is conjugated by $U \in \mathcal{C}_n$: 
\begin{equation}
    U^\dag R_P(\theta) U = R_{U^\dag P U }(\theta) = R_{P'}(\theta).
\end{equation}
That is why the operation performed by a Clifford$+T$ circuit acting on $n$ qubits, represented by a unitary gate $U$, can always be described by a sequence of $\pi/4$ Pauli rotations and a final Clifford operator $C \in \mathcal{C}_n$~\cite{gosset2014algorithm}:
\begin{equation}\label{eq:sequence_prod}
    U = e^{i \phi} C \left( \prod_{i=1}^{m} R_{P_i}(\pi/4)\right)
\end{equation}
where $m$ is the number of $T$ gates in the circuit and $P_i \in \mathcal{P}_n \setminus \{\pm I^{\otimes n}\}$.

Multiple algorithms to reduce the number of $T$ gates are foremostly designed for circuits composed of $\{\text{CNOT}, S, T\}$ gates.
In order to make use of these algorithms for Clifford$+T$ circuits, it is necessary to circumvent the Hadamard gates in the input circuit since these algorithms cannot be directly executed on them.
This can be done using one the two following methods.
The first method consists in dividing the circuit into Hadamard-free subcircuits and Clifford subcircuits containing Hadamard gates.
The number of $T$ gates in the Hadamard-free subcircuits can then be optimized using these algorithms.
There exists multiple strategies that can be employed to create such a partition of the circuit.
It is generally preferred to regroup the $T$ gates in the least number of Hadamard-free subcircuits as possible to take advantage of the fact that the number of $T$ gates in an Hadamard-free circuit can be upper bounded by $\mathcal{O}(n^2)$, where $n$ is the number of qubits.
One approach to partition the circuit is to describe the operation performed by the circuit by a sequence of Pauli rotations and a final Clifford operator, as in Equation~\ref{eq:sequence_prod}, and to then reorder the Pauli rotations in the sequence by forming groups of mutually commuting Pauli rotations.
This can for example be done by using the procedure whose pseudo-code is given in Algorithm~\ref{alg:grouping}.
This algorithm has a complexity of $\mathcal{O}(nm^2)$ since checking whether or not two Pauli rotations commute takes $\mathcal{O}(n)$ operations and such commutativity checks are done at most $m$ times at each iteration of the loop.
This way of grouping the Pauli rotations is not new, an equivalent algorithm (but which has a worst-case complexity of $\mathcal{O}(nm^3)$) was given in Reference~\cite{litinski2019game}.
By using the layers of Pauli rotations produced by Algorithm~\ref{alg:grouping}, Equation~\ref{eq:sequence_prod} can then be rewritten as follows:
\begin{equation}\label{eq:sequence_prod_2}
    U = e^{i \phi} C \left( \prod_{i=1}^{\lvert L \rvert} \prod_{R_P \in L_i}  R_P(\pi/4)\right)
\end{equation}
where $\lvert L \rvert$ denotes the number of layers in $L$.

\begin{algorithm}[t]
    \caption{Grouping of Pauli rotations}
    \label{alg:grouping}
	\SetAlgoLined
	\SetArgSty{textnormal}
	\SetKwInput{KwInput}{Input}
	\SetKwInput{KwOutput}{Output}
    \KwInput{A sequence $R_{P_1},\ldots,R_{P_m}$ of Pauli rotations.}
    \KwOutput{Equivalent sequence as layers of mutually commuting Pauli rotations.}
    $L \leftarrow$ list of empty sets \\
    \For{$i \gets 1$ \KwTo $m$}{
        $j \leftarrow \max\left(\{j \mid R_{P_i} \text{ anticommutes with one element in } L_j \} \cup \{0\}\right)$\\
        $L_{j+1} \leftarrow L_{j+1} \cup \{R_{P_i}\}$ \\
    }
    \Return $L$ \\
\end{algorithm}

A Pauli operator $P$ and a Pauli rotation $R_P(\theta)$ are diagonal if and only if $P$ is a tensor product of the matrices $I$ and $Z$, up to a multiplicative factor of $\pm 1$.
Such a Pauli rotation can be implemented without using any Hadamard gate.
Because the Pauli rotations in each layer $L_i$ of Equation~\ref{eq:sequence_prod_2} are mutually commuting, we can find a Clifford operator $C_i \in \mathcal{C}_n$ such that $C_i^\dag R_P(\pi/4) C_i$ is diagonal for all $R_P \in L_i$.
We say that $C_i$ simultaneously diagonalize the Pauli rotations in $L_i$.
A circuit implementing the Clifford operator $C_i$ can be found with a complexity of $\mathcal{O}(n^2m)$, where $n$ is the number of qubits and $m$ is the number of Pauli rotations in $L_i$~\cite{van2020circuit}.
By performing a simultaneous diagonalization for each layer of Pauli rotations, Equation~\ref{eq:sequence_prod_2} can then be rewritten as follows:
\begin{equation}\label{eq:sequence_prod_3}
\begin{aligned}
    U &= e^{i \phi} C \left( \prod_{i=\lvert L \rvert}^{1} C_i C_i^\dag \left(\prod_{R_P \in L_i}  R_P(\pi/4)\right) C_i C_i^\dag \right) \\
    &= e^{i \phi} C \left( \prod_{i=\lvert L \rvert}^{1} C_i \left(\prod_{R_P \in L_i}  R_{C_i^\dag P C_i}(\pi/4)\right) C_i^\dag \right)
\end{aligned}
\end{equation}
where $C_i \in \mathcal{C}_n$ and $C_i^\dag P C_i$ is diagonal.
Because $C_i^\dag P C_i$ is diagonal, the associated Pauli rotations $R_{C_i^\dag P C_i}(\pi/4)$ can be implemented using exclusively $\text{CNOT}$, $S$ and $T$ gates.
The set of Pauli rotations $R_{C_i^\dag P C_i}(\pi/4)$ for a fixed $i$ and where $R_P \in L_i$ can then be implemented into the same Hadamard-free subcircuit in which the number of $T$ gates can be optimized.

The second method to circumvent the Hadamard gates in the circuit relies on a measurement-based gadget which can substitute a Hadamard gate~\cite{bremner2011classical}.
This gadget, presented in Figure~\ref{fig:gadgetization}, involves an ancilla qubit, a CZ gate and a measurement.
If all the Hadamard gates in the circuits are gadgetized, then the circuit is Hadamard-free and the number of $T$ gates can be optimized using algorithms specifically designed for Hadamard-free circuits.
Only internal Hadamard gates, which are the Hadamard gates comprised between the first and the last $T$ gates of the circuit, are required to be gadgetized.
Indeed, if only the internal Hadamard gates are gadgetized then the circuit can be partitioned into an initial and final Clifford circuit and a Hadamard-free circuit in between containing all the $T$ gates.
The main drawback of this method is that one additional qubit must be used for each internal Hadamard gate that is gadgetized.
This motivates the minimization of internal Hadamard gates.
To do so, we will use the algorithm presented in Reference~\cite{vandaele2024optimal} which performs the synthesis of the sequence of Pauli rotations of Equation~\ref{eq:sequence_prod} with a minimal number of internal Hadamard gates.

The classically controlled $X$ gate of Figure~\ref{fig:gadgetization}, can be commuted through the subsequent $R_Z$ gates of the circuit by relying on the following equality: 
\begin{equation}
    R_Z(\theta)X = X R_Z(-\theta) = X R_Z(-2\theta) R_Z(\theta)
\end{equation}
In a Clifford$+T$ circuit, the angle $\theta$ is a multiple of $\pi/4$, and so $-2\theta$ is a multiple of $\pi/2$, which results in a rotation implementable using only Clifford gates.
The classically controlled $X$ gate is only commuted through the Pauli rotations succeeding the associated Hadamard gate, and not the Pauli rotations preceding it.
Therefore, the optimized circuit in which all Hadamard gates have been gadgetized can be implemented with a measurement depth equal to $\lvert L \rvert - 1$, where $\lvert L \rvert$ corresponds to the number of layers as in Equation~\ref{eq:sequence_prod_3}.
An example of classically controlled Clifford gates resulting from the gadgetization of a Hadamard gate is provided in Figure~\ref{fig:t_opt_correction_example}.

\begin{figure}[t]
    \centering
    \scalebox{0.9}{
        \includegraphics{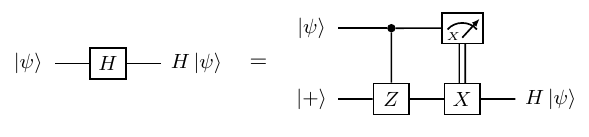}
    }
    \caption{Circuit transformation corresponding to the gadgetization of a Hadamard gate.}
    \label{fig:gadgetization}
\end{figure}

\subsection{Weighted polynomial and signature tensor}\label{sub:weighted_polynomial}

We now describe the established formalism for the optimization of the number of $T$ gates in $\{\text{CNOT}, S, T\}$ circuits.
Let $C$ be a $\{\text{CNOT}, S, T\}$ circuit operating over $n$ qubits.
It has been demonstrated~\cite{amy2013meet} that the action of $C$ on a basis state has the form
\begin{equation}
    \lvert \bs x\rangle \mapsto U_f \lvert g(\bf \bs x)\rangle
\end{equation}
where $g: \mathbb{F}_2^n \rightarrow \mathbb{F}_2^n$ is a linear reversible Boolean function which can be implemented using only CNOT gates~\cite{aaronson2004improved}, and
\begin{equation}\label{eq:uf_gate}
    U_f = \sum_{\bs x \in \mathbb{Z}_2^n} \omega^{f(\bs x)} \lvert \bs x \rangle \langle \bs x \rvert
\end{equation}
where $\omega = e^{i\pi/4}$ and $f: \mathbb{Z}_2^n \rightarrow \mathbb{Z}_8$ is a multilinear polynomial of degree 3 such that
\begin{equation}
    f(\bs x) = \sum_{\alpha}^n l_\alpha x_\alpha - 2 \sum_{\alpha < \beta}^n q_{\alpha, \beta} x_\alpha x_\beta + 4 \sum_{\alpha < \beta < \gamma}^n c_{\alpha, \beta, \gamma} x_\alpha x_\beta x_\gamma \pmod{8}
\end{equation}
where $l_\alpha \in \mathbb{Z}_8$, $q_{\alpha, \beta} \in \mathbb{Z}_4$ and $c_{\alpha, \beta, \gamma} \in \mathbb{Z}_2$.
In the same way as in Reference~\cite{campbell2017unified}, we will refer to the function $f$ as a weighted polynomial due to the fact that each monomial of order $m$ has a coefficient weighted by $2^{m-1}$.
It has been shown in Reference~\cite{campbell2017unified} that $U_f$ belongs to the diagonal subgroup of the third level of the Clifford hierarchy~\cite{gottesman1999demonstrating}.
We will use $\mathcal{D}_3$ to denote this group and $\mathcal{D}^C_3$ to refer to the unitaries of $\mathcal{D}_3$ which are implementable using only CCZ gates.

The parities of the coefficients for a weighted polynomial $f$ can be described by the signature tensor $\mathcal{A}^{(U_f)} \in \mathbb{Z}_2^{(n,n,n)}$~\cite{heyfron2018efficient}, such that $\mathcal{A}^{(U_f)}$ is a symmetric tensor of order 3 satisfying
\begin{equation}
\begin{aligned} 
    \mathcal{A}_{\alpha, \alpha, \alpha} &\equiv l_{\alpha} &\pmod{2} \\
    \mathcal{A}_{\sigma(\alpha,\beta,\beta)} &\equiv \mathcal{A}_{\sigma(\alpha,\alpha,\beta)} = q_{\alpha,\beta} &\pmod{2} \\
    \mathcal{A}_{\sigma(\alpha,\beta,\gamma)} &\equiv c_{\alpha,\beta,\gamma} &\pmod{2}
\end{aligned}
\end{equation}
where $\alpha, \beta, \gamma$ are satisfying $0 \leq \alpha < \beta < \gamma < n$ and $\sigma$ denotes all permutations of the indices.
For convenience, we will drop the superscript $(U_f)$ from $\mathcal{A}$ when it is clear from the context that $\mathcal{A}$ is associated with $U_f$.
It has been proven in Reference~\cite{campbell2017unified} that $U_{2f}$ can be implemented using only Clifford gates for any weighted polynomial $f$.
It implies that two unitaries $U_f$ and $U_{f'}$, where $f$ and $f'$ are weighted polynomials, are Clifford equivalent if the coefficients of $f$ and $f'$ all have the same parity, i.e.\ if $\mathcal{A}^{(U_f)}$ and $\mathcal{A}^{(U_{f'})}$ are equal.

\begin{figure}[t]
    \centering
    \resizebox{1.0\columnwidth}{!}{
        \includegraphics{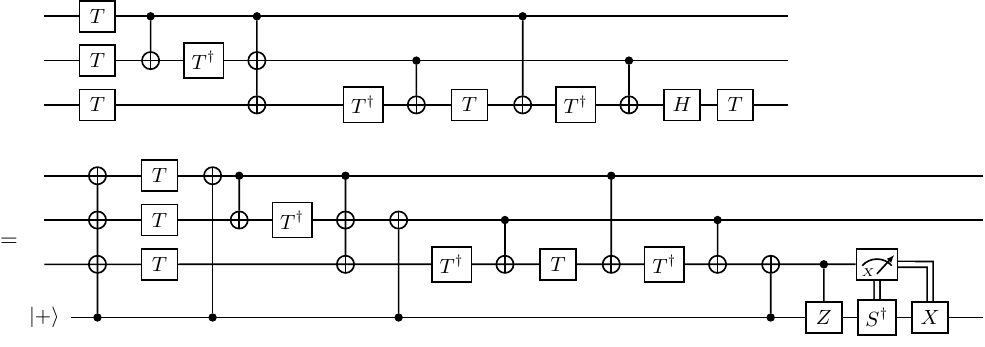}
    }
    \caption{Example of classically controlled Clifford gates resulting from the gadgetization of a Hadamard gate and allowing the optimization of the number of $T$ gates.}
    \label{fig:t_opt_correction_example}
\end{figure}

\subsection{Phase polynomial and parity table}

An implementation of the $U_f$ gate for a weighted polynomial $f$ can be described by a phase polynomial via the circuit-polynomial correspondence~\cite{dawson2004quantum, montanaro2017quantum}.
A phase polynomial $p$ is a linear combination of linear Boolean functions:
\begin{equation}
    p(\bs x) = \sum_{i = 1}^{m} a_i (y^{(i)}_1x_1 \oplus \ldots \oplus y^{(i)}_n x_n) \pmod{8}
\end{equation}
where $\bs y^{(i)} \in \mathbb{Z}_2^n\setminus \{\bs 0\}$, $\bs a \in \mathbb{Z}_{8}^n$ an $m \geq 0$.
We will refer to the Boolean vectors $\bs y^{(i)}$ as the parities of the phase polynomial $p$ and to $\bs a$ as the weights of $p$.
The parities of a phase polynomial can be described by a Boolean matrix, called parity table and denoted $P$, of size $n \times m$ where $n$ is the number of qubits and $m$ is the number of parities weighted by a non-zero $a_i$.
For every weighted polynomial $f$ we can find a phase polynomial $p$ with weights $\bs a$ such that $p(\bs x) = f(\bs x)$ for all $\bs x$.
Such phase polynomial can then be used to implement $U_f$ via a phase polynomial synthesis algorithm which will result in a circuit containing $|\bs a \pmod{2}|$ $T$ gates.
As we are focusing on minimizing the number of $T$ gates, we will represent a parity $\bs y^{(i)}$ in the parity table $P$ if and only if its weight is satisfying $a_i \equiv 1 \pmod{2}$.
The number of columns of $P$ is then equal to the number of $T$ gates required to implement the phase polynomial $p$.
For example, the weighted polynomial associated with the CCZ gate is $f(x_1, x_2, x_3) = 4x_1x_2x_3$.
Its symmetric tensor $\mathcal{A} \in \mathbb{Z}_2^{(3, 3, 3)}$ satisfies:
\begin{equation}
    \mathcal{A}_{\alpha, \beta, \gamma} = 
    \begin{cases}
        1 &\text{if $\alpha \neq \beta \neq \gamma$},\\
        0 &\text{otherwise}.
    \end{cases}
\end{equation}
One possible phase polynomial that can be use to implement this weighted polynomial is $p(x_1, x_2, x_3) = x_1 + x_2 + x_3 + 7(x_1 \oplus x_2) + 7(x_1 \oplus x_3) + 7(x_2 \oplus x_3) + (x_1 \oplus x_2 \oplus x_3)$.
The parity table $P$ and the weights $\bs a$ associated with this phase polynomial are
\begin{align*}
    P &= \begin{pmatrix}
        1 & 0 & 0 & 1 & 1 & 0 & 1 \\
        0 & 1 & 0 & 1 & 0 & 1 & 1 \\
        0 & 0 & 1 & 0 & 1 & 1 & 1 \\
    \end{pmatrix},\quad
    \bs a = \begin{pmatrix} 1 & 1 & 1 & 7 & 7 & 7 & 1 \end{pmatrix}^T.
\end{align*}
We can verify that $f(\bs x) = p(\bs x)$ for all $\bs x$ by using the identity $2xy \equiv x + y + 7(x\oplus y) \pmod{8}$~\cite{selinger2013quantum}.
The synthesis of a phase polynomial $p$ with weights satisfying $a_i \in \mathbb{Z}_8$ for all $i$ can be realized using $\mathrm{CNOT}$ gates and one $R_Z$ gate with an angle of $a_i \pi/4$ for each parity $\bs y^{(i)}$ in $p$.
A circuit implementing the phase polynomial of this example with the CCZ gate is represented in Figure~\ref{fig:ccz_example}.

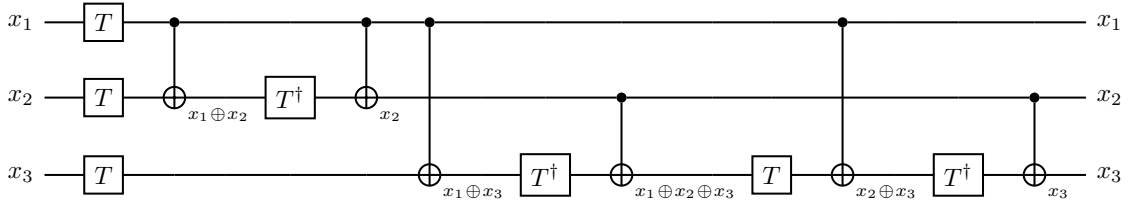
\begin{figure}[t]
    \resizebox{1.0\columnwidth}{!}{
    \begin{quantikz}[column sep=0.55cm]
        \lstick{$x_1$} & \gate{T} & \ctrl{1} & \qw & \qw & \ctrl{1} & \ctrl{2} & \qw & \qw & \qw & \qw & \qw & \qw & \ctrl{2} & \qw & \qw & \qw & \qw \rstick{$x_1$} \\
        \lstick{$x_2$} & \gate{T} & \targ{} \rstick{\\\vspace{-0.2cm}\scriptsize $\!\! x_1\! \oplus\! x_2$} & \qw & \gate{T^\dag} & \targ{}  \rstick{\\\vspace{-0.2cm}\scriptsize $\!\! x_2$} & \qw & \qw & \qw & \ctrl{1} & \qw & \qw & \qw & \qw & \qw & \qw & \ctrl{1} & \qw \rstick{$x_2$} \\
        \lstick{$x_3$} & \gate{T} & \qw & \qw & \qw & \qw & \targ{} \rstick{\\\vspace{-0.2cm}\scriptsize $\!\! x_1\! \oplus\! x_3$} & \qw & \gate{T^\dag} & \targ{} \rstick{\\\vspace{-0.2cm}\scriptsize $\!\! x_1\! \oplus\! x_2\!\oplus\! x_3$} & \qw & \qw & \gate{T} & \targ{} \rstick{\\\vspace{-0.2cm}\scriptsize $\!\! x_2\! \oplus\! x_3$} & \qw & \gate{T^\dag} & \targ{} \rstick{\\\vspace{-0.2cm}\scriptsize $\!\! x_3$} & \qw \rstick{$x_3$}
    \end{quantikz}}
    \caption{An implementation of the $U_f$ gate with weighted polynomial $f(\bs x) = 4x_1x_2x_3$, which corresponds to the CCZ gate.}
    \label{fig:ccz_example}
\end{figure}

\begin{figure}[t]
    \centering
    \resizebox{1.0\columnwidth}{!}{
        \includegraphics{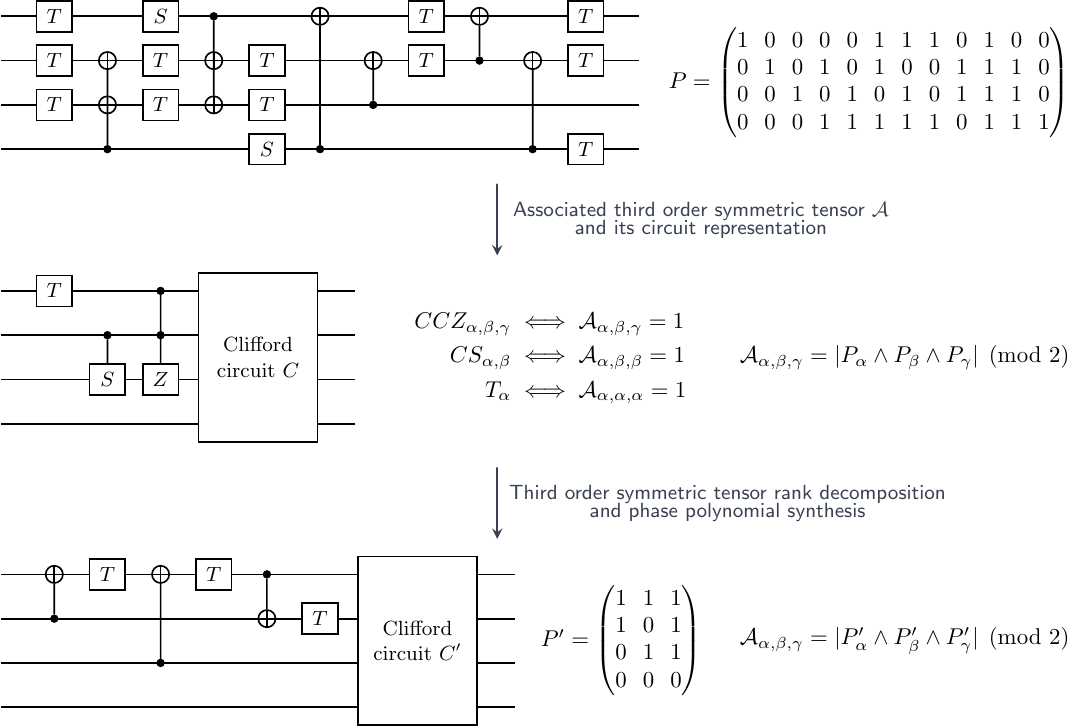}
    }
    \caption{Overview of the process for the optimization of the number of $T$ gates in Hadamard-free circuits.
    We first compute the parity table $P$ associated with the phase polynomial implemented by the initial circuit.
    Then, we compute the weighted polynomial associated with $P$, which can be represented by the third order symmetric tensor $\mathcal{A}$ or by a $\{T, CS, CCZ\}$ circuit.
    We then solve the third order symmetric tensor rank decomposition problem to find another parity table $P'$ associated with $\mathcal{A}$ but which contains a minimal number of columns.
    Finally, we perform the synthesis of the phase polynomial represented by $P'$ along with the final Clifford operator to obtain the optimized circuit.}
    \label{fig:t_opt_example}
\end{figure}

The weighted polynomial associated with a phase polynomial composed of a single parity $p(\bs x) = a(x_1 \oplus \ldots \oplus x_n)$, where $a \in \mathbb{Z}_8$, can be computed using the following equality:
\begin{equation}\label{eq:weighted_polynomial_eq}
a \left(x_1 \oplus \ldots \oplus x_n\right) = a \left(\sum_{\alpha}^n x_\alpha - \sum_{\alpha < \beta}^n 2 x_\alpha x_\beta + \sum_{\alpha < \beta < \gamma}^n 4 x_\alpha x_\beta x_\gamma \right) \pmod{8}.
\end{equation}
We provide a proof of Equation~\ref{eq:weighted_polynomial_eq}, and its more general form, in Section~\ref{sub:preliminaries_d}.
The coefficients of the weighted polynomial represented in the right-hand side of this equation are $l_\alpha = a \pmod{8}$, $q_{\alpha, \beta} = a \pmod{4}$ and $c_{\alpha, \beta, \gamma} = a \pmod{2}$.
Equation~\ref{eq:weighted_polynomial_eq} can then be used for each parity of a given phase polynomial to compute its associated weighted polynomial:
\begin{equation}
\begin{aligned}
        p(\bs x) &= \sum_{i = 1}^{m} a_i (y^{(i)}_1x_1 \oplus \ldots \oplus y^{(i)}_n x_n) \pmod{8} \\
                 &= \sum_{i = 1}^{m} a_i \left(\sum_{\alpha}^n y_\alpha^{(i)} x_\alpha - \sum_{\alpha < \beta}^n 2 y_\alpha^{(i)}y_\beta^{(i)}x_\alpha x_\beta + \sum_{\alpha < \beta < \gamma}^n 4 y_\alpha^{(i)}y_\beta^{(i)}y_\gamma^{(i)}x_\alpha x_\beta x_\gamma \right) \pmod{8} \\
                 &= \sum_{\alpha}^n l_\alpha x_\alpha - \sum_{\alpha < \beta}^n 2 q_{\alpha, \beta} x_\alpha x_\beta + \sum_{\alpha < \beta < \gamma}^n 4c_{\alpha, \beta, \gamma} x_\alpha x_\beta x_\gamma \pmod{8}
\end{aligned}
\end{equation}
where $l_\alpha, q_{\alpha, \beta}$ and $c_{\alpha, \beta, \gamma}$ are satisfying
\begin{equation}
\begin{aligned}
    l_\alpha &= \sum_{i=1}^m a_i y_\alpha^{(i)} \pmod{8}, \\
    q_{\alpha, \beta} &= \sum_{i=1}^m a_i y_\alpha^{(i)} y_\beta^{(i)} \pmod{4}, \\
    c_{\alpha, \beta, \gamma} &= \sum_{i=1}^m a_i y_\alpha^{(i)} y_\beta^{(i)} y_\gamma^{(i)} \pmod{2}. \\
\end{aligned}
\end{equation}

Two weighted polynomials $f$ and $f'$, with associated coefficients $l_\alpha, q_{\alpha, \beta}, c_{\alpha, \beta, \gamma}$ and $l'_\alpha, q'_{\alpha, \beta}, c'_{\alpha, \beta, \gamma}$ respectively, are equal if and only if $l_\alpha = l'_\alpha, q_{\alpha, \beta} = q'_{\alpha, \beta}$ and $c_{\alpha, \beta, \gamma} = c'_{\alpha, \beta, \gamma}$ for all $\alpha, \beta, \gamma$.
Indeed, if these equalities are not satisfied, then we can find some vector $\bs x$ such that $f(\bs x) \neq f'(\bs x)$.
Therefore, two phase polynomials $p$ and $p'$ are implementing the same operator if and only if their associated weighted polynomials $f$ and $f'$ are equal.
The problem of $T$-count optimization then consists in finding a phase polynomial $p$ which implements a given weighted polynomial $f$ with a minimal number of $T$ gates.
Notice that the coefficients of the weighted polynomial represented in Equation~\ref{eq:weighted_polynomial_eq} all have the same parity: $l_{\alpha} \equiv q_{\alpha, \beta} \equiv  c_{\alpha, \beta, \gamma}\equiv a \pmod{2}$.
And if $a$ is even, then the associated rotation of angle $a\pi/4$ is a multiple of $\pi/2$ and can be implemented using only $S$ gates.
That is why a weighted polynomial $f$ can be implemented using only $\{\mathrm{CNOT}, S\}$ gates if and only if all its coefficients have an even parity, in such case we say that $f$ is a Clifford weighted polynomial.
Then, a weighted polynomial $p$ with an associated parity table $P$ is an implentation of a weighted polynomial $f$ with an associated signature tensor $\mathcal{A}$ up to an operator implementable over the $\{\mathrm{CNOT}, S\}$ gate set if and only if the equality
\begin{equation}\label{eq:tensor_parity_table_equality}
    \mathcal{A}_{\alpha,\beta,\gamma} = \lvert P_\alpha \wedge P_\beta \wedge P_\gamma \rvert \pmod{2}
\end{equation}
is satisfied for all $\alpha, \beta, \gamma$ satisfying $0 \leq \alpha \leq \beta \leq \gamma < n$, where $n$ is the number of qubits.
Throughout the paper, the notation $\lvert \bs v \rvert$ will be used to refer to the Hamming weight of the vector $\bs v$, and the symbol $\wedge$ will refer to the logical AND operation.
Let $f'$ be the weighted polynomial implemented by $p$, then $f$ and $f'$ have the same signature tensor, and $f$ can be implemented by performing the synthesis of $p$ and the synthesis of the Clifford operator associated with the Clifford weighted polynomial $f - f'$.
The problem of finding a phase polynomial implementing a given weighted polynomial with a minimal number of $T$ gates and up to an operator implementable over the $\{\mathrm{CNOT}, S\}$ gate set can then be described by the following third order symmetric tensor rank decomposition (3-STR) problem~\cite{heyfron2018efficient}.

\begin{problem}[3-STR]\label{pb:3_str}
    Let $\mathcal{A} \in \mathbb{Z}_2^{(n,n,n)}$ be a symmetric tensor such that
    \begin{equation}
        \mathcal{A}_{\alpha, \beta, \gamma} = \mathcal{A}_{\alpha', \beta', \gamma'}
    \end{equation}
    for all $\alpha, \beta, \gamma$ and $\alpha', \beta', \gamma'$ satisfying the set equality $\{\alpha, \beta, \gamma\} = \{\alpha', \beta', \gamma'\}$.
    Find a Boolean matrix $P$ of size $n \times m$ such that
    \begin{equation}
        \mathcal{A}_{\alpha,\beta,\gamma} = \lvert P_\alpha \wedge P_\beta \wedge P_\gamma \rvert \pmod{2}
    \end{equation}
    for all $\alpha, \beta, \gamma$ satisfying $0 \leq \alpha \leq \beta \leq \gamma < n$, with minimal $m$.
\end{problem}

The $T$-count minimization problem in $\{\mathrm{CNOT}, T, S\}$ circuits have first been shown to be equivalent to finding a minimum distance decoding in the order $n-4$ punctured Reed-Muller code of length $2^n - 1$~\cite{amy2019t}, noted $\mathcal{RM}(n-4, n)^*$, which is equivalent to the 3-STR problem~\cite{seroussi1983maximum}.
The complexity class of the 3-STR problem is unknown, however, the related problem of finding the tensor rank of a tensor of order $3$ is NP-complete~\cite{haastad1989tensor}.
And the more general problem of optimizing the number of $T$ gates in a Clifford$+T$ circuit is NP-hard~\cite{van2023optimising}.

An illustrative overview of the process we described for the optimization of the number of $T$ gates in Hadamard-free circuits is provided in Figure~\ref{fig:t_opt_example}.
In Section~\ref{sec:beyond_merging}, we present two algorithms attempting to solve the 3-STR problem and we prove that the number of $T$ gates in the circuit produced by these algorithms is upper bounded by $(n^2 + n)/2+1$.
We evaluate the perfomances of our algorithms in the benchmarks of Section~\ref{subsec:bench_tohpe_todd}.
And we generalize our results for the optimization of the number of $R_Z(\pi/2^d)$ gates, where $d$ is a non-negative integer, in Section~\ref{sec:higher_orders}.

\section{$T$-count reduction algorithms}\label{sec:beyond_merging}

In this section we tackle the 3-STR problem as defined in Section~\ref{sec:preliminaries}.
First, we propose an algorithm for this problem in Subsection~\ref{sub:tohpe} and we prove that the number of $T$ gates in the solution produced by this algorithm is upper bounded by $n(n+1)/2+1$.
Then, in Subsection~\ref{sub:improving_todd}, we show how the complexity of the \texttt{TODD} algorithm of Reference~\cite{heyfron2018efficient} can be reduced and we propose some modifications to this algorithm to improve its performances.

\subsection{Third order homogeneous polynomials elimination algorithm}\label{sub:tohpe}

The key mechanism used by our algorithm for reducing the $T$-count is based on the following theorem.

\begin{theorem}\label{thm:tohpe}
    Let $P$ be a parity table of size $n \times m$ and $P' = P \oplus \bs z \bs y^T$ where $\bs z$ and  $\bs y$ are vectors of size $n$ and $m$ respectively such that
    \begin{align}
        \lvert \bs y \rvert &\equiv 0 \pmod{2} \label{eq:tohpe_condition_1}\\
        \lvert P_\alpha \wedge \bs y \rvert &\equiv 0 \pmod{2} \label{eq:tohpe_condition_2}\\
        \lvert P_\alpha \wedge P_\beta \wedge \bs y \rvert &\equiv 0 \pmod{2}\label{eq:tohpe_condition_3}
    \end{align}
    for all $0 \leq \alpha < \beta < n$.
    Then we have
    \begin{equation}\label{eq:parity_table_equality}
        \lvert P'_\alpha \wedge P'_\beta \wedge P'_\gamma \rvert \equiv \lvert P_\alpha \wedge P_\beta \wedge P_\gamma \rvert \pmod{2}
    \end{equation}
    for all $0 \leq \alpha \leq \beta \leq \gamma < n$.
\end{theorem}

\begin{proof}
    For all $\alpha, \beta, \gamma$ satisfying $0 \leq \alpha \leq \beta \leq \gamma < n$, we have:
    \begin{equation}\label{eq:tohpe_proof}
    \begin{aligned}
        \lvert P'_\alpha \wedge P'_\beta \wedge P'_\gamma \rvert &= \lvert (P_\alpha \oplus z_\alpha \bs y) \wedge (P_\beta \oplus z_\beta \bs y) \wedge (P_\gamma \oplus z_\gamma \bs y) \rvert \\
          &= \lvert \left[(P_\alpha \wedge P_\beta) \oplus z_\beta(P_\alpha \wedge \bs y) \oplus z_\alpha(P_\beta \wedge \bs y) \oplus z_\alpha z_\beta \bs y\right] \wedge (P_\gamma \oplus z_\gamma \bs y) \rvert \\
          &\equiv \lvert P_\alpha \wedge P_\beta \wedge P_\gamma \rvert + z_\gamma\lvert P_\alpha \wedge P_\beta \wedge \bs y \rvert + z_\beta\lvert P_\alpha \wedge P_\gamma \wedge \bs y \rvert + z_\alpha \lvert P_\beta \wedge P_\gamma \wedge \bs y \rvert \\
          & \quad + z_\beta z_\gamma\lvert P_\alpha \wedge \bs y \rvert + z_\alpha z_\gamma\lvert P_\beta \wedge \bs y \rvert + z_\alpha z_\beta\lvert P_\gamma \wedge \bs y \rvert + z_\alpha z_\beta z_\gamma\lvert \bs y \rvert \pmod{2} \\
&\equiv \lvert P_\alpha \wedge P_\beta \wedge P_\gamma \rvert \pmod{2}
    \end{aligned}
    \end{equation}
\end{proof}

If the conditions of Theorem~\ref{thm:tohpe} are satisfied, then the set of columns selected by $\bs y$ are forming a weighted polynomial that can be divided into two weighted polynomials $f$ and $f'$, where $f$ is a third order homogeneous weighted polynomial:
\begin{equation}\label{eq:third_order_homogeneous_polynomial}
    f(\bs x) = 4 \sum_{\alpha < \beta < \gamma}^n c_{\alpha, \beta, \gamma} x_\alpha x_\beta x_\gamma \pmod{8}
\end{equation}
where $c_{\alpha, \beta, \gamma} = \lvert P_\alpha \wedge P_\beta \wedge P_\gamma \wedge \bs y \rvert \pmod{2}$, and where $f'$  is a Clifford weighted polynomial:
\begin{equation}
    f'(\bs x) = \sum_{\alpha}^n l_\alpha x_\alpha + 2 \sum_{\alpha < \beta}^n q_{\alpha, \beta} x_\alpha x_\beta \pmod{8}
\end{equation}
where $l_\alpha \in \mathbb{Z}_8$, $q_{\alpha, \beta} \in \mathbb{Z}_4$ and $l_\alpha \equiv q_{\alpha, \beta} \equiv 0 \pmod{2}$.
The unitary gate associated with the weighted polynomial $f$ belongs to the $\mathcal{D}^C_3$ group: it is implementable using only CCZ gates~\cite{campbell2017unified}.
It has already been shown that such weighted polynomials can be exploited to reduce the $T$-count, notably via the subadditivity theorem of Reference~\cite{campbell2017unified}.
Let $U \in \mathcal{D}_3$ act on $n$ qubits, in the following we define $\tau[U]$ as the optimal $T$-count to implement $U$ without ancillary qubits:
\begin{equation}
    \tau[U] = \min \{t \mid U = C_0 T_1 C_1 \ldots T_t C_t, \{C_0, \ldots, C_t\} \in \mathcal{C}_n^* \}
\end{equation}
where $\mathcal{C}_n^*$ is the subgroup of Clifford operators which can be implemented with CNOT and $S$ gates.
The subadditivity theorem of Reference~\cite{campbell2017unified} states that if $\tau[U_1] \equiv 1 \pmod{2}$ and $\tau[U_2] > 0$, then $\tau[U_1 \otimes U_2] < \tau[U_1] + \tau[U_2]$ where $U_1 \in \mathcal{D}^C_3$ and $U_2 \in \mathcal{D}_3$.
Based on Theorem~\ref{thm:tohpe}, we can actually remove the condition $\tau[U_1] \equiv 1 \pmod{2}$, which gives the following theorem.

\begin{theorem}[Subadditivity theorem]\label{thm:subadditivity}
    Let $U_1 \in \mathcal{D}^C_3$, and $U_2 \in \mathcal{D}_3$.
    If $\tau[U_1], \tau[U_2] > 0$, then $\tau[U_1 \otimes U_2] < \tau[U_1] + \tau[U_2]$.
\end{theorem}
\begin{proof}
    Let $W = \begin{bmatrix} P & Q \end{bmatrix}$ be a parity table such that $P$ and $Q$ are the parity tables associated with the implementation of $U_1$ and $U_2$ and which have $\tau[U_1]$ and $\tau[U_2]$ columns respectively.
    Then, because $U_1 \in \mathcal{D}^C_3$, $P$ satisfies the following equations: 
    \begin{align}
        \lvert P_\alpha \rvert &\equiv 0 \pmod{2} \\
        \lvert P_\alpha \wedge P_\beta \rvert &\equiv 0 \pmod{2}
    \end{align}
    for all $\alpha, \beta$.
    Let $\bs z = P_{:,i} \oplus Q_{:,j}$ for any $i$ and $j$ satisfying $0\leq i < \tau[U_1]$, $0\leq j < \tau[U_2]$, where $P_{:,i}$ denotes the $i$th column of the parity table $P$.
    And let $P'$ be a parity table such that
    $$
    P' = 
    \begin{cases}
        P \oplus \bs z \bs 1^T &\text{if $\tau[U_1] \equiv 0 \pmod{2}$},\\ 
        \begin{bmatrix} P \oplus \bs z \bs 1^T & \bs z \end{bmatrix} &\text{otherwise}.\\ 
    \end{cases}
    $$
    Then, as stated by Theorem~\ref{thm:tohpe}, the parity table $P'$ satisfies 
    \begin{equation}
        \lvert P'_\alpha \wedge P'_\beta \wedge P'_\gamma \rvert \equiv \lvert P_\alpha \wedge P_\beta \wedge P_\gamma \rvert \pmod{2}
    \end{equation}
    for all $\alpha, \beta, \gamma$.
    And so the parity table $W' = \begin{bmatrix} P' & Q \end{bmatrix}$, which has at most one more column than $W$, also satisfies 
    \begin{equation}\label{eq:w_equality}
        \lvert W'_\alpha \wedge W'_\beta \wedge W'_\gamma \rvert \equiv \lvert W_\alpha \wedge W_\beta \wedge W_\gamma \rvert \pmod{2}
    \end{equation}
    for all $\alpha, \beta, \gamma$.
    However, we can notice that $P'_{:,i} = Q_{:,j}$.
    Therefore, by removing these two columns from $W'$ Equation~\ref{eq:w_equality} still holds and $W'$ has at least one less column than $W$.
    The parity table $W'$ implements the unitary $U_1 \otimes U_2$ up to a Clifford operator and has at most $\tau[U_1] + \tau[U_2] - 1$ columns, thus we have $\tau[U_1 \otimes U_2] \leq \tau[U_1] + \tau[U_2] - 1 < \tau[U_1] + \tau[U_2]$.
\end{proof}

Based on this subadditivity theorem and on Theorem~\ref{thm:tohpe}, we can derive the following upper bound on the number of $T$ gates in a $\{\mathrm{CNOT}, S, T\}$ circuit.
\begin{theorem}\label{thm:tohpe_upper_bound}
    The number of $T$ gates in an $n$-qubits $\{\mathrm{CNOT}$, $T$, $S\}$ circuit can be upper bounded by
    \begin{equation}
        2 \lfloor (n^2 + n)/4 \rfloor + 1 \leq (n^2 + n)/2 + 1
    \end{equation}
    in polynomial time.
\end{theorem}

\begin{proof}
    Let $U$ be a unitary gate implementable by a $\{\mathrm{CNOT}, S, T\}$ gate set, let $P$ be a parity table of size $n \times m$ which implements $U$ up to an operator implementable over the $\{\mathrm{CNOT}$, $S\}$ gate set, and let $L$ be a matrix with rows labelled by $(\alpha\beta)$ such that
    \begin{align}
        L_{\alpha\beta} &= P_\alpha \wedge P_\beta
    \end{align}
    for all $\alpha, \beta$ satisfying $0\leq \alpha \leq \beta < n$.
    If $P$ has strictly more than $(n^2 + n)/2 + 1$ columns then we can necessarily find a non-zero vector $\bs y$ satisfying $L \bs y = \bs 0$ and $\bs y \neq \bs 1$ because $L$ has $(n^2 + n)/2$ rows.
    Note that such vector $\bs y$ necessarily satisfies Equations~\ref{eq:tohpe_condition_2} and~\ref{eq:tohpe_condition_3} of Theorem~\ref{thm:tohpe}.
    We can then divide $P$ into two non-empty parity tables $P^{(1)}$ and $P^{(2)}$ where the column $P_{:,i}$ belongs to $P^{(1)}$ if and only if $y_i = 1$ and to $P^{(2)}$ otherwise.
    The parity tables $P^{(1)}$ and $P^{(2)}$ are implementations of some unitary gates $U_1 \in \mathcal{D}^C_3$ and $U_2 \in \mathcal{D}_3$ respectively.
    The subadditivity theorem (Theorem~\ref{thm:subadditivity}) can then be exploited to reduce the number of columns of $P$.
    Let $\bs z = P_{:,i} \oplus P_{:,j}$ where $i$ and $j$ are satisfying $y_i = 1$ and $y_j = 0$, and let 
    $$
    P' = 
    \begin{cases}
        P \oplus \bs z \bs y^T &\text{if $\lvert \bs y \rvert \equiv 0 \pmod{2}$},\\ 
        \begin{bmatrix} P \oplus \bs z \bs y^T & \bs z \end{bmatrix} &\text{otherwise}.\\ 
    \end{cases}
    $$
    By Theorem~\ref{thm:tohpe}, $P'$ implements the same unitary gate as $P$ up to an operator implementable over the $\{\mathrm{CNOT}$, $S\}$ gate set.
    The parity table $P'$ has at most one more column than $P$ and we have $P'_{:,i} = P'_{:,j}$.
    Therefore, the columns $i$ and $j$ can be removed from $P'$, which entails that $P'$ has at least one less column than $P$.
    We showed that if the number of columns of $P$ is strictly greater than $(n^2 + n)/2 + 1$, then the number of columns of $P$ can be reduced by at least one in polynomial time.
    Moreover, if the number of columns of $P$ is equal to $(n^2 + n)/2 + 1$ and is even, then we can necessarily find a non-zero vector $\bs y$ satisfying $L \bs y = \bs 0$ because $L$ has $(n^2 + n)/2$ rows.
    If there exist $i$ such that $y_i = 0$ then we can reduce the number of columns of $P$ as described above.
    Otherwise we must have $\lvert \bs y \rvert \equiv 0 \pmod{2}$, and so $\bs y$ satisfies the Equations of Theorem~\ref{thm:tohpe}.
    Therefore, if $P' = P \oplus \bs z \bs y^T$ where $\bs z$ is equal to the $i$th column of $P$ for any $i$, then $P'$ implements the same unitary gate as $P$ up to an operator implementable over the $\{\mathrm{CNOT}$, $S\}$ gate set.
    The parity table $P'$ has the same number of columns as $P$ but its $i$th column is equal to the null vector and can therefore be removed, which leads to a parity table containing $(n^2 + n)/2$ columns.
    The polynomial-time procedure described above to reduce the number of columns of $P$ can be repeated until $P$ has a number of columns lower or equal to
    \begin{equation}
        2 \lfloor (n^2 + n)/4 \rfloor + 1 \leq (n^2 + n)/2 + 1
    \end{equation}
\end{proof}

Note that, for $n>5$, this upper bound is better than the previously best known upper bound of $(n^2 + 3n - 14)/2$~\cite{campbell2017unified}.
The proof of Theorem~\ref{thm:tohpe_upper_bound} provides a straightforward algorithm for achieving this upper bound in polynomial time.
We propose some improvements regarding this approach by providing an algorithm whose pseudo-code is given in Algorithm~\ref{alg:tohpe}.

The algorithm starts by computing the set $Z$ which contains all the vectors $\bs z$ which can potentially be used to transform $P$ via Theorem~\ref{thm:tohpe} and reduce its number of columns:
\begin{equation}
        Z = \{P_{:,i} \oplus P_{:,j} \mid 0 \leq i < j < m \} \cup \{P_{:,i}  \mid 0 \leq i < m \}.
\end{equation}
For each vector $\bs z \in Z$, the algorithm also computes the set $S^{(\bs z)}$ of pairs of indices $\{i, j\}$ satisfying $P_{:,i} \oplus P_{:,j} = \bs z$ in the case where $i\neq j$ or satisfying $P_{:,i} = \bs z$ in the case where $i = j$.
Let $L$ be a matrix with rows labelled by $(\alpha\beta)$ such that
\begin{align}
    L_{\alpha\beta} &= P_\alpha \wedge P_\beta
\end{align}
for all $\alpha < \beta$.
In order to exploit the subadditivity theorem, the vector $\bs y$ must satisfy $L\bs y = \bs 0$, $\bs y \neq \bs 0$ and $\bs y \neq \bs 1$.
However, in the case where $\bs y = \bs 1$ and $\lvert \bs y \rvert \equiv 0 \pmod{2}$ the number of columns of $P$ can still be reduced.
Indeed, in such case the vector $\bs y$ satisfies all the condition of Theorem~\ref{thm:tohpe}, therefore the parity table $P'$ defined as $P' = P \oplus P_{:,i} \bs 1^T$ for any $i$ is equivalent to $P$ and its $i$th column is equal to the null vector and can therefore be removed, which reduces the number of columns by one.
To summarize, the number of columns of $P$ can be reduced if $\bs y$ satisfies $L\bs y = \bs 0$ and $\bs y \neq \bs 0$ and $\bs y \neq \bs 1$ or $\lvert \bs y \rvert \equiv 0 \pmod{2}$.
If no such $\bs y$ is found, Algorithm~\ref{alg:tohpe} will stop and return $P$.
Otherwise, the algorithm will select a vector $\bs z$ such that the number of duplicated columns plus the number of all-zero columns in the parity table $P' = P \oplus \bs z \bs y^T$ is maximized.
This number corresponds to the number of columns that can be removed from $P$ if $\bs z$ is chosen (minus one in the case where $\lvert \bs y \rvert \equiv 1 \pmod{2}$), and is given by the following objective function:
\begin{equation}
-\lvert \bs y \rvert \pmod{2} + \sum_{\{i, j\} \in S^{(\bs z)}} 
\begin{cases} 
2(y_i \oplus y_j) & \text{if } i \neq j, \\
y_i + 2(y_i \oplus 1)(\lvert \bs y \rvert \pmod{2}) & \text{if } i = j. 
\end{cases}
\end{equation}
The case where $i \neq j$ is straightforward, as the $i$th column will be equal to the $j$th column of $P'$ if and only if $y_i \oplus y_j = 1$, which results in the removal of 2 columns.
In the case where $i = j$, the equation
\begin{equation}\label{eq:t_opt_i_eq_j}
y_i + 2(y_i \oplus 1)(\lvert \bs y \rvert \pmod{2}) 
\end{equation}
is equal to $1$ if $y_i = 1$ because the $i$th column of $P'$ is the null vector, leading to the removal of 1 column.
Otherwise, if $y_i = 0$ and $\lvert \bs y \rvert$ is odd, then Equation~\ref{eq:t_opt_i_eq_j} is equal to $2$ because the $i$th column of $P'$ will be equal the $\bs z$ column vector that is added to $P'$ in the case where $\lvert \bs y \rvert$  is odd, resulting in the removal of 2 columns.
Once the vector $\bs z$ maximizing this objective function has been found, Algorithm~\ref{alg:tohpe} then computes the new parity table $P'$ which contains fewer columns, and performs a recursive call.

A variant of this algorithm could consist in finding both vectors $\bs z$ and $\bs y$ that are maximizing the objective function.
However, our experiments on this approach showed that it significantly increases the complexity of the algorithm for a rather marginal gain in the number of $T$ gates.
The performances of Algorithm~\ref{alg:tohpe} are evaluated in Section~\ref{subsec:bench_tohpe_todd}.

\begin{algorithm}[t]
    \caption{Third order homogeneous polynomials elimination algorithm}
    \label{alg:tohpe}
	\SetAlgoLined
	\SetArgSty{textnormal}
	\SetKwFunction{proc}{TOHPE}
	\SetKwInput{KwInput}{Input}
	\SetKwInput{KwOutput}{Output}
    \KwInput{A parity table $P$ of size $n \times m$.}
    \KwOutput{An equivalent parity table with an optimized number of columns.}
	\SetKwProg{Fn}{procedure}{}{}
    \Fn{\proc{$P$}}{
        $Z \leftarrow \{P_{:,i} \oplus P_{:,j} \mid 0 \leq i < j < m \} \cup \{P_{:,i}  \mid 0 \leq i < m \}$ \\
        \ForAll{$\bs z \in Z$} {
            $S^{(\bs z)} \leftarrow \{\{i, j\} \mid P_{:,i} \oplus P_{:,j} = \bs z\} \cup \{\{i, i\} \mid P_{:,i} = \bs z\}$
        }
        $L \leftarrow$ matrix whose rows are forming the set $\{P_i \wedge P_j \mid 0 \leq i \leq j < n\}$ \\
        \If{$\nexists\bs y$ such that $L\bs y = \bs 0$ and $\bs y \neq \bs 0$ and ($\bs y \neq \bs 1$ or $\lvert \bs y \rvert \equiv 0 \pmod{2}$)}{
            \Return $P$
        }
        $\bs y \leftarrow$ any vector such that $L\bs y = \bs 0$ and $\bs y \neq \bs 0$ and ($\bs y \neq \bs 1$ or $\lvert \bs y \rvert \equiv 0 \pmod{2}$) \\
        $\displaystyle \bs z \leftarrow \argmax_{\bs z \in Z} \set[\Big]{-\lvert \bs y \rvert \pmod{2} +\!\sum_{\{i, j\} \in S^{(\bs z)}} 2 (y_i \oplus y_j) + \delta_{ij}\left[y_i + 2(y_i \oplus 1)(\lvert \bs y \rvert \pmod{2})\right]}$ \\
        $P' \leftarrow P \oplus \bs z \bs y^T$ \\
        \If{$\lvert \bs y \rvert \equiv 1 \pmod{2}$}{
            $P' \leftarrow \begin{bmatrix} P' & \bs z \end{bmatrix}$ \\
        }
        $P' \leftarrow P'$ with all its duplicated and all-zero columns removed \\
        \Return $\texttt{TOHPE}(P')$ \\
	}
\end{algorithm}

\paragraph{Complexity analysis.}
The set $Z$ can be computed with $\mathcal{O}(nm^2)$ operations, and all the sets $S^{(\bs z)}$ can be created with the same time complexity.
The matrix $L$ has $\mathcal{O}(n^2)$ rows and a number of columns that is at most equal to $m$.
Therefore performing a Gaussian elimination to compute a generating set of the right nullspace of $L$ implies a complexity of $\mathcal{O}(n^2m^2)$.
The vector $\bs z$ satisfying the argmax function can be computed in $\mathcal{O}(m^2)$ operations because the union of all the $S^{(\bs z)}$ sets contains $\mathcal{O}(m^2)$ elements.
Updating the parity table $P$ induces $\mathcal{O}(nm)$ operations.
And the algorithm is performing no more than $m$ recursive calls.
Thus, the overall complexity of Algorithm~\ref{alg:tohpe} is $\mathcal{O}(n^2m^3)$.

\subsection{Improving the \texttt{TODD} algorithm}\label{sub:improving_todd}

In this section we show how the \texttt{TODD} algorithm proposed in Reference~\cite{heyfron2018efficient} can be improved.
We first describe the algorithm and demonstrate how its complexity can be reduced.
We then propose a modified version of this algorithm to improve its performances.

The key mechanism of the \texttt{TODD} algorithm rests on the following theorem, which was first proven in Reference~\cite{heyfron2018efficient}.
We provide its proof for completeness in Appendix~\ref{app:todd}.

\begin{theorem}\label{thm:todd}
    Let $P$ be a parity table of size $n \times m$ and $P' = P \oplus \bs z \bs y^T$ where $\bs z$ and  $\bs y$ are vectors of size $n$ and $m$ respectively such that
    \begin{align}
        \lvert \bs y \rvert &\equiv 0 \pmod{2} \label{eq:condition_1} \\
        \lvert P_\alpha \wedge \bs y \rvert &\equiv 0 \pmod{2} \label{eq:condition_2} \\
        \lvert \left[ z_\alpha(P_\beta \wedge P_\gamma) \oplus z_\beta(P_\alpha \wedge P_\gamma) \oplus z_\gamma(P_\alpha \wedge P_\beta)\right] \wedge \bs y \rvert &\equiv 0 \pmod{2} \label{eq:condition_3}
    \end{align}
    for all $0 \leq \alpha < \beta < \gamma < n$.
    Then we have
    \begin{equation}\label{eq:p_equality}
        \lvert P'_\alpha \wedge P'_\beta \wedge P'_\gamma \rvert \equiv \lvert P_\alpha \wedge P_\beta \wedge P_\gamma \rvert \pmod{2}
    \end{equation}
    for all $0 \leq \alpha \leq \beta \leq \gamma < n$.
\end{theorem}

The \texttt{TODD} algorithm exploits this theorem for optimizing the number of $T$ gates as follows.
Let $P$ be a parity table of size $n \times m$, let $\bs z = P_{:,i} \oplus P_{:,j}$ where $i \neq j$ and let $\chi$ be a matrix that contains a row equal to $P_\alpha$ for every $\alpha$ and a row equal to 
$$z_\alpha(P_\beta \wedge P_\gamma) \oplus z_\beta(P_\alpha \wedge P_\gamma) \oplus z_\gamma(P_\alpha \wedge P_\beta)$$
for every triple $\alpha, \beta, \gamma$ satisfying $0 \leq \alpha < \beta < \gamma < n$.
A vector $\bs y$ satisfying the conditions~\ref{eq:condition_2} and~\ref{eq:condition_3} of Theorem~\ref{thm:todd} can be found by computed the nullspace of $\chi$, i.e.\ if $\bs y$ satisfies $\chi \bs y = \bs 0$ then $\bs y$ also satisfies the conditions~\ref{eq:condition_2} and~\ref{eq:condition_3} of Theorem~\ref{thm:todd}.
If $\bs y$ satisfies $\chi \bs y = \bs 0$ and $y_i \oplus y_j = 1$, then the number of columns of $P$ can be reduced.
Let $$P' =
\begin{cases}
    P \oplus \bs z \bs y^T &\text{if $\lvert \bs y \rvert \equiv 0 \pmod{2}$},\\ 
    \begin{bmatrix} P \oplus \bs z \bs y^T & \bs z \end{bmatrix} &\text{otherwise}.\\ 
\end{cases}
$$
Then, as stated by Theorem~\ref{thm:todd}, $P$ and $P'$ are equivalent.
Moreover, we have $P'_{:,i} = P'_{:,j}$ because $\bs z = P_{:,i} \oplus P_{:,j}$, therefore these two columns can be removed from $P'$, which results in $P'$ having at least one less column than $P$.

If there doesn't exist any $\bs y$ satisfying $\chi \bs y = \bs 0$ and $y_i \oplus y_j = 1$, then a different $\bs z$ vector can be considered.
In the worst case the \texttt{TODD} algorithm performs this search for all the vectors $\bs z$ in the set $\{P_{:,i} \oplus P_{:j} \mid 0 \leq i < j < n\}$, which contains $\mathcal{O}(m^2)$ elements.
The matrix $\chi$ is composed of $m$ columns and $\mathcal{O}(n^3)$ rows.
Performing a Gaussian elimination in order to find some $\bs y$ satisfying $\chi \bs y = \bs 0$ induces a complexity of $\mathcal{O}(n^3m^2)$.
This procedure must be performed for the $\mathcal{O}(m^2)$ $\bs z$ vectors evaluated.
And everytime a simplification is found the algorithm restart the same procedure on the new parity table $P' = P \oplus \bs z\bs y^T$, this happens at most $m$ times.
Thus, the overall complexity of the \texttt{TODD} algorithm is $\mathcal{O}(n^3m^5)$.
The important worst-case complexity of the \texttt{TODD} algorithm makes it rapidly impractical for instances of increasing size.
We will demonstrate how this complexity can be reduced to $\mathcal{O}(n^4m^3)$ by providing an alternative and more efficient way of checking whether or not there exists a vector $\bs y$ satisfying the conditions of Theorem~\ref{thm:todd}, and therefore avoiding the expensive Gaussian elimination required to solve the system $\chi \bs y = \bs 0$.

First, we can notice that Equation~\ref{eq:condition_3} of Theorem~\ref{thm:todd} can be greatly simplified in some cases, depending on the value of $\bs z$.
For instance, if $z_\alpha = 1$ for some $\alpha$ and $z_\beta = 0$ for all $\beta \neq \alpha$, then Equation~\ref{eq:condition_3} is equal to
\begin{equation}\label{eq:simplified_condition_3}
    \lvert z_\alpha(P_\beta \wedge P_\gamma) \wedge \bs y \rvert \equiv 0 \pmod{2}.
\end{equation}
In this specific case, a vector $\bs y$ satisfies Equations~\ref{eq:condition_2} and~\ref{eq:condition_3} of Theorem~\ref{thm:todd} if it satisfies $\chi \bs y = \bs 0$ where $\chi$ contains a row equal to Equation~\ref{eq:condition_2} for all $\alpha$ and a row equal to Equation~\ref{eq:simplified_condition_3} for all $\alpha, \beta, \gamma$ satisfying $\alpha \neq \beta \neq \gamma$, $\beta < \gamma$ and $z_\alpha = 1$.
In total, $\chi$ contains only $\mathcal{O}(n^2)$ rows instead of $\mathcal{O}(n^3)$ rows in the original $\texttt{TODD}$ algorithm, and so the nullspace of $\chi$ can be computed in $\mathcal{O}(n^2m^2)$ operations.
In the more general case where $\lvert \bs z \rvert > 1$, we can perform a change of basis to always end up in this specific case where $\lvert \bs z \rvert = 1$.
A change of basis induces a complexity of $\mathcal{O}(nm)$ because $\mathcal{O}(n)$ additions must be performed between the rows of $P$.
This must be done for all the $\mathcal{O}(m^2)$ $\bs z$ vectors evaluated, and the whole procedure is repeated $\mathcal{O}(m)$ times.
Thus, this version of the $\texttt{TODD}$ algorithm has a complexity of $\mathcal{O}(n^2m^5)$, which is better than the $\mathcal{O}(n^3m^5)$ complexity of the original $\texttt{TODD}$ algorithm.

Instead of performing a change of basis, we can rely on the following theorem to efficiently find the vectors $\bs z $ and $\bs y$ satisfying the Equations~\ref{eq:condition_2} and~\ref{eq:condition_3} of Theorem~\ref{thm:todd}.

\begin{theorem}\label{thm:fasttodd}
    Let $P$ be a parity table of size $n \times m$, and let $\bs z$ and $\bs y$ be vectors of size $n$ and $m$ respectively and such that $\lvert \bs y \rvert \equiv 0 \pmod{2}$.
    Let $L$ and $X$ be matrices with rows labelled by $(\alpha\beta)$ such that
    \begin{align}
        L_{\alpha\beta} &= P_\alpha \wedge P_\beta
    \end{align}
    and 
    \begin{equation}
        X_{\alpha\beta,\gamma} = z_\alpha \delta_{\beta\gamma} \oplus z_\beta \delta_{\alpha\gamma}
    \end{equation}
    for all $\alpha, \beta, \gamma$ satisfying $0 \leq \alpha \leq \beta < n$ and $0 \leq \gamma < n$, and where $\delta$ is the Kronecker delta defined as follows:
    \begin{equation}
        \delta_{\alpha\beta} =
        \begin{cases}
            0 &\text{if $\alpha \neq \beta$},\\
            1 &\text{if $\alpha = \beta$}.
        \end{cases}
    \end{equation}
    There exists $\bs y'$ such that $L\bs y \oplus X\bs y' = \bs 0$ if and only if the following conditions are satisfied:
    \begin{align}
        \lvert P_\alpha \wedge \bs y \rvert &\equiv 0 \pmod{2} \label{eq:condition_2_b} \\
        \lvert \left[ z_\alpha(P_\beta \wedge P_\gamma) \oplus z_\beta(P_\alpha \wedge P_\gamma) \oplus z_\gamma(P_\alpha \wedge P_\beta)\right] \wedge \bs y \rvert &\equiv 0 \pmod{2} \label{eq:condition_3_b}
    \end{align}
    for all $0 \leq \alpha \leq \beta \leq \gamma < n$.
\end{theorem}

\begin{proof}
    In the following we will use the labels $(\beta\alpha)$ and $(\alpha\beta)$ to refer to the same unique row, for instance $L_{\beta\alpha} = L_{\alpha\beta}$.
    For all $\alpha$ we have $X_{\alpha\alpha} = \bs 0$, which implies that
    \begin{equation}
    \begin{aligned}
        L_{\alpha\alpha} \bs y \oplus X_{\alpha\alpha}\bs y' &= L_{\alpha\alpha} \bs y \equiv \lvert P_\alpha \wedge \bs y \rvert \pmod{2} \\
    \end{aligned}
    \end{equation}
    for all $\alpha$.
    Therefore we have 
    \begin{equation}
        L_{\alpha\alpha}\bs y \oplus X_{\alpha\alpha}\bs y' = 0 \iff \lvert P_\alpha \wedge \bs y \rvert \equiv 0 \pmod{2}
    \end{equation}
    for all $\alpha$.
    It remains to show that there exists $\bs y'$ satisfying $L_{\alpha\beta}\bs y \oplus X_{\alpha\beta}\bs y' = 0$ for all $\alpha \neq \beta$ if and only if 
    \begin{equation}\label{eq:4.1}
        \lvert \left[ z_\alpha (P_\beta \wedge P_\gamma) \oplus z_\beta (P_\alpha \wedge P_\gamma) \oplus z_\gamma (P_\alpha \wedge P_\beta)\right] \wedge \bs y \rvert \equiv 0 \pmod{2}
    \end{equation}
    holds for all $0 \leq \alpha < \beta < \gamma < n$.

    In the case where $\lvert \bs z \rvert = 1$, let $\alpha$ be such that $z_\alpha = 1$.
    Then, for such vector $\bs z$, Equation~\ref{eq:4.1} is equivalent to
    \begin{equation}\label{eq:4.2}
        \lvert P_\beta \wedge P_\gamma \wedge \bs y \rvert \equiv 0 \pmod{2}
    \end{equation}
    for all $\beta, \gamma$ satisfying $0 \leq \beta < \gamma < n$ and $\alpha \neq \beta \neq \gamma$.
    We also have $X_{\beta\gamma}\bs y' = 0$ for any $\bs y'$ and for all $\beta, \gamma$ satisfying $\alpha \neq \beta \neq \gamma$ because $X_{\beta\gamma} = \bs 0$ for all $\beta, \gamma$ satisfying $z_\beta = 0, z_\gamma = 0$.
    We then have
    \begin{equation}
    \begin{aligned}
        L_{\beta\gamma}\bs y \oplus X_{\beta\gamma}\bs y' &= L_{\beta\gamma}\bs y \equiv \lvert P_\beta \wedge P_\gamma \wedge \bs y \rvert \pmod{2}
    \end{aligned}
    \end{equation}
    for all $\beta, \gamma$ satisfying $0 \leq \beta < \gamma < n$ and $\alpha \neq \beta \neq \gamma$.
    And therefore
    \begin{equation}
        L_{\beta\gamma}\bs y \oplus X_{\beta\gamma}\bs y' = \bs 0 \iff \lvert P_\beta \wedge P_\gamma \wedge \bs y \rvert \equiv 0 \pmod{2}
    \end{equation}
    for all $\beta, \gamma$ satisfying $0 \leq \beta < \gamma < n$ and $\alpha \neq \beta \neq \gamma$.
    Moreover, the definition of the matrix $X$ implies the following equality for all $\beta, \gamma$:
    \begin{align}\label{eq:4.4}
        X_{\beta\gamma}\bs y' = y'_\beta z_\gamma \oplus y'_\gamma z_\beta
    \end{align}
    which entails
    \begin{equation}\label{eq:4.3}
        L_{\alpha\beta}\bs y \oplus X_{\alpha\beta}\bs y' = L_{\alpha\beta}\bs y \oplus y'_\beta
    \end{equation}
    for all $\beta$ satisfying $\alpha \neq \beta$.
    Therefore, if $\bs y'$ satisfies $y'_\beta = L_{\alpha\beta}\bs y$ for all $\beta$ such that $\alpha \neq \beta$, then $L\bs y \oplus X\bs y' = \bs 0$ and so Theorem~\ref{thm:fasttodd} is true in the case where $\lvert \bs z \rvert = 1$.\

    Let $B$ be a full rank binary matrix of size $n\times n$ which represents a change of basis, and let $\tilde{L}$ and $\tilde{X}$ be matrices constructed in the same way as $L$ and $X$ but with respect to $BP$ and $B\bs z$.
    The proof of Theorem~\ref{thm:fasttodd} can be completed by proving the two following propositions.
    First, if $\bs y$ satisfies Equations~\ref{eq:condition_2_b} and~\ref{eq:condition_3_b} for $P$ and $\bs z$ then $\bs y$ also satisfies Equations~\ref{eq:condition_2_b} and~\ref{eq:condition_3_b} for $BP$ and $B\bs z$.
    Second, if there exists $\bs y'$ such that $L\bs y \oplus X\bs y' = \bs 0$ then there exists $\tilde{\bs y}'$ such that $\tilde{L}\bs y \oplus \tilde{X}\tilde{\bs y}' = \bs 0$.
    Indeed, if these two propositions are true, then if $\bs y$ satisfies Equations~\ref{eq:condition_2_b} and~\ref{eq:condition_3_b} for $P$ and $\bs z$ it also satisfies Equations~\ref{eq:condition_2_b} and~\ref{eq:condition_3_b} for $BP$ and $B\bs z$, where $B$ is chosen such that $\lvert B\bs z\rvert = 1$.
    Then, as previously demonstrated, there exists a vector $\tilde{\bs y}'$ such that $\tilde{L}\bs y \oplus \tilde{X}\tilde{\bs y}' = \bs 0$, which would imply that there exists a vector $\bs y'$ such that $L\bs y \oplus X\bs y' = \bs 0$.

    Let $\tilde{\bs z}$ and $\tilde{P}$ be such that $\tilde{z}_\alpha = z_\alpha \oplus z_\beta$ and $\tilde{P}_\alpha = P_\alpha \oplus P_\beta$ for some fixed $\alpha, \beta$ satisfying $\alpha \neq \beta$ and $\tilde{z}_\gamma = z_\gamma$, $\tilde{P}_\gamma = P_\gamma$ for all $\gamma$ such that $\gamma \neq \alpha$.
    If $\bs y$ satisfies Equations~\ref{eq:condition_2_b} and~\ref{eq:condition_3_b}, then we have
    \begin{equation}
    \begin{aligned}
        \lvert \tilde{P}_\alpha \wedge \bs y \rvert &= \lvert (P_\alpha \oplus P_\beta) \wedge \bs y \rvert \\
        &\equiv \lvert P_\alpha \wedge \bs y \rvert + \lvert P_\beta \wedge \bs y \rvert \pmod{2} \\
        &\equiv 0 \pmod{2}
    \end{aligned}
    \end{equation}
    and
    \begin{equation}
    \begin{aligned}
        &\lvert [\tilde{z}_\alpha(\tilde{P}_\beta \wedge \tilde{P}_\gamma) \oplus \tilde{z}_\beta(\tilde{P}_\alpha \wedge \tilde{P}_\gamma) \oplus \tilde{z}_\gamma(\tilde{P}_\alpha \wedge \tilde{P}_\beta)] \wedge \bs y \rvert \\
        &= \lvert \left[ (z_\alpha \oplus z_\beta)(P_\beta \wedge P_\gamma) \oplus z_\beta((P_\alpha \oplus P_\beta) \wedge P_\gamma) \oplus z_\gamma((P_\alpha \oplus P_\beta) \wedge P_\beta)\right] \wedge \bs y \rvert \\
        &= \lvert \left[ z_\alpha(P_\beta \wedge P_\gamma) \oplus z_\beta(P_\beta \wedge P_\gamma) \oplus z_\beta(P_\alpha \wedge P_\gamma) \oplus z_\beta(P_\beta \wedge P_\gamma) \oplus z_\gamma(P_\alpha \wedge P_\beta) \oplus z_\gamma P_\beta \right] \wedge \bs y \rvert \\
        &\equiv \lvert \left[ z_\alpha(P_\beta \wedge P_\gamma) \oplus z_\beta(P_\alpha \wedge P_\gamma) \oplus z_\gamma(P_\alpha \wedge P_\beta) \right] \wedge \bs y \rvert + z_\gamma\lvert P_\beta \wedge \bs y \rvert \pmod{2} \\
        &\equiv 0 \pmod{2}
    \end{aligned}
    \end{equation}
    and so $\bs y$ also satisfies Equations~\ref{eq:condition_2_b} and~\ref{eq:condition_3_b} for $\tilde{P}, \tilde{\bs z}$.

    Let $\tilde{L}$ and $\tilde{X}$ be matrices constructed in the same way as $L$ and $X$ but with respect to $\tilde{P}$ and $\tilde{\bs z}$.
    Let's assume that $L\bs y \oplus X\bs y' = \bs 0$, we will show that it implies $\tilde{L}\bs y \oplus \tilde{X}\tilde{\bs y}'=\bs 0$, where $\tilde{\bs y}'$ satisfies $\tilde{y}_\alpha' = y_\alpha' \oplus y_\beta'$ and $\tilde{y}_\gamma' = y_\gamma'$ for all $\gamma$ such that $\gamma \neq \alpha$.
    By using Equation~\ref{eq:4.4} we can deduce that
    \begin{equation}
        \tilde{X}_{\gamma\gamma'}\tilde{\bs y}' = \tilde{y}'_\gamma \tilde{z}_{\gamma'} \oplus \tilde{y}'_{\gamma'} \tilde{z}_{\gamma} = y'_\gamma z_{\gamma'} \oplus y'_{\gamma'} z_{\gamma} = X_{\gamma\gamma'} \bs y'
    \end{equation}
    for all $\gamma, \gamma'$ satisfying $\gamma \neq \alpha$ and $\gamma' \neq \alpha$.
    This entails
    \begin{equation}
        \tilde{L}_{\gamma\gamma'}\bs y \oplus \tilde{X}_{\gamma\gamma'}\tilde{\bs y}' = L_{\gamma\gamma'}\bs y \oplus X_{\gamma\gamma'}\bs y' = 0
    \end{equation}
    for all $\gamma, \gamma'$ satisfying $\gamma \neq \alpha$ and $\gamma' \neq \alpha$.
    Furthermore, for all $\gamma$ such that $\gamma \neq \alpha$, we have
    \begin{equation}
    \begin{aligned}
        \tilde{L}_{\alpha\gamma}\bs y &\equiv \lvert \tilde{P}_\alpha \wedge \tilde{P}_\gamma \wedge \bs y \rvert \pmod{2} \\
        &\equiv \lvert (P_\alpha \oplus P_\beta) \wedge P_\gamma \wedge \bs y \rvert \pmod{2} \\
        &\equiv \lvert P_\alpha \wedge P_\gamma\wedge \bs y \rvert + \lvert P_\beta \wedge P_\gamma \wedge \bs y \rvert \pmod{2} \\
        &\equiv L_{\alpha\gamma}\bs y + L_{\beta\gamma}\bs y \pmod{2} \\
        &\equiv X_{\alpha\gamma}\bs y' + X_{\beta\gamma}\bs y' \pmod{2} \\
    \end{aligned}
    \end{equation}
    which entails
    \begin{equation}
    \begin{aligned}
        \tilde{L}_{\alpha\gamma}\bs y \oplus \tilde{X}_{\alpha\gamma}\tilde{\bs y}' &= X_{\alpha\gamma}\bs y' \oplus X_{\beta\gamma}\bs y' \oplus \tilde{X}_{\alpha\gamma}\tilde{\bs y}' \\
        &= y'_\alpha z_\gamma \oplus y'_\gamma z_\alpha \oplus y'_\beta z_\gamma \oplus y'_\gamma z_\beta \oplus \tilde{y}_\alpha' \tilde{z}_\gamma \oplus \tilde{y}_\gamma' \tilde{z}_\alpha \\
        &= y'_\alpha z_\gamma \oplus y'_\gamma z_\alpha \oplus y'_\beta z_\gamma \oplus y'_\gamma z_\beta \oplus (y'_\alpha \oplus y'_\beta)z_\gamma \oplus y_\gamma' (z_\alpha \oplus z_\beta) \\
        &= 0
    \end{aligned}
    \end{equation}
    by relying on Equation~\ref{eq:4.4}.
    Finally, we have
    \begin{equation}
    \begin{aligned}
        \tilde{L}_{\alpha\alpha}\bs y \oplus \tilde{X}_{\alpha\alpha}\tilde{\bs y}'
        &= \tilde{L}_{\alpha\alpha}\bs y \oplus \tilde{y}'_\alpha \tilde{z}_\alpha \oplus \tilde{y}'_\alpha \tilde{z}_\alpha \\
        &\equiv \lvert \tilde{P}_\alpha \wedge \bs y \rvert \pmod{2} \\
        &\equiv \lvert (P_\alpha \oplus P_\beta) \wedge \bs y \rvert \pmod{2} \\
        &\equiv \lvert P_\alpha \wedge \bs y \rvert + \lvert P_\beta \wedge \bs y \rvert \pmod{2} \\
        &\equiv 0 \pmod{2}.
    \end{aligned}
    \end{equation}
    Thus,
    \begin{equation}
        L\bs y \oplus X\bs y'=\bs 0 \implies \tilde{L}\bs y \oplus \tilde{X}\tilde{\bs y}'=\bs 0
    \end{equation}
    which concludes the proof of Theorem~\ref{thm:fasttodd}.
\end{proof}

Theorem~\ref{thm:fasttodd} can be exploited to design an algorithm equivalent to the \texttt{TODD} algorithm, but which has a lower complexity.
Before introducing this algorithm, we demonstrate how the key mechanism utilized by the \texttt{TODD} algorithm to reduce the $T$-count can be improved.
The conditions given by Equations~\ref{eq:condition_2} and~\ref{eq:condition_3} of Theorem~\ref{thm:todd} are sufficient for Equation~\ref{eq:p_equality} to hold but they are not necessary.
The following theorem gives necessary and sufficient conditions for Equation~\ref{eq:p_equality} to be satisfied.

\begin{theorem}\label{thm:todd_iff}
    Let $P$ be a parity table of size $n \times m$ and $P' = P \oplus \bs z \bs y^T$ where $\bs z$ and $\bs y$ are vectors of size $n$ and $m$ respectively and such that $\lvert \bs y \rvert \equiv 0 \pmod{2}$.
    And let $L$ and $X$ be matrices and $\bs v$ be a vector, all with rows labelled by $(\alpha\beta)$ such that
    \begin{align}
        L_{\alpha\beta} &= P_\alpha \wedge P_\beta \\
        X_{\alpha\beta,\gamma} &= z_\alpha \delta_{\beta\gamma} \oplus z_\beta \delta_{\alpha\gamma} \\
        v_{\alpha\beta} &= z_\alpha \wedge z_\beta 
    \end{align}
    for all $\alpha, \beta, \gamma$ satisfying $0 \leq \alpha \leq \beta < n$ and $0 \leq \gamma < n$, and where $\delta$ is the Kronecker delta defined as follows:
    \begin{equation}
        \delta_{\alpha\beta} =
        \begin{cases}
            0 &\text{if $\alpha \neq \beta$},\\
            1 &\text{if $\alpha = \beta$}.
        \end{cases}
    \end{equation}
    Then, the equality
    \begin{equation}\label{eq:iff_0}
        \lvert P'_\alpha \wedge P'_\beta \wedge P'_\gamma \rvert \equiv \lvert P_\alpha \wedge P_\beta \wedge P_\gamma \rvert \pmod{2}
    \end{equation}
    holds for all $0 \leq \alpha \leq \beta \leq \gamma < n$ if and only if there exists a vector $\bs y'$ and a Boolean $b$ such that $L\bs y \oplus X\bs y' \oplus b\bs v = \bs 0$.
\end{theorem}

\begin{proof}
    In the same way as in the proof of Theorem~\ref{thm:fasttodd}, we will use the labels $(\beta\alpha)$ and $(\alpha\beta)$ to refer to the same unique row, for instance $L_{\beta\alpha} = L_{\alpha\beta}$.
    Analogously to Equation~\ref{eq:tohpe_proof}, we have
    \begin{equation}
    \begin{aligned}\label{eq:iff_1}
        \lvert P'_\alpha \wedge P'_\beta \wedge P'_\gamma \rvert &\equiv \lvert P_\alpha \wedge P_\beta \wedge P_\gamma \rvert + z_\gamma\lvert P_\alpha \wedge P_\beta \wedge \bs y \rvert + z_\beta\lvert P_\alpha \wedge P_\gamma \wedge \bs y \rvert + z_\alpha\lvert P_\beta \wedge P_\gamma \wedge \bs y \rvert \\
          & \quad + z_\beta z_\gamma\lvert P_\alpha \wedge \bs y \rvert + z_\alpha z_\gamma\lvert P_\beta \wedge \bs y \rvert + z_\alpha z_\beta\lvert P_\gamma \wedge \bs y \rvert + z_\alpha z_\beta z_\gamma\lvert \bs y \rvert \pmod{2}
    \end{aligned}
    \end{equation}
    for all $\alpha, \beta, \gamma$ satisfying $0 \leq \alpha \leq \beta \leq \gamma < n$.
    Because $\lvert \bs y \rvert \equiv 0 \pmod{2}$, we can deduce from Equation~\ref{eq:iff_1} that Equations~\ref{eq:iff_0} holds if and only if the following equation is satisfied
    \begin{equation}
    \begin{aligned}\label{eq:iff_2}
        \lvert \left[ z_\gamma (P_\alpha \wedge P_\beta) \oplus z_\beta (P_\alpha \wedge P_\gamma) \oplus z_\alpha (P_\beta \wedge P_\gamma) \oplus z_\beta z_\gamma P_\alpha \oplus z_\alpha z_\gamma P_\beta \oplus z_\alpha z_\beta P_\gamma \right] \wedge \bs y \rvert \equiv 0 \pmod{2}
    \end{aligned}
    \end{equation}
    for all $\alpha, \beta, \gamma$ satisfying $0 \leq \alpha \leq \beta \leq \gamma < n$.

    Let $\bs z$ be such that $\lvert \bs z \rvert = 1$, and let $\alpha$ be such that $z_\alpha = 1$.
    Then, for such vector $\bs z$, Equation~\ref{eq:iff_2} can be rewritten as
    \begin{equation}\label{eq:iff_5}
    \begin{aligned}
        \lvert P_\beta \wedge P_\gamma \wedge \bs y \rvert \equiv 0 \pmod{2}
    \end{aligned}
    \end{equation}
    where $\beta, \gamma$ are satisfying $\beta \neq \alpha$, $\gamma \neq \alpha$, and as
    \begin{equation}\label{eq:iff_4}
    \begin{aligned}
        \lvert P_\gamma \wedge \bs y \rvert \equiv 0 \pmod{2}
    \end{aligned}
    \end{equation}
    where $\beta, \gamma$ are satisfying $\beta = \alpha$ and $\gamma \neq \alpha$.
    Note that Equation~\ref{eq:iff_2} is necessarily satisfied in the case where $\alpha = \beta = \gamma$, and that if Equation~\ref{eq:iff_5} is satisfied then Equation~\ref{eq:iff_4} is also satisfied.
    Therefore, proving Theorem~\ref{thm:todd_iff} in the case where $\lvert\bs z\rvert = 1$ can then be done by showing that there exists a vector $\bs y'$ and a Boolean $b$ such that $L\bs y \oplus X\bs y' \oplus b\bs v = \bs 0$ if and only if Equation~\ref{eq:iff_5} is satisfied.
    By definition we have $X_{\beta\gamma} = \bs 0$ and $v_{\beta\gamma} = 0$ for all $\beta, \gamma$ satisfying $\beta \neq \alpha$ and $\gamma \neq \alpha$.
    Hence, 
    \begin{equation}
        L_{\beta\gamma}\bs y \oplus X_{\beta\gamma}\bs y' \oplus b v_{\beta\gamma} = L_{\beta\gamma}\bs y = \lvert P_\beta \wedge P_\gamma \wedge \bs y\rvert \pmod{2}
    \end{equation}
    for all $\beta, \gamma$ satisfying $\beta \neq \alpha$ and $\gamma \neq \alpha$,
    and so if $L\bs y \oplus X\bs y' \oplus b\bs v = \bs 0$ then $L_{\beta\gamma}\bs y = 0$ which imply that Equation~\ref{eq:iff_5} is satisfied.
    Conversely, if Equation~\ref{eq:iff_5} is satisfied then we must have 
    \begin{equation}
        L_{\beta\gamma}\bs y \oplus X_{\beta\gamma}\bs y' \oplus b v_{\beta\gamma} = L_{\beta\gamma}\bs y = 0
    \end{equation}
    for all $\beta, \gamma$ satisfying $\beta \neq \alpha$ and $\gamma \neq \alpha$ and for any vector $\bs y'$ and Boolean $b$.
    Moreover we have
    \begin{equation}
        L_{\alpha\alpha}\bs y \oplus X_{\alpha\alpha}\bs y' \oplus b v_{\alpha\alpha} = L_{\alpha\alpha}\bs y \oplus b
    \end{equation}
    because $X_{\alpha\alpha} = \bs 0$ and $v_{\alpha\alpha} = 1$.
    And 
    \begin{equation}
        L_{\alpha\beta}\bs y \oplus X_{\alpha\beta}\bs y' \oplus b v_{\alpha\beta} = L_{\alpha\beta}\bs y \oplus y'_{\beta}
    \end{equation}
    where $\beta \neq \alpha$ by using Equation~\ref{eq:4.4} and because $v_{\alpha\beta} = 0$.
    Let $\bs y'$ be such that $y'_\beta = L_{\alpha\beta}\bs y$ for all $\beta \neq \alpha$ and $b$ be such that $b = L_{\alpha\alpha}\bs y$, then we have $L\bs y \oplus X\bs y' \oplus b\bs v = \bs 0$.
    Thus, we proved that
    \begin{equation}
        L\bs y \oplus X\bs y' \oplus b\bs v = \bs 0 \iff \lvert P_\beta \wedge P_\gamma \wedge \bs y \rvert \equiv 0 \pmod{2}
    \end{equation}
    for all $\beta, \gamma$ satisfying $\beta \neq \alpha$ and $\gamma \neq \alpha$, and so Theorem~\ref{thm:todd_iff} is true in the case where $\lvert \bs z \rvert = 1$.

    Let $B$ be a full rank binary matrix of size $n\times n$ which represents a change of basis, and let $\tilde{L}$, $\tilde{X}$, $\tilde{\bs v}$ be constructed in the same way as $L$, $X$ and $\bs v$ but with respect to $BP$ and $B\bs z$.
    The proof of Theorem~\ref{thm:todd_iff} can be completed by proving that if there exists a vector $\bs y'$ and a Boolean $b$ such that $L\bs y \oplus X\bs y' \oplus b \bs v = \bs 0$ then there exists a vector $\tilde{\bs y}'$ and a Boolean $\tilde{b}$ such that $\tilde{L}\bs y \oplus \tilde{X}\tilde{\bs y}' \oplus \tilde{b}\tilde{\bs v} = \bs 0$.
    Suppose this proposition to be true and let $B$ be such that $\lvert B\bs z\rvert = 1$.
    Then, as previously demonstrated, Equation~\ref{eq:iff_0} holds for $BP$ and $B\bs z$ if and only if there exists a vector $\tilde{\bs y}'$ and a Boolean $\tilde{b}$ such that $\tilde{L}\bs y \oplus \tilde{X}\tilde{\bs y}' \oplus \tilde{b}\tilde{\bs v} = \bs 0$, which would imply that there exists a vector $\bs y'$ and a Boolean $b$ such that $L\bs y \oplus X\bs y' \oplus b\bs v = \bs 0$ if and only if Equation~\ref{eq:iff_0} holds for $P$ and $\bs z$.

    Let $\tilde{\bs z}$ and $\tilde{P}$ be such that $\tilde{z}_\alpha = z_\alpha \oplus z_\beta$ and $\tilde{P}_\alpha = P_\alpha \oplus P_\beta$ for some fixed $\alpha, \beta$ satisfying $\alpha \neq \beta$ and $\tilde{z}_\gamma = z_\gamma$, $\tilde{P}_\gamma = P_\gamma$ for all $\gamma$ such that $\gamma \neq \alpha$.
    Let $\tilde{L}$, $\tilde{X}$ and $\tilde{\bs v}$ be constructed in the same way as $L$, $X$ and $\bs v$ but with respect to $\tilde{P}$ and $\tilde{\bs z}$.
    Let's assume that $L\bs y \oplus X\bs y' \oplus b\bs v = \bs 0$, we will show that it implies $\tilde{L}\bs y \oplus \tilde{X} \tilde{\bs y}' \oplus b \tilde{\bs v} = \bs 0$, where $\tilde{\bs y}'$ satisfies $\tilde{y}'_\alpha = y'_\alpha \oplus y'_\beta$ and $\tilde{y}'_\gamma = y'_\gamma$ for all $\gamma$ such that $\gamma \neq \alpha$.
    By using Equation~\ref{eq:4.4} we can deduce that
    \begin{equation}\label{eq:iff_6}
        \tilde{X}_{\gamma\gamma'}\tilde{\bs y}' = \tilde{y}'_\gamma \tilde{z}_{\gamma'} \oplus \tilde{y}'_{\gamma'} \tilde{z}_{\gamma} = y'_\gamma z_{\gamma'} \oplus y'_{\gamma'} z_{\gamma} = X_{\gamma\gamma'} \bs y'
    \end{equation}
    for all $\gamma, \gamma'$ satisfying $\gamma \neq \alpha$ and $\gamma' \neq \alpha$.
    And from the definitions of the vectors $\tilde{\bs v}$ and $\bs v$, we can deduce that
    \begin{equation}\label{eq:iff_7}
        \tilde{v}_{\gamma\gamma'} = \tilde{z}_\gamma \wedge \tilde{z}_\gamma' = z_\gamma \wedge z_\gamma' = v_{\gamma\gamma'}
    \end{equation}
    for all $\gamma, \gamma'$ satisfying $\gamma \neq \alpha$ and $\gamma' \neq \alpha$.
    Equations~\ref{eq:iff_6} and~\ref{eq:iff_7} imply that
    \begin{equation}
        \tilde{L}_{\gamma\gamma'}\bs y \oplus \tilde{X}_{\gamma\gamma'} \tilde{\bs y}' \oplus b\tilde{v}_{\gamma\gamma'} = L_{\gamma\gamma'}\bs y \oplus X_{\gamma\gamma'}\bs y' \oplus b v_{\gamma\gamma'}= 0
    \end{equation}
    for all $\gamma, \gamma'$ satisfying $\gamma \neq \alpha$ and $\gamma' \neq \alpha$.
    Furthermore, for all $\gamma$ such that $\gamma \neq \alpha$, we have
    \begin{equation}
    \begin{aligned}
        \tilde{L}_{\alpha\gamma}\bs y &\equiv \lvert \tilde{P}_\alpha \wedge \tilde{P}_\gamma \wedge \bs y \rvert \pmod{2} \\
        &\equiv \lvert (P_\alpha \oplus P_\beta) \wedge P_\gamma \wedge \bs y \rvert \pmod{2} \\
        &\equiv \lvert P_\alpha \wedge P_\gamma\wedge \bs y \rvert + \lvert P_\beta \wedge P_\gamma \wedge \bs y \rvert \pmod{2} \\
        &\equiv L_{\alpha\gamma}\bs y + L_{\beta\gamma}\bs y \pmod{2} \\
        &\equiv X_{\alpha\gamma}\bs y' + b v_{\alpha\gamma} + X_{\beta\gamma}\bs y' + b v_{\beta\gamma} \pmod{2}
    \end{aligned}
    \end{equation}
    which entails
    \begin{equation}
    \begin{aligned}
        &\tilde{L}_{\alpha\gamma}\bs y \oplus \tilde{X}_{\alpha\gamma}\tilde{\bs y}' \oplus b \tilde{v}_{\alpha\gamma} \\
        &= X_{\alpha\gamma}\bs y' \oplus X_{\beta\gamma}\bs y' \oplus \tilde{X}_{\alpha\gamma}\tilde{\bs y}'\oplus b v_{\alpha\gamma} \oplus b v_{\beta\gamma} \oplus b \tilde{v}_{\alpha\gamma} \\
        &= y'_\alpha z_\gamma \oplus y'_\gamma z_\alpha \oplus y'_\beta z_\gamma \oplus y'_\gamma z_\beta \oplus \tilde{y}_\alpha' \tilde{z}_\gamma \oplus \tilde{y}_\gamma' \tilde{z}_\alpha \oplus b (z_\alpha \wedge z_\gamma) \oplus b (z_\beta \wedge z_\gamma) \oplus b (\tilde{z}_\alpha \wedge z_\gamma) \\
        &= y'_\alpha z_\gamma \oplus y'_\gamma z_\alpha \oplus y'_\beta z_\gamma \oplus y'_\gamma z_\beta \oplus (y'_\alpha \oplus y'_\beta)z_\gamma \oplus y_\gamma' (z_\alpha \oplus z_\beta)\\
        & \quad \oplus b (z_\alpha \wedge z_\gamma) \oplus b (z_\beta \wedge z_\gamma) \oplus b ((z_\alpha \oplus z_\beta) \wedge z_\gamma) \\
        &= 0
    \end{aligned}
    \end{equation}
    by relying on Equation~\ref{eq:4.4}.
    Finally, we have
    \begin{equation}
    \begin{aligned}
        \tilde{L}_{\alpha\alpha}\bs y \oplus \tilde{X}_{\alpha\alpha}\tilde{\bs y}' \oplus b\tilde{v}_{\alpha\alpha}
        &= \tilde{L}_{\alpha\alpha'}\bs y \oplus \tilde{y}'_\alpha \tilde{z}_\alpha \oplus \tilde{y}'_\alpha \tilde{z}_\alpha \oplus b\tilde{z}_\alpha \\
        &= \tilde{L}_{\alpha\alpha'}\bs y \oplus b(z_\alpha \oplus z_\beta) \\
        &\equiv \lvert \tilde{P}_\alpha \wedge \bs y \rvert + bv_{\alpha\alpha} + bv_{\beta\beta} \pmod{2} \\
        &\equiv \lvert (P_\alpha \oplus P_\beta) \wedge \bs y \rvert + bv_{\alpha\alpha} + bv_{\beta\beta}  \pmod{2} \\
        &\equiv \lvert P_\alpha \wedge \bs y \rvert + \lvert P_\beta \wedge \bs y \rvert + bv_{\alpha\alpha} + bv_{\beta\beta}  \pmod{2} \\
        &= L_{\alpha\alpha} \oplus bv_{\alpha\alpha} \oplus L_{\beta\beta} \oplus bv_{\beta\beta} \\
        &= 0 .
    \end{aligned}
    \end{equation}
    Thus,
    \begin{equation}
        L\bs y \oplus X\bs y' \oplus b \bs v = \bs 0 \implies \tilde{L}\bs y \oplus \tilde{X}\tilde{\bs y}' \oplus b \tilde{v} = \bs 0
    \end{equation}
    which concludes the proof of Theorem~\ref{thm:todd_iff}.
\end{proof}

Consider the algorithm whose pseudo-code is given in Algorithm~\ref{alg:fast_todd} and which is built upon Theorem~\ref{thm:todd_iff}.
The algorithm starts by performing a call to the \texttt{TOHPE} algorithm (Algorithm~\ref{alg:tohpe}).
The reason for calling the \texttt{TOHPE} algorithm is that it is more efficient to first search for a vector $\bs y$ satisfying $L\bs y = \bs 0$ and exploit Theorem~\ref{thm:tohpe} than to search for vectors $\bs y$, $\bs y'$ and a Boolean $b$ satisfying $L\bs y \oplus X\bs y' \oplus b \bs v = \bs 0$, where $L, X, \bs v$ are defined as in Theorem~\ref{thm:todd_iff}.
The algorithm then proceeds by creating the matrix $L$ and the set $Z$ which contains all the vectors $\bs z$ that can potentially be used to transform $P$ via Theorem~\ref{thm:todd_iff} and reduce its number of columns.
For each vector $\bs z \in Z$, the algorithm constructs the set $S^{(\bs z)}$ of pairs of indices $\{i, j\}$ satisfying $P_{:,i} \oplus P_{:,j} = \bs z$ in the case where $i\neq j$ or satisfying $P_{:,i} = \bs z$ in the case where $i = j$, as well as the matrix $X^{(\bs z)}$ and the vector $\bs v^{(\bs z)}$ which are constructed in the same way as in Theorem~\ref{thm:todd_iff}.
The set of vectors $\bs y$ satisfying $L\bs y \oplus X^{(\bs z)}\bs y' \oplus b \bs v^{(\bs z)} = \bs 0$ for some vector $\bs y'$ and some Boolean $b$ are stored in $N^{(\bs z)}$.
The algorithm then computes the pair of vectors $\bs z, \bs y$ where $\bs z \in Z$ and $\bs y \in N^{(\bs z)}$ for which the number of duplicated and all-zero columns in $P'$ is maximized, where $P'$ is defined as follows:
\begin{equation}
P' =
\begin{cases}
    P \oplus \bs z \bs y^T &\text{if $\lvert \bs y \rvert \equiv 0 \pmod{2}$},\\ 
    \begin{bmatrix} P \oplus \bs z \bs y^T & \bs z \end{bmatrix} &\text{otherwise}.\\ 
\end{cases}
\end{equation}
This number is given by the following objective function:
\begin{equation}
-\lvert \bs y \rvert \pmod{2} + \sum_{\{i, j\} \in S^{(\bs z)}} 2 (y_i \oplus y_j) + \delta_{ij}\left[y_i + 2(y_i \oplus 1)(\lvert \bs y \rvert \pmod{2})\right]
\end{equation}
which is the same as the one used in Algorithm~\ref{alg:tohpe}.
After computing the vectors $\bs z$ and $\bs y$ maximizing this function, Algorithm~\ref{alg:fast_todd} apply the transformation on the parity table and a recursive call is performed if the number of columns in the parity table has successfully been reduced, otherwise the parity table is returned.

\begin{algorithm}[t]
    \caption{Faster version of the third order duplicate and destroy algorithm}
    \label{alg:fast_todd}
	\SetAlgoLined
	\SetArgSty{textnormal}
	\SetKwFunction{proc}{FastTODD}
	\SetKwInput{KwInput}{Input}
	\SetKwInput{KwOutput}{Output}
    \KwInput{A parity table $P$ of size $n \times m$.}
    \KwOutput{An equivalent parity table with an optimized number of columns.}
	\SetKwProg{Fn}{procedure}{}{}
    \Fn{\proc{$P$}}{
        $P \leftarrow \texttt{TOHPE}(P)$ \\
        $Z \leftarrow \{P_{:,i} \oplus P_{:,j} \mid 0 \leq i < j < m \} \cup \{P_{:,i} \mid 0 \leq i < m \}$ \\
        $L \leftarrow$ matrix whose rows are forming the set $\{P_i \wedge P_j \mid 0 \leq i \leq j < n\}$ labelled by $ij$\\
        \ForAll{$\bs z \in Z$} {
            $S^{(\bs z)} \leftarrow \{\{i, j\} \mid P_{:,i} \oplus P_{:,j} = \bs z\} \cup \{\{i, i\} \mid P_{:,i} = \bs z\}$ \\
            $X^{(\bs z)} \leftarrow$ matrix such that $X^{(\bs z)}_{\alpha\beta, \gamma} = z_\alpha \delta_{\beta\gamma} \oplus z_\beta \delta_{\alpha\gamma}$ \\
            $\bs v^{(\bs z)} \leftarrow$ vector such that $v_{\alpha\beta} = z_\alpha \wedge z_\beta$ \\
            $N^{(\bs z)} \leftarrow$ generators of the set $\{\bs y \mid L\bs y \oplus X^{(\bs z)}\bs y' \oplus b \bs v^{(\bs z)} = \bs 0\}$
        }
        $\displaystyle \bs z, \bs y \leftarrow \!\!\!\argmax_{\bs z \in Z, \bs y \in N^{(\bs z)}} \set[\Big]{-\lvert \bs y \rvert \pmod{2} +\!\!\!\sum_{\{i, j\} \in S^{(\bs z)}} 2 (y_i \oplus y_j) + \delta_{ij}\left[y_i + 2(y_i \oplus 1)(\lvert \bs y \rvert \pmod{2})\right]}$ \\

        $P' \leftarrow P \oplus \bs z \bs y^T$ \\
        \If{$\lvert \bs y \rvert \equiv 1 \pmod{2}$}{
            $P' \leftarrow \begin{bmatrix} P' & \bs z \end{bmatrix}$ \\
        }
        $P' \leftarrow P'$ with all its duplicated and all-zero columns removed \\
        \If{$\text{size}(P') < \text{size}(P)$}{
            \Return $\texttt{FastTODD}(P')$ \\
        }
        \Return $P$ \\
	}
\end{algorithm}

\paragraph{Complexity analysis.}
Performing a call to the \texttt{TOHPE} algorithm induces a complexity of $\mathcal{O}(n^2m^3)$.
The set $Z$ and the sets $S^{(\bs z)}$ can be created with $\mathcal{O}(nm^2)$ operations.
We can take advantage of the fact that $L$ doesn't depend on $\bs z$ to compute $N^{(\bs z)}$ efficiently.
Let $\tilde{L}$ be the reduced column echelon form of $L$ and let $B$ be the matrix such that $\tilde{L}B = L$.
Let $\tilde{X}^{(\bs z)}$ be such that
\begin{equation}
    \tilde{X}^{(\bs z)}_{:,\ell} = X_{:,\ell} \bigoplus_{j \in J^{(\ell)}} \tilde{L}_{:,j}
\end{equation}
where $J^{(\ell)}$ is defined as follows:
\begin{equation}
    J^{(\ell)} = \{j \mid \exists i \text{ such that } X^{(\bs z)}_{i,\ell} = 1, \tilde{L}_{i, j} = 1, \tilde{L}_{i, k} = 0, \forall k \neq j \}.
\end{equation}
And let $\tilde{\bs v}^{(\bs z)}$ be such that
\begin{equation}
    \tilde{\bs v}^{(\bs z)} = \bs v \bigoplus_{j \in J} \tilde{L}_{:,j}
\end{equation}
where $J$ is defined as follows:
\begin{equation}
    J = \{j \mid \exists i \text{ such that } v_i = 1, \tilde{L}_{i, j} = 1, \tilde{L}_{i, k} = 0, \forall k \neq j \}.
\end{equation}
If there exists no vector $\bs y$ such that $\bs y \neq \bs 0$ and $L \bs y = \bs 0$ (which is the case after having executed the \texttt{TOHPE} algorithm, except potentially for $\bs y = \bs 1$ which can be ignored as it cannot be used to reduce the number of columns in the parity table), then there exists a vector $\bs y \neq \bs 0$, a vector $\bs y'$ and a Boolean $b$ such that $L\bs y \oplus X^{(\bs z)}\bs y' \oplus b \bs v^{(\bs z)} = \bs 0$ if and only if there exists a vector $\tilde{\bs y}'$ and a bolean $\tilde{b}$ such that $\tilde{X}^{(\bs z)} \tilde{\bs y}' \oplus \tilde{b} \tilde{\bs v}^{(\bs z)} = \bs 0$.
If the equation $\tilde{X}^{(\bs z)} \tilde{\bs y}' \oplus \tilde{b} \tilde{\bs v}^{(\bs z)} = \bs 0$ is satisfied then the associated vector $\bs y$ is equal to 
\begin{equation}\label{eq:iff_8}
    \bs y = \begin{bmatrix} \tilde{B}^{(X)} & \tilde{B}^{(\bs v)} \end{bmatrix}\begin{bmatrix} \tilde{\bs y}' \\ \tilde{b} \end{bmatrix}
\end{equation}
where $\tilde{B}^{(X)}$ is defined as follows:
\begin{equation}
    \tilde{B}^{(X)}_{:,\ell} = \bigoplus_{j \in J^{(\ell)}} \tilde{L}_{:,j}
\end{equation}
and $\tilde{B}^{(\bs v)}$ is defined as follows:
\begin{equation}
    \tilde{B}^{(\bs v)} = \bigoplus_{j \in J} \tilde{L}_{:,j}.
\end{equation}

Constructing the matrix $\tilde{X}^{(\bs z)}$ can be done in $\mathcal{O}(n^4)$ operations because $X^{(\bs z)}$ has $n$ columns and $\mathcal{O}(n^2)$ rows, and $\lvert X^{(\bs z)}_{:,\ell} \rvert < n$ which imply that the set $J^{(\ell)}$ contains no more than $n$ elements.
Constructing the vector $\tilde{\bs v}^{(\bs z)}$ can also be done in $\mathcal{O}(n^4)$ operations because $\bs v^{(\bs z)}$ has $\mathcal{O}(n^2)$ rows and the set $J$ contains $\mathcal{O}(n^2)$ elements.
Performing a Gaussian elimination on the matrix $\begin{bmatrix} \tilde{X}^{(\bs z)} & \tilde{\bs v}^{(\bs z)} \end{bmatrix}$ to compute $\tilde{\bs y}'$ and $\tilde{b}$ requires $\mathcal{O}(n^4)$ operations because the matrix has $\mathcal{O}(n)$ columns and $\mathcal{O}(n^2)$ rows.
The matrix $\tilde{B}^{(X)}$ and the vector $\tilde{B}^{(\bs v)}$ can be computed with the same computational complexity as the matrix $\tilde{X}^{(\bs z)}$ and the vector $\tilde{\bs v}^{(\bs z)}$.
The right nullspace of $\begin{bmatrix} \tilde{X}^{(\bs z)} & \tilde{\bs v}^{(\bs z)} \end{bmatrix}$ is generated by $\mathcal{O}(n)$ pairs of vector and Boolean $\tilde{\bs y}', \tilde{b}$.
For each one of these pairs, the associated vector $\bs y \in N^{(\bs z)}$ for which there exists a vector $\bs y'$ and a Boolean $b$ satisfying $L\bs y \oplus X^{(\bs z)}\bs y' \oplus b \bs v^{(\bs z)} = \bs 0$ can be computed in $\mathcal{O}(n^3)$ by using Equation~\ref{eq:iff_8}, for a total complexity of $\mathcal{O}(n^4)$.
This must be done for all the $\mathcal{O}(m^2)$ vectors $\bs z$, which implies a complexity of $\mathcal{O}(n^4m^2)$ for computing all the sets $N^{(\bs z)}$.
The vectors $\bs z$ and $\bs y$ satisfying the argmax function can be computed in $\mathcal{O}(n^3m^2)$ operations because the union of all the $S^{(\bs z)}$ sets contains $\mathcal{O}(m^2)$ elements, computing the value $\lvert \bs y \rvert \pmod{2}$ requires $\mathcal{O}(n^2)$ operations and there are $\mathcal{O}(nm^2)$ vectors $\bs y$.
Updating the parity table $P$ induces $\mathcal{O}(nm)$ operations.
Finally, the algorithm is performing no more than $m$ recursive calls.
Thus, the overall worst-case complexity of Algorithm~\ref{alg:fast_todd} is $\mathcal{O}(n^4m^3)$.

This complexity is significantly better than the $\mathcal{O}(n^3m^5)$ complexity of the original \texttt{TODD} algorithm.
Note that $n \leq m$, otherwise the instance of the problem can be simplified by eliminating a row of the parity table.
Indeed, if $n > m$, then a change of basis can be performed onto the parity table $P$ so that there exists a row $\alpha$ satisfying $P_\alpha = \bs 0$, this row can then be removed from the parity table $P$ which leads to a smaller instance of the problem.

\section{Benchmarks}\label{sec:bench}

In this section we evaluate the performances of the algorithms we presented, and we compare them to state-of-the-art alternatives.
The benchmarks are performed on a set of circuits commonly used to compare the performances of $T$-count optimizers, which were obtained from References~\cite{amyGithub} and~\cite{reversibleBenchmarks}.
We also include a circuit implementing the block cipher DEFAULT, as given in Reference~\cite{default}.

In Subsection~\ref{subsec:bench_tohpe_todd}, we benchmark the $T$-count optimizers presented in Section~\ref{sec:beyond_merging}.
We consider the case where ancillary qubits can be used, as well as the case where no ancillary qubits are used.
Then, in Subsection~\ref{sub:bench_gf}, we compare the performance of our algorithms with state-of-the-art methods for the optimization of multiplication circuits in finite fields.

\subsection{\texttt{TOHPE} and \texttt{FastTODD} algorithms}\label{subsec:bench_tohpe_todd}

In this subsection we compare the performances of the \texttt{TOHPE} procedure (Algorithm~\ref{alg:tohpe}), the \texttt{FastTODD} procedure (Algorithm~\ref{alg:fast_todd}), the $\texttt{PHAGE}$ procedure of Reference~\cite{de2020fast}, the \texttt{TODD} procedure of Reference~\cite{heyfron2018efficient} and the deep reinforcement learning method of Reference~\cite{ruiz2024quantum} named AlphaTensor-Quantum.
The \texttt{PHAGE} procedure was implemented in Haskell by the authors of the method~\cite{stompCode}, while the \texttt{TOHPE}, \texttt{FastTODD} and \texttt{TODD} procedures were implemented in the Rust programming language.
Our implementations of the \texttt{TOHPE} and \texttt{FastTODD} procedures used for the benchmarks are open source~\cite{github}.

The execution times of AlphaTensor-Quantum are not reported in the benchmarks because they were not provided in Reference~\cite{ruiz2024quantum}.
However, due to the nature of the method, we can expect it to have much longer execution times than all the other methods.
Also, some circuits had to be split into several parts for optimizing them with AlphaTensor-Quantum in a reasonable time.

In Reference~\cite{ruiz2024quantum}, it is proposed to optimize the number of $T$ gates by implementing $CCZ$ gates using a gadgetization technique presented in Reference~\cite{jones2013low}.
This gadgetization technique uses a ancillary qubit to implement the $CCZ$ with 4 $T$ gates instead of $7$ $T$ gates in the case where no ancillary qubits are used.
Using this gadgetization technique can therefore lead to a better $T$-count.
However, if the purpose of optimizing the number of $T$ gates is to use fewer resources to implement them using magic states distillation protocols, then this approach can be counterproductive.
Indeed, as shown in Reference~\cite{campbell2017unified}, an implementation of a phase polynomial associated with a third order homogeneous weighted polynomial (as in Equation~\ref{eq:third_order_homogeneous_polynomial}) can benefit from a free round of distillation with quadratic error suppression.
By using an ancillary qubit, the gadgetization technique for implementing the $CCZ$ gate breaks this important property of the weighted polynomial.
Such protocols for implementing and distilling a phase polynomial in a single step are refered to as synthillation protocols.
The simplest example of this is that the $CCZ$ gate can be implemented more efficiently by using the corresponding synthillation protocol rather than by distilling $4$ magic states and then implementing the $CCZ$ gate using the gadgetization technique.
Concretely, for distilling $n$ $CCZ$ gates with quadratic error suppression, the synthillation protocol consumes $\mathcal{O}(6n)$ magic states, whereas using the gadgetization technique would consume $\mathcal{O}(12n)$ magic states by using the Bravyi-Haah protocol (which is the magic state distillation protocol with quadratic error suppression that consumes the fewest number of magic states, up to a small constant)~\cite{bravyi2012magic}.
That is why, in order to maintain consistent metrics and to assign the same cost to each $T$ gate, we won't consider the gadgetization technique of Reference~\cite{jones2013low} for implementing the $CCZ$ gates in the benchmarks.

Also, the only gate count metric considered in the benchmarks will be the $T$-count.
It is important to note that all the methods considered in the benchmarks for optimization of the number of $T$ gates can significantly increase the number of CNOT gates in the circuit.
If necessary, a phase polynomial synthesis algorithm can be used to perform the synthesis of the phase polynomial produced by the $T$-count optimizer with an optimized number of CNOT gates~\cite{amy2018controlled, vandaele2022phase}.

We first consider the case where internal Hadamard gates are gadgetized in in Subsection~\ref{sub:bench_tohpe_todd_ancilla}, and then the case where no ancillary qubits are used in Subsection~\ref{sub:bench_tohpe_todd_no_ancilla}.

\subsubsection{With ancillary qubits}\label{sub:bench_tohpe_todd_ancilla}
\begin{table}[t]
\resizebox{1.0\columnwidth}{!}{
\begin{tabular}{lrrrrrrrrrrrr} 
        \toprule
         & \multicolumn{3}{c}{Pre-optimization} & \texttt{Alpha}~\cite{ruiz2024quantum} & \multicolumn{2}{c}{\texttt{PHAGE}~\cite{de2020fast}} & \multicolumn{2}{c}{\texttt{TOHPE}} & \multicolumn{2}{c}{$\texttt{FastTODD}$} & \multicolumn{2}{c}{\texttt{TODD}~\cite{heyfron2018efficient}} \\
        \cmidrule(lr){2-4} \cmidrule(lr){5-5} \cmidrule(lr){6-7} \cmidrule(lr){8-9} \cmidrule(lr){10-11} \cmidrule{12-13}
    Circuit & $n$ & $h$ & $T$-count & $T$-count & $T$-count & t (s) & $T$-count & t (s) & $T$-count & t (s) & $T$-count & t (s) \\
        \midrule
        Adder$_8$ & 24 & 37 & 173 & 139 & 173 & 25 & \bf 119 & 0 & 119 & 1 & 128 & 5580 \\
        Barenco Tof$_3$ & 5 & 3 & 16 & 13 & 13 & 0 & 13 & 0 & 13 & 0 & 14 & 0 \\
        Barenco Tof$_4$ & 7 & 7 & 28 & 23 & 26 & 1 & 23 & 0 & 23 & 0 & 24 & 0 \\
        Barenco Tof$_5$ & 9 & 11 & 40 & 33 & 39 & 11 & 33 & 0 & 33 & 0 & 34 & 0 \\
        Barenco Tof$_{10}$ & 19 & 31 & 100 & 83 & 100 & 27 & 83 & 0 & 83 & 0 & 84 & 237 \\
        CSLA MUX$_3$ & 15 & 6 & 62 & 39 & 46 & 39 & 39 & 0 & 39 & 0 & 42 & 6 \\
        CSUM MUX$_9$ & 30 & 12 & 84 & 71 & 84 & 28 & 71 & 0 & 71 & 0 & 72 & 87 \\
        DEFAULT  & 640 & 11936 & 39744 & - & - & - & \bf 38638 & - & 38638 & - & 39744 & - \\
        Grover$_5$ & 9 & 68 & 166 & 152 & 166 & 26 & \bf 143 & 0 & 143 & 3 & 144 & 6749 \\
        Ham$_{15}$ (high) & 20 & 331 & 1019 & 773 & 1019 & 37 & 691 & 2 & \bf 643 & - & 1001 & - \\
        Ham$_{15}$ (low) & 17 & 29 & 97 & \bf 73 & 97 & 26 & 77 & 0 & 77 & 0 & 76 & 453 \\
        Ham$_{15}$ (med) & 17 & 54 & 212 & 156 & 212 & 28 & 147 & 0 & \bf 137 & 19 & 148 & 21472 \\
        HWB$_{6}$ & 7 & 20 & 75 & 51 & 68 & 117 & 51 & 0 & 51 & 0 & 51 & 29 \\
        HWB$_{8}$ & 12 & 1103 & 3517 & - & 3517 & 228 & \bf 2763 & 169 & 2763 & - & 3517 & - \\
        HWB$_{10}$ & 16 & 5191 & 15891 & - & 15891 & 10422 & \bf 12845 & 51781 & 12845 & - & 15891 & - \\
        HWB$_{11}$ & 15 & 14441 & 44500 & - & - & - & \bf 42643 & - & 42643 & - & 44500 & - \\
        Mod Adder$_{1024}$ & 28 & 304 & 1011 & 762 & 1010 & 36 & 657 & 2 & \bf 575 & - & 957 & - \\
        Mod Mult$_{55}$ & 9 & 3 & 35 & 17& 21 & 0 & 17 & 0 & 17 & 0 & 18 & 0 \\
        Mod Red$_{21}$ & 11 & 17 & 73 & 51 & 72 & 25 & 51 & 0 & 51 & 0 & 53 & 21 \\
        Mod5$_4$ & 5 & 0 & 8 & 7 & 7 & 0 & 7 & 0 & 7 & 0 & 8 & 0 \\
        QCLA Adder$_{10}$ & 36 & 25 & 162 & 135 & 159 & 26 & 113 & 0 & \bf 109 & 4 & 111 & 4216 \\
        QCLA Com$_7$ & 24 & 18 & 95 & 59 & 91 & 28 & 59 & 0 & 59 & 0 & 60 & 149 \\
        QCLA Mod$_7$ & 26 & 58 & 237 & 199 & 237 & 26 & 167 & 0 & \bf 159 & 41 & 168 & 60415 \\
        QFT$_{4}$ & 5 & 38 & 67 & 53 & 55 & 33 & 53 & 0 & 53 & 0 & 55 & 54 \\
        RC Adder$_6$ & 14 & 10 & 47 & 37 & 44 & 32 & 37 & 0 & 37 & 0 & 38 & 2 \\
        Tof$_3$ & 5 & 2 & 15 & 13 & 13 & 0 & 13 & 0 & 13 & 0 & 14 & 0 \\
        Tof$_4$ & 7 & 4 & 23 & 19 & 20 & 0 & 19 & 0 & 19 & 0 & 20 & 0 \\
        Tof$_5$ & 9 & 6 & 31 & 25 & 28 & 1 & 25 & 0 & 25 & 0 & 26 & 0 \\
        Tof$_{10}$ & 19 & 16 & 71 & 55 & 69 & 26 & 55 & 0 & 55 & 0 & 56 & 24 \\
        VBE Adder$_3$ & 10 & 4 & 24 & 19 & 22 & 1 & 19 & 0 & 19 & 0 & 20 & 0 \\
        \bottomrule
\end{tabular}}
\caption{Comparison of different procedures for the optimization of the number of $T$ gates.
    $n$ represents to the number of input qubits in the circuits.
    $h$ denotes the number of internal Hadamard gates in the input circuits.
    All internal Hadamard gates were gadgetized, resulting in $h$ ancillary qubits.
    The $T$-count after optimization and the execution time is reported for each procedure.
    A blank entry in the execution time indicates that the execution couldn’t be carried out in less than a day.
    In such cases, the reported $T$-count is the one obtained after 24 hours of execution.
    A blank entry for the AlphaTensor-Quantum column indicates that no data were reported by the authors on the corresponding circuit.
    A blank entry in the $T$-count for the \texttt{PHAGE} procedure indicates that the execution was halted because the memory usage exceeded 8 terabytes.}\label{tab:bench_gadgetization}
\end{table}

We compare the performances of the algorithms in the case where all the internal Hadamard gates have been gadgetized, as explained in Section~\ref{sub:t_count_hadamard_free}.
All the circuits were pre-optimized using the \texttt{FastTMerge} procedure of Reference~\cite{vandaele2024optimalnumberparametrizedrotations}, followed by the \texttt{InternalHOpt} procedure of Reference~\cite{vandaele2024optimal} to minimize the number of internal Hadamard gates.
This pre-optimization is important as reducing the number of internal Hadamard gates leads to fewer ancillary qubits when these Hadamard gates are gadgetized, and therefore improve the efficiency of the $T$-count optimizers.

The benchmark results are presented in Table~\ref{tab:bench_gadgetization}.
As expected from the complexity analysis of the algorithms, we can notice that the execution times of the \texttt{FastTODD} procedure are much lower than the execution times of the \texttt{TODD} procedure.
Moreover, the \texttt{FastTODD} procedure is providing the best-known $T$-count on every circuits, except for the ``Ham$_{15}$ (low)" circuit.
While the \texttt{TODD} procedure and the \texttt{FastTODD} procedure are based on the same mechanism for reducing the $T$-count, the modifications that we made in the \texttt{FastTODD} procedure are resulting in significantly better $T$-count reduction for some circuits.
For instance, for the ``Mod Adder$_{1024}$" circuit, the \texttt{FastTODD} procedure generates a circuit with a $T$-count that is $40\%$ lower than the $T$-count of the circuit generated by the \texttt{TODD} procedure.
This $40\%$ improvement remains consistent even if we allow the algorithms to terminate by not stopping them after 24 hours.

The \texttt{TOHPE} procedure is also achieving an equivalent or better $T$-count than the state-of-the-art algorithms for all circuits except the ``Ham$_{15}$ (low)" circuit.
Moreover, on every circuit, the execution times of the \texttt{TOHPE} and \texttt{FastTODD} procedures are lower than the execution times of the other procedure (for the ``HWB$_{8}$" and ``HWB$_{10}$" circuits, the \texttt{PHAGE} procedure have lower execution times but but does not improve the $T$-count at all for these circuits).
This demonstrates that our algorithms are faster than state-of-the-art $T$-count optimizers and are producing circuits with a lower or equivalent $T$-count.

\subsubsection{Without ancillary qubits}\label{sub:bench_tohpe_todd_no_ancilla}
\begin{table}[t]
\resizebox{1.0\columnwidth}{!}{
\begin{tabular}{lrrrrrrrrrrrrrrr} 
        \toprule
        & \multicolumn{2}{c}{Pre-optimization} && \multicolumn{2}{c}{\texttt{PHAGE}~\cite{de2020fast}} && \multicolumn{2}{c}{\texttt{TOHPE}} && \multicolumn{2}{c}{$\texttt{FastTODD}$} && \multicolumn{2}{c}{\texttt{TODD}~\cite{heyfron2018efficient}} \\
        \cmidrule(lr){2-3} \cmidrule(lr){5-6} \cmidrule(lr){8-9} \cmidrule(lr){11-12} \cmidrule(lr){14-15}
        Circuit & $n$ & $T$-count && $T$-count & t (s) && $T$-count & t (s) && $T$-count & t (s) && $T$-count & t (s) \\
        \midrule
        Adder$_8$ & 24 & 173 && 172 & 17 && 173 & 0 && \bf 170 & 0 && 172 & 3 \\
        Barenco Tof$_3$ & 5 & 16 && 16 & 0 && 16 & 0 && 16 & 0 && 16 & 0 \\
        Barenco Tof$_4$ & 7 & 28 && 28 & 0 && 28 & 0 && 28 & 0 && 28 & 0 \\
        Barenco Tof$_5$ & 9 & 40 && 40 & 0 && 40 & 0 && 40 & 0 && 40 & 0 \\
        Barenco Tof$_{10}$ & 19 & 100 && 100 & 3 && 100 & 0 && 100 & 0 && 100 & 0 \\
        CSLA MUX$_3$ & 15 & 62 && 50 & 2 && \bf 49 & 0 && 49 & 0 && 50 & 0 \\
        CSUM MUX$_9$ & 30 & 84 && 84 & 79 && \bf 73 & 0 && 73 & 0 && 76 & 7 \\
        DEFAULT & 640 & 39744 && - & - && \bf 39666 & - && 39666 & - && 39744 & - \\
        Grover$_5$ & 9 & 166 && 166 & 0 && 166 & 0 && 166 & 0 && 166 & 0 \\
        Ham$_{15}$ (high) & 20 & 1019 && 1019 & 5 && 1019 & 0 && 1019 & 0 && 1019 & 1 \\
        Ham$_{15}$ (low) & 17 & 97 && 94 & 0 && 94 & 0 && 94 & 0 && 94 & 0 \\
        Ham$_{15}$ (med) & 17 & 212 && 212 & 3 && 212 & 0 && 212 & 0 && 212 & 1 \\
        HWB$_{6}$ & 7 & 75 && 75 & 0 && 75 & 0 && 75 & 0 && 75 & 0 \\
        HWB$_{8}$ & 12 & 3517 && 3508 & 32 && 3498 & 0 && \bf 3487 & 0 && 3493 & 2 \\
        HWB$_{10}$ & 16 & 15891 && 15858 & 392 && 15741 & 1 && \bf 15693 & 2 && 15698 & 14 \\
        HWB$_{11}$ & 15 & 44500 && 44479 & 1418 && 44336 & 3 && \bf 44188 & 5 && 44191 & 28 \\
        HWB$_{12}$ & 20 & 85611 && 85588 & 8153 && 85362 & 10 && \bf 85233 & 16 && 85239 & 135 \\
        Mod Adder$_{1024}$ & 28 & 1011 && 1010 & 5636 && \bf 1009 & 0 && 1009 & 0 && 1010 & 13 \\
        Mod Mult$_{55}$ & 9 & 35 && 28 & 0 && 28 & 0 && 28 & 0 && 28 & 0 \\
        Mod Red$_{21}$ & 11 & 73 && 73 & 0 && 73 & 0 && 73 & 0 && 73 & 0 \\
        QCLA Adder$_{10}$ & 36 & 162 && 161 & 158 && 161 & 0 && 161 & 0 && 161 & 53 \\
        QCLA Com$_7$ & 24 & 95 && 95 & 16 && 95 & 0 && 95 & 0 && 95 & 2 \\
        QCLA Mod$_7$ & 26 & 237 && 237 & 64 && 237 & 0 && 237 & 0 && 237 & 7 \\
        QFT$_{4}$ & 5 & 67 && 67 & 0 && 67 & 0 && 66 & 0 && 66 & 0 \\
        RC Adder$_6$ & 14 & 47 && 47 & 0 && 47 & 0 && 47 & 0 && 47 & 0 \\
        Tof$_3$ & 5 & 15 && 15 & 0 && 15 & 0 && 15 & 0 && 15 & 0 \\
        Tof$_4$ & 7 & 23 && 23 & 0 && 23 & 0 && 23 & 0 && 23 & 0 \\
        Tof$_5$ & 9 & 31 && 31 & 0 && 31 & 0 && 31 & 0 && 31 & 0 \\
        Tof$_{10}$ & 19 & 71 && 71 & 6 && 71 & 0 && 71 & 0 && 71 & 0 \\
        VBE Adder$_3$ & 10 & 24 && 24 & 0 && 24 & 0 && 24 & 0 && 24 & 0 \\
        \bottomrule
\end{tabular}}
\caption{Comparison of different procedures for the optimization of the number of $T$ gates without using any ancillary qubits.
    The value $n$ corresponds to the number of qubits in the circuits.
    The $T$-count after optimization and the execution time is reported for each procedure.
    A blank entry in the execution time indicates that the execution couldn’t be carried out in less than a day.
    In such cases, the reported $T$-count is the one obtained after 24 hours of execution.
    A blank entry in the $T$-count for the \texttt{PHAGE} procedure indicates that the execution was halted because the memory usage exceeded 8 terabytes.}\label{tab:bench_no_ancilla}
\end{table}

The benchmarks results in the case where the internal Hadamard gates are not gadgetized are presented in Table~\ref{tab:bench_no_ancilla}.
All the circuits were pre-optimized using the \texttt{FastTMerge} procedure of Reference~\cite{vandaele2024optimalnumberparametrizedrotations}.
Then the circuits were divided into Hadamard-free subcircuits interposed with Clifford circuits, as described in Subsection~\ref{sub:t_count_hadamard_free} and using Algorithm~\ref{alg:grouping} to minimize the number of Hadamard-free subcircuits.

Similarly to the case where the internal Hadamard gates are gadgetized, the \texttt{FastTODD} procedure provides the best-known $T$-count for all the circuits.
Also, the \texttt{TOHPE} and \texttt{FastTODD} procedures have the lowest execution times.
However, we can notice that the number of $T$ gates was not reduced at all for multiple quantum circuits, and the $T$-count reduction achieved for the other circuits was not very substantial.
This indicates the importance of the gadgetization of internal Hadamard gates for lowering the number of $T$ gates once the \texttt{FastTMerge} procedure described in Reference~\cite{vandaele2024optimalnumberparametrizedrotations} has been applied.
A possible direction for future research work would be to develop effective strategies for optimizing the number of $T$ gates by exploiting the gadgetization of internal Hadamard gates in cases where the number of ancillary at disposal qubits is limited.

\subsection{Galois field multiplier circuits}\label{sub:bench_gf}
\begin{table}[t]
\centering
\begin{tabular}{lrrrrrrrr} 
        \toprule
         & & \multicolumn{2}{c}{\texttt{Alpha}~\cite{ruiz2024quantum}} & Ref.~\cite{vandaele2025quantum} & \multicolumn{2}{c}{$\texttt{FastTODD}$} & \multicolumn{2}{c}{\makecell{Ref.~\cite{vandaele2025quantum} \\ \& \texttt{FastTODD}}} \\
        \cmidrule(lr){3-4} \cmidrule(lr){5-5} \cmidrule(lr){6-7} \cmidrule(lr){8-9}
        Circuit & $T$-count & $T$-count & $CCZ$ & $CCZ$ & $T$-count & t (s) & $T$-count & t (s) \\
        \midrule
        GF$(2^2)$ Mult & 18 & 17 & 3 & 3 & 18 & 0 & 17 & 0 \\
        GF$(2^3)$ Mult & 45 & 29 & 6 & 6 & 36 & 0 & \bf 23 & 0 \\
        GF$(2^4)$ Mult & 68 & \bf 39 & 9 & 9 & 49 & 0 & 43 & 0 \\
        GF$(2^5)$ Mult & 115 & \bf 59 & \bf 13 & 14 & 81 & 0 & 61 & 0 \\
        GF$(2^6)$ Mult & 150 & \bf 77 & 18 & 18 & 113 & 0 & 83 & 0 \\
        GF$(2^7)$ Mult & 217 & \bf 104 & \bf 22 & 23 & 155 & 1 & 111 & 0 \\
        GF$(2^8)$ Mult & 264 & \bf 123 & 29 & \bf 27 & 205 & 1 & 135 & 0 \\
        GF$(2^9)$ Mult & 351 & \bf 161 & \bf 35 & 38 & 257 & 5 & 185 & 2 \\
        GF$(2^{10})$ Mult & 410 & \bf 196 & 46 & \bf 42 & 315 & 17 & 207 & 3 \\
        GF$(2^{16})$ Mult & 1040 & - & - & \bf 81 & 797 & 320 & \bf 425 & 40 \\
        GF$(2^{32})$ Mult & 4128 & - & - & \bf 243 & 3213 & 49386 & \bf 1255 & 16485\\
        GF$(2^{64})$ Mult & 16448 & - & - & \bf 729 & 12503 & - & \bf 3817 & - \\
        GF$(2^{128})$ Mult & 65664 & - & - & \bf 2187 & 50887 & - & \bf 11545 & - \\
        GF$(2^{256})$ Mult & 262400 & - & - & \bf 6561 & 220442 & - & \bf 34731 & - \\
        GF$(2^{512})$ Mult & 1048576 & - & - & \bf 19683 & 1036270 & - & \bf 91931 & - \\
        \bottomrule
\end{tabular}
\caption{Comparison of different procedures for the optimization of the number of $T$ and $CCZ$ gates in GF($2^n$) multiplier circuits.
    The $CCZ$ columns refers to the number of $CCZ$ (or Toffoli) gates in the optimized circuit.
    A blank entry in the execution time for the \texttt{FastTODD} procedure indicates that the execution couldn’t be carried out in less than a day.
    In such cases, the reported $T$-count is the one obtained after 24 hours of execution.
    A blank entry for the AlphaTensor-Quantum column indicates that no data were reported by the authors on the corresponding circuit.}\label{tab:bench_gf}
\end{table}

In this subsection, we compare the performances of AlphaTensor-Quantum~\cite{ruiz2024quantum}, the \texttt{FastTODD} procedure, and the synthesis algorithm of Reference~\cite{vandaele2025quantum} for optimizing the number of $CCZ$ and $T$ gates in GF($2^n$) multiplier circuits.
The circuits were pre-optimized using the \texttt{FastTMerge} procedure of Reference~\cite{vandaele2024optimalnumberparametrizedrotations}, followed by the \texttt{InternalHOpt} procedure of Reference~\cite{vandaele2024optimal} to minimize the number of internal Hadamard gates.
After this optimization, the number of internal Hadamard gates is reduced to zero, thus eliminating the need for Hadamard gate gadgetization and ancillary qubits.

The benchmarks results are presented in Table~\ref{tab:bench_gf}.
We can notice that AlphaTensor-Quantum finds the best number of $CCZ$ gates in the case where $n$ is equal to $5$, $7$ and $9$.
This demonstrates that, while the method of Reference~\cite{vandaele2025quantum} generates circuits with a subquadratic number of $CCZ$ gates, the optimized circuit may not have an optimal number of $CCZ$ gates (at least in the case where $n$ is not a power of 2).
The major drawback of AlphaTensor-Quantum is its large complexity, which imply that it can only be used on small circuits.
Conversely, the \texttt{FastTODD} algorithms was able to optimize the number of $T$ gates even on the largest circuit, although it did not terminate within 24 hours.

The benchmarks results are a good demonstration of the advantages and disadvantages of the AlphaTensor-Quantum method and the \texttt{FastTODD} procedure.
Moreover, the results indicate that these two algorithms are not really in direct competition but rather represent complementary approaches.
One of the differences between the two methods is that AlphaTensor-Quantum performs a complete resynthesis of the Hadamard-free parts of the quantum circuit.
This implies that AlphaTensor-Quantum cannot be used to improve upon existing quantum circuits that may be already well optimized.
For instance, whether or not the number of the number of $T$ gates in the ``GF($2^n$) Mult" input circuit is well optimized will not change the circuit produced by AlphaTensor-Quantum.
The main advantage of this resynthesis approach of AlphaTensor-Quantum is that it can be used to discover radically different circuits that have much lower number of $T$ gates.
We can clearly see that this is the case for the GF($2^n$) multiplier circuits: the input circuit contains $\mathcal{O}(n^2)$ but AlphaTensor-Quantum finds a circuit with a number of $CCZ$ gates similar to the method of Reference~\cite{vandaele2025quantum} which produces a circuit with $\mathcal{O}(n^{1.58})$ $CCZ$ gates.
This demonstrates that AlphaTensor-Quantum could be used to give valuable insights for the design of efficient synthesis methods.
On the contrary, the \texttt{FastTODD} procedure does not perform well on the ``GF($2^n$) Mult" circuits containing $\mathcal{O}(n^2)$ $CCZ$ gates.
However, the \texttt{FastTODD} procedure can be used to optimize the number of $T$ gates in the circuit produced by the method given in Reference~\cite{vandaele2025quantum}.
The \texttt{FastTODD} procedure could also be used to improve the solution of AlphaTensor-Quantum.
For instance, as shown in Table~\ref{tab:bench_gf}, some of the circuits produced by AlphaTensor-Quantum contain an even number of $T$ gates.
In such cases, the parity table associated with the circuit satisfies all the conditions of Theorem~\ref{thm:tohpe}.
Consequently, applying the \texttt{FastTODD} procedure on these circuits will necessarily reduce the number of $T$ gates.

\section{Extension to higher levels of the Clifford hierarchy}\label{sec:higher_orders}

In this section, we extend the results of Section~\ref{sec:beyond_merging} for the optimization of the number of $R_Z(\pi/2^d)$ gates in $\{\mathrm{CNOT}$, $R_Z(\pi/2^d)$, $R_Z(2\pi/2^d)\}$ circuits, for any non-negative integer $d$.
These gates appear in various quantum algorithms, such as Shor's algorithm~\cite{shor1994algorithms}, and there exists distillation protocols for implementing them fault tolerantly~\cite{duclos2015reducing, campbell2018magic}.
Optimizing the number of $R_Z(\pi/2^d)$ gates may increase the number of $R_Z(2\pi/2^d)$ gates.
However, this is motivated by the fact that the $R_Z(\pi/2^d)$ gate is typically more costly to implement than the $R_Z(2\pi/2^d)$ gate.

We formalize the problem of $R_Z(\pi/2^d)$ gates optimization in $\{\mathrm{CNOT}, R_Z(\pi/2^d), R_Z(2\pi/2^d)\}$ circuits in Subsection~\ref{sub:preliminaries_d}.
Then, in Subsection~\ref{sub:d_tohpe}, we extend Theorem~\ref{thm:tohpe_upper_bound} to give an upper bound on the number of $R_Z(\pi/2^d)$ gates in $\{\mathrm{CNOT}, R_Z(\pi/2^d), R_Z(2\pi/2^d)\}$ circuits, which can be satisfied in polynomial time by generalizing Algorithm~\ref{alg:tohpe}.
Finally, in Subsection~\ref{sub:d_fast_todd}, we show how we can construct algorithms that are similar to Algorithm~\ref{alg:fast_todd} for optimizing the number of $R_Z(\pi/2^d)$ gates by providing a generalization of Theorem~\ref{thm:todd_iff}.

\subsection{Symmetric tensor rank decomposition problem}\label{sub:preliminaries_d}

In this subsection we formalize the problem of $R_Z(\pi/2^d)$ gates optimization in $\{\mathrm{CNOT}$, $R_Z(\pi/2^d)$, $R_Z(2\pi/2^d)\}$ circuits for any non-negative integer $d$, by showing its equivalence to the following symmetric tensor rank decomposition problem.

\begin{problem}[$d$-STR]\label{pb:str}
    Let $\mathcal{A} \in \mathbb{Z}_2^{(n, \ldots, n)}$ be a symmetric tensor of order $d$ such that 
    \begin{equation}
        \mathcal{A}_{\alpha_1, \ldots, \alpha_d} = \mathcal{A}_{\beta_1, \ldots, \beta_d}
    \end{equation}
    for all $\alpha_1, \ldots, \alpha_d$ and $\beta_1, \ldots, \beta_d$ such that the set equality $\{\alpha_i, \ldots, \alpha_d\} = \{\beta_i, \ldots, \beta_d\}$ is satisfied.
    Find a Boolean matrix $P$ of size $n \times m$ such that
    \begin{equation}
        \mathcal{A}_{\alpha_1,\ldots,\alpha_d} = \big\lvert \bigwedge_{i=1}^d P_{\alpha_i} \big\rvert \pmod{2}
    \end{equation}
    for all $\alpha_1, \ldots, \alpha_d$ satisfying $0 \leq \alpha_1 \leq \ldots \leq \alpha_d < n$, with minimal $m$.
\end{problem}

We now prove the equivalence between the problem of minimizing the number of $R_Z(\pi/2^d)$ gates in a $\{\mathrm{CNOT}$, $R_Z(\pi/2^d)$, $R_Z(2\pi/2^d)\}$ circuit for any non-negative integer $d$ and the $(d+1)$-STR problem.
The following theorem was first proven in Reference~\cite{amy2019t} by showing the equivalence between the problem of optimizing $R_Z(\pi/2^d)$ gates and the minimum distance decoding problem for the punctured Reed-Muller code of order $n-d-2$ and length $2^n - 1$, which is equivalent to the $d$-STR problem~\cite{seroussi1983maximum}.
We provide a more straightforward proof of this theorem in Appendix~\cite{vandaele2025quantum}.
\begin{theorem}\label{thm:d_str}
    Let $d$ be a non-negative integer.
    The $R_Z(\pi/2^d)$-count optimization problem over the $\{\mathrm{CNOT}, R_Z(\pi/2^d), R_Z(2\pi/2^d)\}$ gate set and the $(d+1)$-STR problem are equivalent.
\end{theorem}

\subsection{$R_Z(\pi/2^d)$-count upper bound in $\{\mathrm{CNOT}, R_Z(\pi/2^d), R_Z(2\pi/2^d)\}$ circuits}\label{sub:d_tohpe}

In this section we present an upper bound for the number of $R_Z(\pi/2^d)$ gates in a $\{\mathrm{CNOT},$ $R_Z(\pi/2^d),$ $R_Z(2\pi/2^d)\}$ circuit, and we present a method for satisfying this upper bound in polynomial time.
We first prove the following theorem, which is a generalization of Theorem~\ref{thm:tohpe}.

\begin{theorem}\label{thm:d_tohpe}
    Let $d$ be a non-negative integer, let $P$ be a parity table of size $n \times m$ and let $P' = P \oplus \bs z \bs y^T$ where $\bs z$ and  $\bs y$ are vectors of size $n$ and $m$ respectively such that
    \begin{align}
        \lvert \bs y \rvert &\equiv 0 \pmod{2} \label{eq:d_condition_1} \\
        \big\lvert \bigwedge_{i=1}^{d} P_{\alpha_i} \wedge \bs y \big\rvert &\equiv 0 \pmod{2} \label{eq:d_condition_2}
    \end{align}
    for all $0 \leq \alpha_1 \leq \ldots \leq \alpha_d < n$.
    Then we have
    \begin{equation}
        \big\lvert \bigwedge_{i=1}^{d+1} P'_{\alpha_i} \big\rvert \equiv \big\lvert \bigwedge_{i=1}^{d+1} P_{\alpha_i} \big\rvert \pmod{2}
    \end{equation}
    for all $0 \leq \alpha_1 \leq \ldots \leq \alpha_{d+1} < n$.
\end{theorem}

\begin{proof}
    For all $\alpha_1, \dots, \alpha_{d+1}$ satisfying $0 \leq \alpha_1 \leq \ldots \leq \alpha_{d+1} < n$, we have:
    \begin{equation}
    \begin{aligned}
        \big\lvert \bigwedge_{i=1}^{d+1} P'_{\alpha_i} \big\rvert &= \big\lvert \bigwedge_{i=1}^{d+1} (P_{\alpha_i} \oplus z_{\alpha_i}\bs y) \big\rvert \\
        &= \big\lvert \bigoplus_{k_{d+1}=0}^{1} \ldots \bigoplus_{k_1=0}^{1} \left[ \bigwedge_{i=1}^{d+1} (k_i P_{\alpha_i} \vee (1 - k_i) z_{\alpha_i} \bs y)\right] \big\rvert \\
        &\equiv \sum_{k_{d+1}=0}^{1} \ldots \sum_{k_1=0}^{1} \big\lvert \left[ \bigwedge_{i=1}^{d+1} (k_i P_{\alpha_i} \vee (1 - k_i) z_{\alpha_i} \bs y) \right] \big\rvert \pmod{2}.\\
    \end{aligned}
    \end{equation}
    The expression 
    \begin{equation}
    \begin{aligned}\label{eq:d_expression_1}
        \big\lvert \bigwedge_{i=1}^{d+1} (k_i P_{\alpha_i} \vee (1 - k_i) z_{\alpha_i} \bs y) \big\rvert \pmod{2} \\
    \end{aligned}
    \end{equation}
    is equal to
    \begin{equation}
    \begin{aligned}\label{eq:d_expression_2}
        \big\lvert \bigwedge_{i=1}^{d+1} z_{\alpha_i} \bs y \big\rvert \pmod{2} \\
    \end{aligned}
    \end{equation}
    in the case where $k_{d+1} = \ldots = k_1 = 0$, which is equal to $0$ because Equation~\ref{eq:d_condition_1} is satisfied.
    Moreover, in the case where $\bs k \in \mathbb{Z}_2^{d+1}$ and the equations $k_{d+1} = \ldots = k_1 = 1$ and $k_{d+1} = \ldots = k_1 = 0$ are not satisfied, Expression~\ref{eq:d_expression_2} is equal to 0 if there exists $i$ such that $k_i = 0$ and $z_{\alpha_i} = 0$.
    Otherwise it is equal to 
    \begin{equation}
    \begin{aligned}\label{eq:d_expression_3}
        \big\lvert \bigwedge_{i=1}^{d} P_{\beta_i} \wedge \bs y \big\rvert \pmod{2} \\
    \end{aligned}
    \end{equation}
    for some $\beta_1, \ldots, \beta_d$ satisfying $0 \leq \beta_1 \leq \ldots \leq \beta_d < n$.
    Expression~\ref{eq:d_expression_3} is equal to Equation~\ref{eq:d_condition_2} and is therefore equal to $0$.
    Thus,
    \begin{equation}
    \begin{aligned}
        \big\lvert \bigwedge_{i=1}^{d+1} (k_i P_{\alpha_i} \vee (1 - k_i) z_{\alpha_i} \bs y) \big\rvert \equiv 0 \pmod{2} \\
    \end{aligned}
    \end{equation}
    for all $\bs k \in \mathbb{Z}_2^{d+1}$ such that $k_i = 0$ for some $i$.
    Finally, in the case where $k_{d+1} = \ldots = k_1 = 1$, Expression~\ref{eq:d_expression_1} is equal to 
    \begin{equation}
    \begin{aligned}\label{eq:d_expression_4}
        \big\lvert \bigwedge_{i=1}^{d+1} P_{\alpha_i} \big\rvert \pmod{2} \\
    \end{aligned}
    \end{equation}
    Therefore,
    \begin{equation}
    \begin{aligned}
        \big\lvert \bigwedge_{i=1}^{d+1} P'_{\alpha_i} \big\rvert &\equiv \sum_{k_{d+1}=0}^{1} \ldots \sum_{k_1=0}^{1} \big\lvert \left[ \bigwedge_{i=1}^{d+1} (k_i P_{\alpha_i} \vee (1 - k_i) z_{\alpha_i} \bs y) \right] \big\rvert \pmod{2} \\
        &\equiv \big\lvert \bigwedge_{i=1}^{d+1} P_{\alpha_i} \big\rvert \pmod{2}.
    \end{aligned}
    \end{equation}\\
\end{proof}

Based on Theorem~\ref{thm:d_tohpe}, we can prove the following subadditivity theorem which is a generalization of Theorem~\ref{thm:subadditivity}.
\begin{theorem}\label{thm:d_subadditivity}
    Let $d$ be a non-negative integer, let $U_1 \in \mathcal{D}^C_{d+1}$, and let $U_2 \in \mathcal{D}_{d+1}$ where $\mathcal{D}_{d+1}$ is the diagonal subgroup of the $(d+1)$th level of the Clifford hierarchy, and $\mathcal{D}_{d+1}^C$ is the subgroup of $\mathcal{D}_{d+1}$ which can be implemented using only $C^{\otimes d}Z$ gates.
    If $\tau[U_1], \tau[U_2] > 0$, then $\tau[U_1 \otimes U_2] < \tau[U_1] + \tau[U_2]$ where $\tau[U]$ is the optimal number of $R_Z(\pi/2^d)$ gates required to implement $U$ without ancillary qubits over the $\{\mathrm{CNOT}$, $R_Z(\pi/2^d)$, $R_Z(2\pi/2^d)\}$ gate set.
\end{theorem}

\begin{proof}
    Let $W = \begin{bmatrix} P & Q \end{bmatrix}$ be a parity table such that $P$ and $Q$ are the parity tables associated with the implementation of $U_1$ and $U_2$ and which have $\tau[U_1]$ and $\tau[U_2]$ columns respectively.
    Then, because $U_1 \in \mathcal{D}^C_{d+1}$, $P$ satisfies
    \begin{equation}
        \big\lvert \bigwedge_{i=1}^d P_{\alpha_i} \big\rvert \equiv 0 \pmod{2}
    \end{equation}
    for all $\alpha_1, \ldots, \alpha_d$ satisfying $0 \leq \alpha_1 < \ldots \leq \alpha_d < n$ where $n$ is the number of qubits on which $U_1$ is acting.
    Let $\bs z = P_{:,i} \oplus Q_{:,j}$ for any $i$ and $j$ satisfying $0\leq i < \tau[U_1]$, $0\leq j < \tau[U_2]$, and let $P'$ be a parity table such that
    $$
    P' = 
    \begin{cases}
        P \oplus \bs z \bs 1^T &\text{if $\tau[U_1] \equiv 0 \pmod{2}$},\\ 
        \begin{bmatrix} P \oplus \bs z \bs 1^T & \bs z \end{bmatrix} &\text{otherwise}.\\ 
    \end{cases}
    $$
    Then, as stated by Theorem~\ref{thm:d_tohpe}, the parity table $P'$ satisfies 
    \begin{equation}
        \big\lvert \bigwedge_{i=1}^{d+1} P'_{\alpha_i} \big\rvert \equiv \big\lvert \bigwedge_{i=1}^{d+1} P_{\alpha_i} \big\rvert \pmod{2}
    \end{equation}
    for all $0 \leq \alpha_1 \leq \ldots \leq \alpha_{d+1} < n$.
    And so the parity table $W' = \begin{bmatrix} P' & Q \end{bmatrix}$, which has at most one more column than $W$, also satisfies 
    \begin{equation}\label{eq:d_w_equality}
        \big\lvert \bigwedge_{i=1}^{d+1} W'_{\alpha_i} \big\rvert \equiv \big\lvert \bigwedge_{i=1}^{d+1} W_{\alpha_i} \big\rvert \pmod{2}
    \end{equation}
    for all $0 \leq \alpha_1 \leq \ldots \leq \alpha_{d+1} < n$.
    However, we can notice that $P'_{:,i} = Q_{:,j}$.
    Therefore, by removing these two columns from $W'$ Equation~\ref{eq:d_w_equality} still holds and $W'$ has at least one less column than $W$.
    The parity table $W'$ implements the unitary $U_1 \otimes U_2$ up to an operator implementable over the $\{\mathrm{CNOT},$ $R_Z(2\pi/2^d)\}$ gate set and has at most $\tau[U_1] + \tau[U_2] - 1$ columns, thus we have $\tau[U_1 \otimes U_2] \leq \tau[U_1] + \tau[U_2] - 1 < \tau[U_1] + \tau[U_2]$.
\end{proof}

Based on this subadditivity theorem and on Theorem~\ref{thm:d_tohpe}, we can derive the following upper bound on the number of $R_Z(\pi/2^d)$ gates in a $\{\mathrm{CNOT}$, $R_Z(\pi/2^d)$, $R_Z(2\pi/2^d)\}$ circuit.

\begin{theorem}\label{thm:d_tohpe_upper_bound}
    Let $d$ be a non-negative integer.
    The number of $R_Z(\pi/2^d)$ gates in an $n$-qubit $\{\mathrm{CNOT}$, $R_Z(\pi/2^d)$, $R_Z(2\pi/2^d)\}$ circuit can be upper bounded by 
    \begin{equation}
        2 \left\lfloor \sum_{i=1}^{d} \frac{{n \choose i}}{2} \right\rfloor + 1 \leq \sum_{i=0}^{d} {n \choose i}
    \end{equation}
    in polynomial time.
\end{theorem}

\begin{proof}
    Let $U$ be a unitary gate implementable over the $\{\mathrm{CNOT}$, $R_Z(\pi/2^d)$, $R_Z(2\pi/2^d)\}$ gate set, let $P$ be a parity table of size $n \times m$ which implements $U$ up to an operator implementable over the $\{\mathrm{CNOT}$, $R_Z(2\pi/2^d)\}$ gate set, and let $L$ be a matrix with rows labelled by $(\alpha_1\ldots\alpha_k)$ such that
    \begin{align}
        L_{\alpha_1\ldots\alpha_k} &= \bigwedge_{i=1}^k P_{\alpha_i}
    \end{align}
    for all $\alpha_1, \ldots, \alpha_k$ satisfying $0 \leq \alpha_i < \ldots < \alpha_k < n$ and $k$ satisfying $1 \leq k \leq d$.
    If $P$ has strictly more than $\sum_{i=0}^{d} {n \choose i}$ columns then we can necessarily find a non-zero vector $\bs y$ satisfying $L \bs y = \bs 0$ and $\bs y \neq \bs 1$ because $L$ has $\sum_{i=1}^{d} {n \choose i}$ rows.
    Note that such vector $\bs y$ necessarily satisfies Equations~\ref{eq:d_condition_2} of Theorem~\ref{thm:d_tohpe}.
    We can then divide $P$ into two non-empty parity tables $P^{(1)}$ and $P^{(2)}$ where the column $P_{:,i}$ belongs to $P^{(1)}$ if and only if $y_i = 1$ and to $P^{(2)}$ otherwise.
    The parity tables $P^{(1)}$ and $P^{(2)}$ are implementations of some unitary gates $U_1 \in \mathcal{D}^C_{d+1}$ and $U_2 \in \mathcal{D}_{d+1}$ respectively.
    The subadditivity theorem (Theorem~\ref{thm:d_subadditivity}) can then be exploited to reduce the number of columns of $P$.
    Let $\bs z = P_{:,i} \oplus P_{:,j}$ where $i$ and $j$ are satisfying $y_i = 1$ and $y_j = 0$, and let 
    $$
    P' = 
    \begin{cases}
        P \oplus \bs z \bs y^T &\text{if $\lvert \bs y \rvert \equiv 0 \pmod{2}$},\\ 
        \begin{bmatrix} P \oplus \bs z \bs y^T & \bs z \end{bmatrix} &\text{otherwise}.\\ 
    \end{cases}
    $$
    By Theorem~\ref{thm:d_tohpe}, $P'$ implements the same unitary gate as $P$ up to an operator implementable over the $\{\mathrm{CNOT}$, $R_Z(2\pi/2^d)\}$ gate set.
    The parity table $P'$ has at most one more column than $P$ and we have $P'_{:,i} = P'_{:,j}$.
    Therefore, the columns $i$ and $j$ can be removed from $P'$, which entails that $P'$ has at least one less column than $P$.
    We showed that if the number of columns of $P$ is strictly greater than $\sum_{i=0}^{d} {n \choose i}$, then the number of columns of $P$ can be reduced by at least one in polynomial time.
    Moreover, if the number of columns of $P$ is equal to $\sum_{i=0}^{d} {n \choose i}$ and is even, then we can necessarily find a non-zero vector $\bs y$ satisfying $L \bs y = \bs 0$ because $L$ has $\sum_{i=1}^{d} {n \choose i}$ rows.
    If there exist $i$ such that $y_i = 0$ then we can reduced the number of columns of $P$ as described above.
    Otherwise we must have $\lvert \bs y \rvert \equiv 0 \pmod{2}$, and so $\bs y$ satisfies the Equations of Theorem~\ref{thm:d_tohpe}.
    Therefore, if $P' = P \oplus \bs z \bs y^T$ where $\bs z$ is equal to the $i$th column of $P$ for any $i$, then $P'$ implements the same unitary gate as $P$ up to an operator implementable over the $\{\mathrm{CNOT}$, $R_Z(2\pi/2^d)\}$ gate set.
    The parity table $P'$ has the same number of columns as $P$ but its $i$th column is equal to the null vector and can therefore be removed, which leads to a parity table containing $\sum_{i=1}^{d} {n \choose i}$ columns.

    The polynomial-time procedure described above to reduce the number of columns of $P$ can be repeated until $P$ has a number of columns lower or equal to
    \begin{equation}
        2 \left\lfloor \sum_{i=1}^{d} \frac{{n \choose i}}{2} \right\rfloor + 1 \leq \sum_{i=0}^{d} {n \choose i}.
    \end{equation}
\end{proof}

The overall complexity of the algorithm described in the proof of Theorem~\ref{thm:d_tohpe_upper_bound} is $\mathcal{O}(n^{2d + 1} m)$ where $m$ is the number of $R_Z(\pi/2^d)$ gates in the initial circuit.
For $d=0$, the algorithm yields a parity table which contains at most one column, which is optimal.
For $d=1$, the algorithm corresponds to Lempel's matrix factorization algorithm~\cite{lempel1975matrix}, which is optimal.
For $d=2$, the algorithm is similar to the \texttt{TOHPE} algorithm (Algorithm~\ref{alg:tohpe}).
Note that the \texttt{TOHPE} algorithm has a complexity of $\mathcal{O}(n^2 m^3)$ instead of $\mathcal{O}(n^5 m)$ because of the heuristic used by the algorithm to better optimize the number of $T$ gates.
The same heuristic can be used in the case where $d>2$, which would lead to an algorithm having a complexity of $\mathcal{O}(n^d m^3)$.

\subsection{Generalization of the \texttt{FastTODD} algorithm for optimizing $R_Z(\pi/2^d)$ gates}\label{sub:d_fast_todd}

In this subsection we generalize the \texttt{FastTODD} algorithm (Algorithm~\ref{alg:fast_todd}) for the optimization of the number of $R_Z(\pi/2^d)$ gates in a $\{\mathrm{CNOT}$, $R_Z(\pi/2^d)$, $R_Z(2\pi/2^d)\}$ circuit.

\begin{theorem}\label{thm:iff_higer_orders}
    Let $d$ be a non-negative integer, let $P$ be a parity table of size $n \times m$ and $P' = P \oplus \bs z \bs y^T$ where $\bs z$ and $\bs y$ are vectors of size $n$ and $m$ respectively and such that $\lvert \bs y \rvert \equiv 0 \pmod{2}$.
    And let $L$ and $X$ be matrices and $\bs v$ be a vector, all with rows labelled by $(\alpha_1 \ldots \alpha_k)$ such that
    \begin{align}
        L_{\alpha_1 \ldots \alpha_k} &= \bigwedge_{i=1}^{k} P_i \\
        X_{\alpha_1 \ldots \alpha_k,\beta_1 \ldots \beta_{k-1}} &= \bigvee_{i=1}^k z_{\alpha_i} \delta_{\alpha_1 \ldots \alpha_{i-1} \alpha_{i+1} \ldots \alpha_k, \beta_1 \ldots \beta_{k-1}} \\
        v_{\alpha_1 \ldots \alpha_k} &= \bigwedge_{i=1}^{k} z_{\alpha_i}
    \end{align}
    for all $\alpha_1, \ldots, \alpha_k$ and $\beta_1, \ldots, \beta_{k-1}$ satisfying $0 \leq \alpha_1 < \ldots < \alpha_k < n$ and $0 \leq \beta_1 < \ldots < \beta_{k-1} < n$ where $k$ satisfies $1 \leq k \leq d$, and where $\delta$ is defined as follows:
    \begin{equation}
        \delta_{\alpha_1 \ldots \alpha_{i-1} \alpha_{i+1} \ldots \alpha_k, \beta_1 \ldots \beta_{k-1}} =
        \begin{cases}
            1 &\text{if $\alpha_1 = \beta_1, \ldots, \alpha_{i-1} = \beta_{i-1}, \alpha_{i+1} = \beta_{i}, \ldots, \alpha_k = \beta_{k-1}$},\\
            0 &\text{otherwise}.
        \end{cases}
    \end{equation}
    All entries of $L, X$ and $\bs v$ are set to zero in the case where $d=0$ for $L$ and in the case where $d\leq 1$ for $X$ and $\bs v$.
    Then, the equality
    \begin{equation}\label{eq:iff_higer_orders_1}
        \big\lvert \bigwedge_{i=1}^{k} P'_{\alpha_i} \big\rvert \equiv \big\lvert \bigwedge_{i=1}^{k} P_{\alpha_i} \big\rvert \pmod{2}
    \end{equation}
    holds for all $0 \leq \alpha_1 < \ldots < \alpha_k < n$ and $1 \leq k \leq d+1$ if and only if there exists a vector $\bs y'$ and a Boolean $b$ such that $L\bs y \oplus X\bs y' \oplus b\bs v = \bs 0$.
\end{theorem}

\begin{proof}
    In the case where $k=1$, we have
    \begin{equation}
    \begin{aligned}
        \lvert P'_{\alpha_1} \rvert &= \lvert P_{\alpha_1} \oplus z_{\alpha_1}\bs y \rvert \\
        &\equiv \lvert P_{\alpha_1} \rvert + \lvert z_{\alpha_1} \bs y \rvert \pmod{2} \\
        &\equiv \lvert P_{\alpha_1} \rvert \pmod{2} \\
    \end{aligned}
    \end{equation}
    because $\lvert \bs y \rvert \equiv 0 \pmod{2}$.
    This proves Theorem~\ref{thm:iff_higer_orders} in the case where $d=0$, because in such case the equation $L\bs y \oplus X \bs y' \oplus b \bs v = \bs 0$ is always satisfied.
    We now prove Theorem~\ref{thm:iff_higer_orders} in the case where $d > 0$ and $\bs z$ is satisfying $\lvert \bs z \rvert = 1$.
    Without loss of generality we will assume that $z_{\alpha_1} = 1$ and $z_{\alpha_i} = 0$ for all $i$ such that $i \neq 1$, then
    \begin{equation}
    \begin{aligned}
        \big\lvert \bigwedge_{i=1}^{k} P'_{\alpha_i} \big\rvert &= \big\lvert (P_{\alpha_1} \oplus z_{\alpha_1}\bs y) \bigwedge_{i=2}^{k} P_{\alpha_i} \big\rvert \\
        &= \big\lvert (P_{\alpha_1} \oplus \bs y) \bigwedge_{i=2}^{k} P_{\alpha_i} \big\rvert \\
        &= \big\lvert \bigwedge_{i=1}^{k} P_{\alpha_i} \oplus \bigwedge_{i=2}^{k} P_{\alpha_i} \wedge \bs y \big\rvert \\
        &\equiv \big\lvert \bigwedge_{i=1}^{k} P_{\alpha_i} \big\rvert + \big\lvert \bigwedge_{i=2}^{k} P_{\alpha_i} \wedge \bs y \big\rvert \pmod{2} \\
    \end{aligned}
    \end{equation}
    where $\alpha_1, \ldots, \alpha_k$ are satisfying $0 \leq \alpha_1 < \ldots < \alpha_k < n$ and $k$ is satisfying $1 \leq k \leq d+1$.
    Therefore, proving Theorem~\ref{thm:iff_higer_orders} in the case where $d > 0$ and $\lvert \bs z \rvert = 1$ can be done by showing that there exists a vector $\bs y'$ and a Boolean $b$ such that $L\bs y \oplus X\bs y' \oplus b \bs v = \bs 0$ if and only if the equation
    \begin{equation}
    \begin{aligned}\label{eq:iff_higer_orders_0}
        \big\lvert \bigwedge_{i=2}^{k} P_{\alpha_i} \wedge \bs y \big\rvert &\equiv 0 \pmod{2} \\
    \end{aligned}
    \end{equation}
    holds for all $\alpha_2, \ldots, \alpha_k$ satisfying $\alpha_1 < \alpha_2 < \ldots < \alpha_{k} < n$ and $k$ satisfying $2 \leq k \leq d+1$.
    By definition we have $X_{\alpha_2\ldots\alpha_{k}} = \bs 0$ and $v_{\alpha_2\ldots\alpha_{k}} = 0$ for all $\alpha_2, \ldots, \alpha_k$ satisfying $0 \leq \alpha_1 < \ldots < \alpha_{k} < n$ and $k$ satisfying $2 \leq k \leq d+1$.
    Hence,
    \begin{equation}
        L_{\alpha_2\ldots\alpha_{k}} \bs y \oplus X_{\alpha_2\ldots\alpha_{k}} \bs y' \oplus b v_{\alpha_2\ldots\alpha_{k}} = L_{\alpha_2\ldots\alpha_{k}} \bs y = \big\lvert \bigwedge_{i=2}^{k} P_{\alpha_i} \wedge \bs y \big\rvert \pmod{2}
    \end{equation}
    and so if $L\bs y \oplus X\bs y' \oplus b \bs v = \bs 0$ then $L_{\alpha_2\ldots\alpha_{k}}\bs y = 0$ which imply that Equation~\ref{eq:iff_higer_orders_0} is satisfied.
    Conversely, if Equation~\ref{eq:iff_higer_orders_0} is satisfied then we must have 
    \begin{equation}
        L_{\alpha_2\ldots\alpha_{k}}\bs y \oplus X_{\alpha_2\ldots\alpha_{k}}\bs y' \oplus b \bs v_{\alpha_2\ldots\alpha_{k}} = L_{\alpha_2\ldots\alpha_{k}} = \bs 0
    \end{equation}
    for all $\alpha_2, \ldots, \alpha_k$ satisfying $\alpha_1 < \alpha_2 < \ldots < \alpha_{k} < n$ and $k$ satisfying $2 \leq k \leq d+1$.
    Moreover we have
    \begin{equation}
        L_{\alpha_1}\bs y \oplus X_{\alpha_1}\bs y' \oplus b v_{\alpha_1} = L_{\alpha_1}\bs y \oplus b
    \end{equation}
    because $X_{\alpha_1} = \bs 0$ and $v_{\alpha_1} = 1$.
    And
    \begin{equation}
        L_{\alpha_1\ldots\alpha_{k}}\bs y \oplus X_{\alpha_1\ldots\alpha_{k}}\bs y' \oplus b v_{\alpha_1\ldots\alpha_{k}} = L_{\alpha_1\ldots\alpha_{k}}\bs y \oplus y'_{\alpha_2 \ldots \alpha_{k}}
    \end{equation}
    for all $\alpha_2, \ldots, \alpha_k$ satisfying $\alpha_1 < \alpha_2 < \ldots < \alpha_{k} < n$ and $k$ satisfying $2 \leq k \leq d$, because $v_{\alpha_1 \ldots \alpha_{k}} = 0$ and
    \begin{equation}\label{eq:iff_higer_orders_2}
        X_{\alpha_1 \ldots \alpha_k}\bs y' = \bigoplus_{i=1}^k z_{\alpha_i} y'_{\alpha_1 \ldots \alpha_{i-1} \alpha_{i+1} \ldots \alpha_{k}}
    \end{equation}
    for all $\alpha_1, \ldots, \alpha_k$ satisfying $0 \leq \alpha_1 < \ldots < \alpha_{k} < n$ and $k$ satisfying $2 \leq k \leq d$ and where the rows of $\bs y'$ are labelled the same way as the columns of $X$.
    Let $\bs y'$ be such that $y'_{\alpha_2 \ldots \alpha_k} = L_{\alpha_1\ldots\alpha_{k}}\bs y$ for all $\alpha_2, \ldots, \alpha_k$ satisfying $\alpha_1 < \alpha_2 < \ldots < \alpha_{k} < n$ and $k$ satisfying $2 \leq k \leq d$ and $b$ be such that $b = L_{\alpha_1}\bs y$, then we have $L\bs y \oplus X\bs y' \oplus b\bs v = \bs 0$.
    Thus, we proved that
    \begin{equation}
        L\bs y \oplus X\bs y' \oplus b\bs v = \bs 0 \iff \big\lvert \bigwedge_{i=2}^{k} P_{\alpha_i} \wedge \bs y \big \rvert \equiv 0 \pmod{2}
    \end{equation}
    for all $\alpha_2, \ldots, \alpha_k$ satisfying $\alpha_1 < \alpha_2 < \ldots < \alpha_{k} < n$ and $k$ satisfying $2 \leq k \leq d+1$, and so Theorem~\ref{thm:iff_higer_orders} is true in the case where $\lvert \bs z \rvert = 1$.

    Let $B$ be a full rank binary matrix of size $n\times n$ which represents a change of basis, and let $\tilde{L}, \tilde{X}$ and $\tilde{\bs v}$ be constructed in the same way as $L, X$ and $\bs v$ but with respect to $BP$ and $B\bs z$.
    Notice that if Equation~\ref{eq:iff_higer_orders_1} is satisfied for $P$ and $\bs z$ then it is also satisfied for $BP$ and $B\bs z$.
    The proof of Theorem~\ref{thm:iff_higer_orders} can then be completed by proving that if there exists a vector $\bs y'$ and a Boolean $b$ such that $L\bs y \oplus X\bs y' \oplus b \bs v = \bs 0$ then there exists a vector $\tilde{\bs y}'$ and a Boolean $\tilde{b}$ such that $\tilde{L}\bs y \oplus \tilde{X}\tilde{\bs y}' \oplus \tilde{b}\tilde{\bs v} = \bs 0$.
    Suppose this proposition to be true and let $B$ be such that $\lvert B\bs z \rvert = 1$.
    Then, as previously demonstrated, Equation~\ref{eq:iff_higer_orders_1} holds for $BP$ and $B\bs z$ if and only if there exists a vector $\tilde{\bs y}'$ and a Boolean $\tilde{b}$ such that $\tilde{L}\bs y \oplus \tilde{X}\tilde{\bs y}' \oplus \tilde{b}\tilde{\bs v} = \bs 0$, which would imply that there exists a vector $\bs y'$ and a Boolean $b$ such that $L\bs y \oplus X\bs y' \oplus b\bs v = \bs 0$ if and only if Equation~\ref{eq:iff_higer_orders_1} holds for $P$ and $\bs z$.

    Let $\tilde{\bs z}$ and $\tilde{P}$ be such that $\tilde{z}_{\alpha_1} = z_{\alpha_1} \oplus z_{\alpha_2}$ and $\tilde{P}_{\alpha_1} = P_{\alpha_1} \oplus P_{\alpha_2}$ for some fixed $\alpha_1, \alpha_2$ satisfying $\alpha_1 \neq \alpha_2$ and $\tilde{z}_\beta = z_\beta$, $\tilde{P}_\beta = P_\beta$ for all $\beta$ such that $\beta \neq \alpha_1$.
    Without loss of generality, we will assume that $\alpha_1 = 0$ and $\alpha_2 = 1$.
    Indeed, all the other possibilities can be reduced to this case simply by permuting the rows of the parity table $P$ and the vector $\bs z$.
    Let $\tilde{L}$, $\tilde{X}$ and $\tilde{\bs v}$ be constructed in the same way as $L$, $X$ and $\bs v$ but with respect to $\tilde{P}$ and $\tilde{\bs z}$.
    We will now prove that if there exists a vector $\bs y'$ and a Boolean $b$ such that $L\bs y \oplus X\bs y' \oplus b\bs v = \bs 0$ then $\tilde{L}\bs y \oplus \tilde{X} \tilde{\bs y}' \oplus b \tilde{\bs v} = \bs 0$, where $\tilde{\bs y}'$ satisfies $\tilde{y}'_{\alpha_1\alpha_3\ldots\alpha_k} = y'_{\alpha_1\alpha_3\ldots\alpha_k} \oplus y'_{\alpha_2\ldots\alpha_k}$ and $\tilde{y}'_{\alpha_3\ldots\alpha_k} = y'_{\alpha_3\ldots\alpha_k}$ such that $\alpha_i \neq \alpha_1$ for all $i$.
    By using Equation~\ref{eq:iff_higer_orders_2} we can deduce that
    \begin{equation}\label{eq:iff_higer_orders_3}
        \tilde{X}_{\alpha_3\ldots \alpha_k}\tilde{\bs y}' = \bigoplus_{i=3}^{k} \tilde{z}_{\alpha_i} \tilde{y}'_{\alpha_3\ldots\alpha_{i-1}\alpha_{i+1}\ldots\alpha_k} = \bigoplus_{i=3}^{k} z_{\alpha_i} y'_{\alpha_3\ldots\alpha_{i-1}\alpha_{i+1}\ldots\alpha_k} = X_{\alpha_3\ldots \alpha_k}\bs y'
    \end{equation}
    for all $\alpha_3, \ldots, \alpha_k$ satisfying $\alpha_2 < \alpha_3 < \ldots < \alpha_k < n$ and $k$ satisfying $3 \leq k \leq d+2$, and where the rows of $\tilde{y}'$ and $y'$ are labelled in the same as the columns of $\tilde{X}$ and $X$.
    And from the definitions of the vectors $\tilde{\bs v}$ and $\bs v$, we can deduce that
    \begin{equation}\label{eq:iff_higer_orders_4}
        \tilde{v}_{\alpha_3\ldots\alpha_k} = \bigwedge_{i=3}^k \tilde{z}_{\alpha_i} = \bigwedge_{i=3}^k z_{\alpha_i} = v_{\alpha_3\ldots\alpha_k}
    \end{equation}
    for all $\alpha_3, \ldots, \alpha_k$ satisfying $\alpha_2 < \alpha_3 < \ldots < \alpha_k < n$ and $k$ satisfying $3 \leq k \leq d+2$.
    Equations~\ref{eq:iff_higer_orders_3} and~\ref{eq:iff_higer_orders_4} imply that
    \begin{equation}
        \tilde{L}_{\alpha_3\ldots \alpha_k}\bs y \oplus \tilde{X}_{\alpha_3\ldots \alpha_k} \tilde{\bs y}' \oplus b\tilde{v}_{\alpha_3\ldots \alpha_k} = L_{\alpha_3\ldots \alpha_k}\bs y \oplus X_{\alpha_3\ldots \alpha_k}\bs y' \oplus b v_{\alpha_3\ldots \alpha_k} = 0
    \end{equation}
    for all $\alpha_3, \ldots, \alpha_k$ satisfying $\alpha_2 < \alpha_3 < \ldots < \alpha_k < n$ and $k$ satisfying $3 \leq k \leq d+2$.
    Furthermore, we have
    \begin{equation}
    \begin{aligned}
        \tilde{L}_{\alpha_1\ldots\alpha_k}\bs y &\equiv \big\lvert \bigwedge_{i=1}^{k} \tilde{P}_{\alpha_i} \wedge \bs y \big\rvert \pmod{2} \\
        &\equiv \big\lvert (P_{\alpha_1} \oplus P_{\alpha_2}) \wedge \bigwedge_{i=2}^k P_{\alpha_k} \wedge \bs y \big\rvert \pmod{2} \\
        &\equiv \big\lvert \bigwedge_{i=1}^k P_{\alpha_k} \wedge \bs y \big\rvert + \big\lvert \bigwedge_{i=2}^k P_{\alpha_k} \wedge \bs y \big\rvert\pmod{2} \\
        &\equiv L_{\alpha_1\ldots\alpha_k}\bs y + L_{\alpha_2\ldots\alpha_k}\bs y \pmod{2} \\
        &\equiv X_{\alpha_1\ldots\alpha_k}\bs y' + b v_{\alpha_1\ldots\alpha_k} + X_{\alpha_2\ldots\alpha_k}\bs y' + b v_{\alpha_2\ldots\alpha_k} \pmod{2} \\
    \end{aligned}
    \end{equation}
    for all $\alpha_3, \ldots, \alpha_k$ satisfying $\alpha_2 < \alpha_3 < \ldots < \alpha_k < n$ and $k$ satisfying $3 \leq k \leq d$, which entails
{\allowdisplaybreaks
    \begin{align} \notag
        &\tilde{L}_{\alpha_1\ldots\alpha_k}\bs y \oplus \tilde{X}_{\alpha_1\ldots\alpha_k}\tilde{\bs y}' \oplus b \tilde{v}_{\alpha_1\ldots\alpha_k} \\ \notag
        &= X_{\alpha_1\ldots\alpha_k}\bs y' \oplus X_{\alpha_2\ldots\alpha_k}\bs y' \oplus \tilde{X}_{\alpha_1\ldots\alpha_k}\tilde{\bs y}'\oplus b v_{\alpha_1\ldots\alpha_k} \oplus b v_{\alpha_2\ldots\alpha_k} \oplus b \tilde{v}_{\alpha_1\ldots\alpha_k} \\ \notag
        &= \bigoplus_{i=1}^{k} z_{\alpha_i} y'_{\alpha_1\ldots\alpha_{i-1}\alpha_{i+1}\ldots\alpha_k} \oplus \bigoplus_{i=2}^{k} z_{\alpha_i} y'_{\alpha_2\ldots\alpha_{i-1}\alpha_{i+1}\ldots\alpha_k} \oplus \bigoplus_{i=1}^{k} \tilde{z}_{\alpha_i} \tilde{y}'_{\alpha_1\ldots\alpha_{i-1}\alpha_{i+1}\ldots\alpha_k} \\ \notag
        &\quad \oplus b \bigwedge_{i=1}^k z_{\alpha_i} \oplus b \bigwedge_{i=2}^k z_{\alpha_i} \oplus b \bigwedge_{i=1}^k \tilde{z}_{\alpha_i}\\ \notag
        &= \bigoplus_{i=1}^{k} z_{\alpha_i} y'_{\alpha_1\ldots\alpha_{i-1}\alpha_{i+1}\ldots\alpha_k} \oplus \bigoplus_{i=2}^{k} z_{\alpha_i} y'_{\alpha_2\ldots\alpha_{i-1}\alpha_{i+1}\ldots\alpha_k} \oplus (z_{\alpha_1} \oplus z_{\alpha_2})y'_{\alpha_2\ldots\alpha_k} \\ \notag
        &\quad \oplus \bigoplus_{i=2}^{k} z_{\alpha_i} \tilde{y}'_{\alpha_1\ldots\alpha_{i-1}\alpha_{i+1}\ldots\alpha_k} \oplus b \bigwedge_{i=1}^k z_{\alpha_i} \oplus b \bigwedge_{i=2}^k z_{\alpha_i} \oplus b \bigwedge_{i=2}^k z_{\alpha_i} \wedge (z_{\alpha_1} \oplus z_{\alpha_2}) \\
        &= \bigoplus_{i=2}^{k} z_{\alpha_i} (y'_{\alpha_1\ldots\alpha_{i-1}\alpha_{i+1}\ldots\alpha_k} \oplus y'_{\alpha_2\ldots\alpha_{i-1}\alpha_{i+1}\ldots\alpha_k}) \oplus z_{\alpha_2} y'_{\alpha_2\ldots\alpha_k} \\ \notag
        & \quad \oplus \bigoplus_{i=2}^{k} z_{\alpha_i} \tilde{y}'_{\alpha_1\ldots\alpha_{i-1}\alpha_{i+1}\ldots\alpha_k} \oplus b \bigwedge_{i=1}^k z_{\alpha_i} \oplus b \bigwedge_{i=2}^k z_{\alpha_i} \oplus b \bigwedge_{i=1}^k z_{\alpha_i} \oplus b \bigwedge_{i=2}^k z_{\alpha_i} \\ \notag
        &= \bigoplus_{i=2}^{k} z_{\alpha_i} (y'_{\alpha_1\ldots\alpha_{i-1}\alpha_{i+1}\ldots\alpha_k} \oplus y'_{\alpha_2\ldots\alpha_{i-1}\alpha_{i+1}\ldots\alpha_k}) \oplus z_{\alpha_2} y'_{\alpha_2\ldots\alpha_k} \oplus z_{\alpha_2} \tilde{y}'_{\alpha_1\alpha_{3}\ldots\alpha_k} \\ \notag
        &\quad \oplus \bigoplus_{i=3}^{k} z_{\alpha_i} \tilde{y}'_{\alpha_1\ldots\alpha_{i-1}\alpha_{i+1}\ldots\alpha_k} \\ \notag
        &= \bigoplus_{i=2}^{k} z_{\alpha_i} (y'_{\alpha_1\ldots\alpha_{i-1}\alpha_{i+1}\ldots\alpha_k} \oplus y'_{\alpha_2\ldots\alpha_{i-1}\alpha_{i+1}\ldots\alpha_k}) \oplus z_{\alpha_2} y'_{\alpha_2\ldots\alpha_k} \oplus z_{\alpha_2} (y'_{\alpha_1\alpha_{3}\ldots\alpha_k} \oplus y'_{\alpha_2\ldots\alpha_k} \\ \notag
        &\quad \oplus y'_{\alpha_{3}\ldots\alpha_k}) \oplus \bigoplus_{i=3}^{k} z_{\alpha_i} (y'_{\alpha_1\ldots\alpha_{i-1}\alpha_{i+1}\ldots\alpha_k} \oplus y'_{\alpha_2\ldots\alpha_{i-1}\alpha_{i+1}\ldots\alpha_k})\\ \notag
        &= 0.
    \end{align}
}
    Finally, we have
    \begin{equation}
    \begin{aligned}
        \tilde{L}_{\alpha_1}\bs y \oplus \tilde{X}_{\alpha_1}\tilde{\bs y}' \oplus b\tilde{v}_{\alpha_1}
        &= \tilde{L}_{\alpha_1}\bs y \oplus b\tilde{z}_{\alpha_1} \\
        &= \tilde{L}_{\alpha_1}\bs y \oplus b(z_{\alpha_1} \oplus z_{\alpha_2}) \\
        &\equiv \lvert \tilde{P}_{\alpha_1} \wedge \bs y \rvert + bv_{\alpha_1} + bv_{\alpha_2} \pmod{2} \\
        &\equiv \lvert (P_{\alpha_1} \oplus P_{\alpha_2}) \wedge \bs y \rvert + bv_{\alpha_1} + bv_{\alpha_2} \pmod{2} \\
        &\equiv \lvert P_{\alpha_1} \wedge \bs y \rvert + \lvert P_{\alpha_2} \wedge \bs y \rvert + bv_{\alpha_1} + bv_{\alpha_2} \pmod{2} \\
        &= L_{\alpha_1} \oplus bv_{\alpha_1} \oplus L_{\alpha_2} \oplus bv_{\alpha_2} \\
        &= 0.
    \end{aligned}
    \end{equation}
    Thus, we proved that
    \begin{equation}
        L\bs y \oplus X\bs y' \oplus b \bs v = \bs 0 \implies \tilde{L}\bs y \oplus \tilde{X}\tilde{\bs y}' \oplus b \tilde{v} = \bs 0
    \end{equation}
    which concludes the proof of Theorem~\ref{thm:iff_higer_orders}.
\end{proof}

\section{$R_Z(\pi/2^d)$-count upper bound in universal gate sets}\label{sec:k_upper_bound}

In this section we demonstrate an upper bound achievable in polynomial time for the $R_Z(\pi/2^d)$-count within a $\{\mathrm{CNOT},$ $H$, $R_Z(\pi/2^d)$, $R_Z(2\pi/2^d)\}$ circuit where $d$ is a non-negative integer.
Note that it corresponds to an upper bound for the number of $T$ gates in a Clifford$+T$ circuit in the case where $d=2$.
Our proof for the upper bound will partially rest on the following lemma, which has already been proven with a different approach in Reference~\cite{campbell2017unified} for the case where $d=2$.

\begin{lemma}\label{lem:d_upper_bound}
    Let $d$ be a non-negative integer and let $U \in \mathcal{D}_{d+1}$ act on $n$ qubits where $\mathcal{D}_{d+1}$ is the diagonal subgroup of the $(d+1)$th level of the Clifford hierarchy.
    Then $U$ can be decomposed into two unitary gates $U = \tilde{U} \tilde{U}'$ which can be found in polynomial time and such that $\tilde{U} \in \mathcal{D}_{d+1}$ is acting on $n-1$ qubits.
    Furthermore, we can find a $\{\mathrm{CNOT},$ $R_Z(\pi/2^d)$, $R_Z(2\pi/2^d)\}$ circuit implementing $\tilde{U}' \in \mathcal{D}_{d+1}$ with no more than
    \begin{equation}
        \sum_{i=0}^{d-1} {n-1 \choose i} + 1
    \end{equation}
    $R_Z(\pi/2^d)$ gates in polynomial time.
\end{lemma}
\begin{proof}
    Let $P$ be a parity table associated with an implementation of $U$, let $\beta$ be the qubit on which $\tilde{U}$ is not acting, and let $\tilde{P}'$ be a parity table such that
    \begin{equation}\label{eq:eq_1_d_upper_bound}
        \big\lvert \bigwedge_{i=1}^d \tilde{P}'_{\alpha_i} \big\rvert \equiv \big\lvert \bigwedge_{i=1}^d P_{\alpha_i} \wedge P_\beta \big\rvert \pmod{2}
    \end{equation}
    for all $\alpha_1, \ldots, \alpha_d$ satisfying $0 \leq \alpha_1 \leq \ldots \leq \alpha_d < n$ and $\alpha_i \neq \beta$ for all $i$.
    The row $\beta$ of $\tilde{P}'$ can be ignored, and, as stated by Theorem~\ref{thm:d_tohpe_upper_bound}, we can optimize the number of columns of $\tilde{P}'$ such that  it is lower or equal to
    \begin{equation}
        \sum_{i=0}^{d-1} {n-1 \choose i}
    \end{equation}
    and still satisfies Equation~\ref{eq:eq_1_d_upper_bound}.
    If we add a null column to $\tilde{P}'$ in the case where the number of columns of $\tilde{P}'$ is not equal to $\lvert P_\beta \rvert \pmod{2}$, and then set $\tilde{P}'_\beta = \bs 1$, then we have 
    \begin{equation}\label{eq:eq_2_d_upper_bound}
        \big\lvert \bigwedge_{i=1}^d \tilde{P}'_{\alpha_i} \wedge \tilde{P}'_\beta \big\rvert = \big\lvert \bigwedge_{i=1}^d \tilde{P}'_{\alpha_i} \big\rvert \equiv \big\lvert \bigwedge_{i=1}^d P_{\alpha_i} \wedge P_\beta \big\rvert \pmod{2}
    \end{equation}
    for all $0 \leq \alpha_1 \leq \ldots \leq \alpha_d < n$ and the number of columns of $\tilde{P}$ is at most 
    \begin{equation}
        \sum_{i=0}^{d-1} {n-1 \choose i} + 1.
    \end{equation}
    Finally, we can easily find a parity table $\tilde{P}$ such that $\tilde{P}_\beta = \bs 0$ and
    \begin{equation}\label{eq:eq_3_d_upper_bound}
        \big\lvert \bigwedge_{i=1}^{d+1} \tilde{P}_{\alpha_i} \big\rvert \equiv \big\lvert \bigwedge_{i=1}^{d+1} \tilde{P}'_{\alpha_i} \big\rvert + \big\lvert \bigwedge_{i=1}^{d+1} P_{\alpha_i} \big\rvert \pmod{2}
    \end{equation}
    for all $0 \leq \alpha_1 \leq \ldots \leq \alpha_d < n$, which implies its associated unitary gate $\tilde{U}$ doesn't act on qubit $\beta$.
    Thus, $\tilde{U}$ and $\tilde{U}'$ are forming a decomposition of $U$ up to an operator $V$ implementable over $\{\mathrm{CNOT}$, $R_Z(2\pi/2^d)\}$ gate set: $U = \tilde{U}\tilde{U}'V$, where $\tilde{U}'$ is the unitary gate associated with the parity table $\tilde{P}'$.
    The operator $V$ can then be merged with the unitary $\tilde{U}'$ to comply with the decomposition of $U$ established in Lemma~\ref{lem:d_upper_bound}.
\end{proof}

Based on Lemma~\ref{lem:d_upper_bound} and Theorem~\ref{thm:d_tohpe_upper_bound}, we can prove the following theorem which provides an upper bound for the number of $R_Z(\pi/2^d)$ gates in a $\{\mathrm{CNOT}$, $H$, $R_Z(\pi/2^d)$, $R_Z(2\pi/2^d)\}$ circuit.
Note that in the case where $d = 2$, we get an upper bound of $(n+1)(n+2h)/2 + 1$ for the number of $T$ gates in a Clifford$+T$ circuit, which is significantly better than the previously best-known ancilla-free upper bound of $\mathcal{O}(n^2h)$~\cite{amy2019t}, where $n$ is the number of qubits and $h$ is the number of Hadamard gates in the circuit.

\begin{theorem}\label{thm:upper_bound}
    Let $U$ be a unitary gate acting on $n$ qubits and implementable over the $\{\mathrm{CNOT}$, $H$, $R_Z(\pi/2^d)$, $R_Z(2\pi/2^d)\}$ gate set where $d$ is a non-negative integer and with $h$ internal Hadamard gates such that
    \begin{equation}\label{eq:d_u_decomposition}
        U = C_f U_0 \prod_{i=1}^h \left[ H_{\alpha_i} U_i \right] C_{init}
    \end{equation}
    where $C_f, C_{init}$ are implementable over the $\{\mathrm{CNOT}$, $H$, $R_Z(2\pi/2^d)\}$ gate set, $H_{\alpha_i}$ denotes the Hadamard gate applied on some qubit $\alpha_i$, and $U_0,\ldots, U_h$ are implementable over the $\{\mathrm{CNOT}$, $R_Z(\pi/2^d)$, $R_Z(2\pi/2^d)\}$ gate set.
    Then we can find a implementation of $U$ over the $\{\mathrm{CNOT}$, $H$, $R_Z(\pi/2^d)$, $R_Z(2\pi/2^d)\}$ gate set, in polynomial time, and such that the number of $R_Z(\pi/2^d)$ gates is lower or equal to
    \begin{equation}
        \sum_{i=0}^{d} {n \choose i} + h\left(\sum_{i=0}^{d-1} {n-1 \choose i} + 1\right).
    \end{equation}
\end{theorem}

\begin{proof}
    As stated by Lemma~\ref{lem:d_upper_bound}, the unitary $U_h$ can be decomposed into two unitary gates: $U_h = \tilde{U}_h \tilde{U}'_h$, such that $\tilde{U}_h$ is not acting on qubit $\alpha_h$, which imply that $\tilde{U}_h$ commutes with the $h$th internal Hadamard gate.
    Let $U_i \tilde{U}_{i+1} = \tilde{U}_i \tilde{U}'_i$ for all $i$ satisfying $0 < i < h$ where $\tilde{U}_i \tilde{U}'_i$ is a decomposition of $U_i \tilde{U}'_{i+1}$ such as given by Lemma~\ref{lem:d_upper_bound} where $\tilde{U}'_i$ is not acting on qubit $\alpha_i$.
    Then Equation~\ref{eq:d_u_decomposition} is equivalent to 
    \begin{equation}
        U = C_f U_0\tilde{U}_1 \prod_{i=1}^h \left[ H_{\alpha_i} \tilde{U}'_i \right] C_{init}
    \end{equation}
    and, as stated by Lemma~\ref{lem:d_upper_bound}, we can find an implementation of $\tilde{U}'_i$ in polynomial time and with no more than
    \begin{equation}
        \sum_{k=0}^{d-1} {n-1 \choose k} + 1
    \end{equation}
    $R_Z(\pi/2^d)$ gates, for all $i$ satisfying $0 < i \leq h$.
    Moreover, as stated by Theorem~\ref{thm:d_tohpe_upper_bound}, we can find an implementation of $U_0 \tilde{U}_1$ in polynomial time and with no more than
    \begin{equation}
        \sum_{i=0}^{d} {n \choose i}
    \end{equation}
    $R_Z(\pi/2^d)$ gates.
    Thus, $U$ can be implemented in polynomial time over the $\{\mathrm{CNOT}$, $H$, $R_Z(\pi/2^d)$, $R_Z(2\pi/2^d)\}$ gate set and with a number of $R_Z(\pi/2^d)$ gates that is lower or equal to
    \begin{equation}
        \sum_{i=0}^{d} {n \choose i} + h\left(\sum_{i=0}^{d-1} {n-1 \choose i} + 1\right).
    \end{equation}
\end{proof}

\section{Conclusion}
We presented several polynomial-time algorithms for reducing the number of $T$ gates in a Clifford$+T$ circuit.
Benchmarks show that our algorithms consistently achieve the lowest execution times and provide the best $T$-count reduction on almost all the quantum circuits evaluated when compared to state-of-the-art $T$-count optimizers. 
As such, our algorithms not only achieve state-of-the-art $T$-count reduction but also offer much greater scalability, thereby allowing efficient $T$-count optimization on larger quantum circuits.

We proved that our algorithms are producing a circuit in which the number of $T$ gates is upper bounded by $(n^2 + n)/2 + 1$ when they are executed on a Hadamard-free circuit, where $n$ is the number of qubits.
It has been shown that there exists an asymptotic upper bound of $n^2/2 - 1$~\cite{cohen1992covering}.
The question of whether or not there exists a polynomial-time algorithm satisfying this asymptotic upper bound remains open.

We also demonstrated how the number of $T$ gates in a Clifford$+T$ circuit can be optimized such that it is lower or equal to $(n+1)(n+2h)/2 + 1$ where $n$ is the number of qubits and $h$ is the number of internal Hadamard gates in the circuit, without using any ancillary qubits and in polynomial time.
This reinforces the interdependence between the problems of optimizing the number of internal Hadamard gates and optimizing the number of $T$ gates.
As part of future work, it would be beneficial to more clearly determine the roles that internal Hadamard gates and ancillary qubits have in the $T$-count optimization problem. 

\section*{Acknowledgments}
We acknowledge funding from the Plan France 2030 through the projects NISQ2LSQ ANR-22-PETQ-0006 and EPIQ ANR-22-PETQ-007.

\bibliographystyle{quantum}
\bibliography{ref.bib}

\newpage
\appendix

\section{Proof of Theorem~\ref{thm:todd}}\label{app:todd}

\begin{proof}[Proof of Theorem~\ref{thm:todd}]
    Analogously to Equation~\ref{eq:tohpe_proof} and using Equations~\ref{eq:condition_1} and~\ref{eq:condition_2}, we have
    \begin{equation}
    \begin{aligned}
        &\lvert P'_\alpha \wedge P'_\beta \wedge P'_\gamma \rvert \\
        &\equiv \lvert P_\alpha \wedge P_\beta \wedge P_\gamma \rvert + z_\gamma\lvert P_\alpha \wedge P_\beta \wedge \bs y \rvert + z_\beta\lvert P_\alpha \wedge P_\gamma \wedge \bs y \rvert + z_\alpha\lvert P_\beta \wedge P_\gamma \wedge \bs y \rvert \\
          & \quad + z_\beta z_\gamma\lvert P_\alpha \wedge \bs y \rvert + z_\alpha z_\gamma\lvert P_\beta \wedge \bs y \rvert + z_\alpha z_\beta\lvert P_\gamma \wedge \bs y \rvert + z_\alpha z_\beta z_\gamma\lvert \bs y \rvert \pmod{2} \\
          &\equiv \lvert P_\alpha \wedge P_\beta \wedge P_\gamma \rvert + \lvert \left[ z_\alpha(P_\beta \wedge P_\gamma) \oplus z_\beta(P_\alpha \wedge P_\gamma) \oplus z_\gamma(P_\alpha \wedge P_\beta)\right] \wedge \bs y \rvert \pmod{2}
    \end{aligned}
    \end{equation}
    for all $\alpha, \beta, \gamma$ satisfying $0 \leq \alpha \leq \beta \leq \gamma < n$.
    In the case where $\alpha \neq \beta \neq \gamma$, we obtain 
    \begin{equation}
        \lvert P'_\alpha \wedge P'_\beta \wedge P'_\gamma \rvert \equiv \lvert P_\alpha \wedge P_\beta \wedge P_\gamma \rvert \pmod{2} \\
    \end{equation}
    by using Equation~\ref{eq:condition_3}.
    In the case where $\alpha = \beta$, we obtain
    \begin{equation}
    \begin{aligned}
        \lvert P'_\alpha \wedge P'_\beta \wedge P'_\gamma \rvert &\equiv \lvert P_\alpha \wedge P_\beta \wedge P_\gamma \rvert + \lvert \left[ z_\alpha(P_\alpha \wedge P_\gamma) \oplus z_\alpha(P_\alpha \wedge P_\gamma) \oplus z_\gamma P_\alpha \right] \wedge \bs y \rvert \pmod{2} \\
        &\equiv \lvert P_\alpha \wedge P_\beta \wedge P_\gamma \rvert \pmod{2} \\
    \end{aligned}
    \end{equation}
    by using Equation~\ref{eq:condition_2}.
    Finally, in the case where $\alpha = \beta = \gamma$, we obtain
    \begin{equation}
    \begin{aligned}
        \lvert P'_\alpha \wedge P'_\beta \wedge P'_\gamma \rvert &\equiv \lvert P_\alpha \wedge P_\beta \wedge P_\gamma \rvert + \lvert\left[z_\alpha P_\alpha \oplus z_\alpha P_\alpha \oplus z_\alpha P_\alpha \right]\wedge \bs y \rvert \pmod{2} \\
        &\equiv \lvert P_\alpha \wedge P_\beta \wedge P_\gamma \rvert \pmod{2} \\
    \end{aligned}
    \end{equation}
    by using Equation~\ref{eq:condition_2}.
\end{proof}

\section{Proof of Theorem~\ref{thm:d_str}}\label{app:d_str}
\begin{proof}[Proof of Theorem~\ref{thm:d_str}]
    The $\mathrm{CNOT}$ and $R_Z(\theta)$ gates are acting as follows on basis states $\lvert x_1, x_2 \rangle$, $\lvert x_1 \rangle$:
    \begin{align}
        \mathrm{CNOT}\lvert x_1, x_2 \rangle &= \lvert x_1, x_1 \oplus x_2 \rangle \\
        R_Z(\theta)\lvert x_1 \rangle &= e^{i\theta x_1} \lvert x_1 \rangle
    \end{align}
    Therefore, the action of an $n$-qubits $\{\mathrm{CNOT}, R_Z(\pi/2^d), R_Z(2\pi/2^d)\}$ circuit on a basis state $\lvert \bs x \rangle$ has the form
    \begin{equation}
        \lvert \bs x\rangle \mapsto e^{i\frac{\pi}{2^d} p(\bs x)} \lvert g(\bf \bs x) \rangle
    \end{equation}
    where $g: \mathbb{Z}_2^n \rightarrow \mathbb{Z}_2^n$ is a linear reversible Boolean function which can be implemented using only CNOT gates, and $p$ is a linear combination of linear Boolean functions:
    \begin{equation}
        p(\bs x) = \sum_{j = 1}^{m} a_j (y^{(j)}_1x_1 \oplus \ldots \oplus y^{(j)}_n x_n) \pmod{2^{d+1}}
    \end{equation}
    where $\bs y^{(j)} \in \mathbb{Z}_2^n\setminus \{\bs 0\}$, $\bs a \in \mathbb{Z}_{2\ell}^n$ and $m \geq 0$.
    The function $p$ is called a phase polynomial, we will refer to the Boolean vectors $\bs y^{(j)}$ as the parities of the phase polynomial $p$ and we will refer to $\bs a$ as the weights of the phase polynomial.
    Notice that if a weight $a_j$ associated to a parity $\bs y^{(j)}$ is odd, then the associated rotation can be implemented using only the $R_Z(2\pi/2^d)$ gate.
    Therefore, the number of $R_Z(\pi/2^d)$ gates required to implement $p$ is equal to $\lvert \bs a \pmod{2} \rvert$.
    That is why the problem of minimizing the number of $R_Z(\pi/2^d)$ gates in a $\{\mathrm{CNOT}$, $R_Z(\pi/2^d)$, $R_Z(2\pi/2^d)\}$ circuit consists in finding a phase polynomial $p'$ equivalent to $p$ but with weights $\bs a'$ such that $\lvert \bs a' \pmod{2} \rvert$ is minimal.
    We now prove the following equality, which will be useful to characterize the set of phase polynomials that are equivalent:
    \begin{equation}\label{eq:pp_wp_eq}
        x_1 \oplus \ldots \oplus x_n = \sum_{k = 1}^{n} \sum_{\alpha_{1} < \ldots < \alpha_k} (-2)^{k-1} \prod_{i=1}^k x_{\alpha_i}
    \end{equation}
    for all $n \geq 1$ and where $\bs x \in \mathbb{Z}_2^n$.
    This equality is trivially true for $n=1$, and for $n=2$ we can easily verify by case distinction that
    \begin{equation}
        x_1 \oplus x_2 = x_1 + x_2 - 2 x_1 x_2.
    \end{equation}
    Let assume that the equality is true for $n$ and let $\tilde{x}_1 = x_1 \oplus x_{n+1}$, then we have
    \begin{equation}
    \begin{aligned}
        x_1 \oplus \ldots \oplus x_{n+1} &= \tilde{x}_1 \oplus \ldots \oplus x_n \\
         &= \tilde{x}_1\left(1 + \sum_{k = 2}^{n} \sum_{1 < \alpha_{1} < \ldots < \alpha_{k-1}} (-2)^{k-1} \prod_{i=2}^k x_{\alpha_i}\right) \\
         &\quad + \sum_{k = 1}^{n} \sum_{1 < \alpha_{1} < \ldots < \alpha_k} (-2)^{k-1} \prod_{i=1}^k x_{\alpha_i} \\
         &= (x_1 \oplus x_{n+1})\left( 1 + \sum_{k = 2}^{n} \sum_{1 < \alpha_{1} < \ldots < \alpha_{k-1}} (-2)^{k-1} \prod_{i=2}^k x_{\alpha_i}\right) \\
         &\quad + \sum_{k = 1}^{n} \sum_{1 < \alpha_{1} < \ldots < \alpha_k} (-2)^{k-1} \prod_{i=1}^k x_{\alpha_i} \\
         &= (x_1 + x_{n+1} - 2x_1x_{n+1}) \left(1 + \sum_{k = 2}^{n} \sum_{1 < \alpha_{1} < \ldots < \alpha_{k-1}} (-2)^{k-1} \prod_{i=2}^k x_{\alpha_i}\right) \\
         &\quad+ \sum_{k = 1}^{n} \sum_{1 < \alpha_{1} < \ldots < \alpha_k} (-2)^{k-1} \prod_{i=1}^k x_{\alpha_i} \\
         &= \sum_{k = 1}^{n+1} \sum_{\alpha_{1} < \ldots < \alpha_k} (-2)^{k-1} \prod_{i=1}^k x_{\alpha_i} \\
    \end{aligned}
    \end{equation}
    Moreover, the equality holds under the modulo of any even number, and so
    \begin{equation}
    \begin{aligned}
        x_1 \oplus \ldots \oplus x_n &= \sum_{k = 1}^{n} \sum_{\alpha_{1} < \ldots < \alpha_k} (-2)^{k-1} \prod_{i=1}^k x_{\alpha_i} \pmod{2^{d+1}} \\
                                     &= \sum_{k = 1}^{d+1} \sum_{\alpha_{1} < \ldots < \alpha_k} (-2)^{k-1} \prod_{i=1}^k x_{\alpha_i} \pmod{2^{d+1}}
    \end{aligned}
    \end{equation}
    where $d$ is a non-negative integer satisfying $d+1 \leq n$.
    Then we have
    \begin{equation}
    \begin{aligned}
        p(\bs x) &= \sum_{j = 1}^{m} a_j (y^{(j)}_1x_1 \oplus \ldots \oplus y^{(j)}_n x_n) \pmod{2^{d+1}} \\
                 &= \sum_{j = 1}^{m} a_j \sum_{k = 1}^{d+1} \sum_{\alpha_{1} < \ldots < \alpha_k} (-2)^{k-1} \prod_{i=1}^k y_{\alpha_i}^{(j)} x_{\alpha_i} \pmod{2^{d+1}} \\
                 &= \sum_{k = 1}^d \sum_{\alpha_{1} < \ldots < \alpha_k} (-2)^{k-1} c_{\alpha_1,\ldots,\alpha_k} \prod_{i=1}^k x_{\alpha_i} \pmod{2^{d+1}}
    \end{aligned}
    \end{equation}
    where $c_{\alpha_1,\ldots,\alpha_k} = \sum_{j=1}^{m} a_{j} \prod_{i=1}^k y_{\alpha_1}^{(j)} \pmod{2^{d-k+1}}$.
    Notice that two phase polynomials $p$ and $p'$ with parities and weights $\bs y^{(j)}$, $\bs a$ and $\bs y^{\prime(j)}$, $\bs a'$ respectively are equal if and only if
    \begin{equation}
    \begin{aligned}
        c_{\alpha_1,\ldots,\alpha_k} &= \sum_{j=1}^{m} a_{j} \prod_{i=1}^k y_{\alpha_1}^{(j)} \pmod{2^{d-k+1}} \\
                                 &= \sum_{j=1}^{m'} a'_{j} \prod_{i=1}^k y_{\alpha_1}^{\prime(j)} \pmod{2^{d-k+1}} \\
                                 &= c'_{\alpha_1,\ldots,\alpha_k}
    \end{aligned}
    \end{equation}
    for all $\alpha_1, \ldots, \alpha_k$ satisfying $\alpha_1 < \ldots < \alpha_k$.
    Moreover, $p$ and $p'$ are equivalent up to an operator implementable over the $\{\mathrm{CNOT},$ $R_Z(2\pi/2^d)\}$ gate set if and only if
    \begin{equation}
    \begin{aligned}
        c_{\alpha_1,\ldots,\alpha_k} &\equiv c'_{\alpha_1,\ldots,\alpha_k} \pmod{2}
    \end{aligned}
    \end{equation}
    for all $\alpha_1, \ldots, \alpha_k$ satisfying $\alpha_1 < \ldots < \alpha_k$.
    Let $P$ and $P'$ be the parity tables (constructed as explained in Section~\ref{sec:preliminaries}) associated with $p$ and $p'$ and such that the columns of $P$ and $P'$ are encoding the parities $\bs y^{(j)}, \bs y^{\prime(j)}$ associated with an odd weight $a_j$ or $a'_j$.
    Then we have
    \begin{equation}
    \begin{aligned}
        c_{\alpha_1,\ldots,\alpha_k} &\equiv\sum_{j=1}^{m} a_{j} \prod_{i=1}^k y_{\alpha_1}^{(j)} \pmod{2} \\
                                     &\equiv \big\lvert \bigwedge_{i=1}^{k} P_{\alpha_i} \big\rvert \pmod{2}
    \end{aligned}
    \end{equation}
    for all $\alpha_1, \ldots, \alpha_k$ satisfying $\alpha_1 < \ldots < \alpha_k$ and $k$ satisfying $1 \leq k \leq d+1$.
    And so $p$ and $p'$ are equivalent up to an operator implementable over the $\{\mathrm{CNOT},$ $R_Z(2\pi/2^d)\}$ gate set if and only if
    \begin{equation}
    \begin{aligned}
        \big\lvert \bigwedge_{i=1}^{k} P_{\alpha_i} \big\rvert \equiv \big\lvert \bigwedge_{i=1}^{k} P'_{\alpha_i} \big\rvert \pmod{2}
    \end{aligned}
    \end{equation}
    for all $\alpha_1, \ldots, \alpha_k$ satisfying $\alpha_1 < \ldots < \alpha_k$ and $k$ satisfying $1 \leq k \leq d+1$.
    Let $\mathcal{A} \in \mathbb{Z}_2^{(n, \ldots, n)}$ be a symmetric tensor of order $d+1$ such that 
    \begin{equation}
        \mathcal{A}_{\alpha_1, \ldots, \alpha_{d+1}} = \big\lvert \bigwedge_{i=1}^{d+1} P_{\alpha_i} \big\rvert \pmod{2}
    \end{equation}
    for all $\alpha_1, \ldots, \alpha_{d+1}$ satisfying $0 \leq \alpha_1 \leq \ldots \leq \alpha_{d+1} < n$.
    Then $p'$ is equivalent to $p$ up to an operator implementable over the $\{\mathrm{CNOT},$ $R_Z(2\pi/2^d)\}$ gate set if and only if
    \begin{equation}\label{eq:tensor_parity_table}
        \mathcal{A}_{\alpha_1,\ldots,\alpha_{d+1}} = \big\lvert \bigwedge_{i=1}^{k} P'_{\alpha_i} \big\rvert \pmod{2}
    \end{equation}
    for all $\alpha_1, \ldots, \alpha_{d+1}$ satisfying $0 \leq \alpha_1 \leq \ldots \leq \alpha_{d+1} < n$.
    The number of $R_Z(\pi/2^d)$ gates required to implement $p'$ over the $\{\mathrm{CNOT},$ $R_Z(\pi/2^d),$ $R_Z(2\pi/2^d)\}$ gate set is equal to the number of column of $P'$.
    Thus, minimizing the number of $R_Z(\pi/2^d)$ gates consists in finding a parity table $P'$ satisfying Equation~\ref{eq:tensor_parity_table} with a minimal number of columns, which correponds to the $(d+1)$-STR problem.
\end{proof}

\end{document}